\documentclass[11pt]{article}
\usepackage{amssymb}
\usepackage{colortbl}
\usepackage{amsfonts,amsmath, longtable}
\usepackage{ulem}
\usepackage{comment}
\usepackage{enumerate}

\usepackage{hyperref}

        \def\theequation{\thesection.\arabic{equation}}
        \numberwithin{equation}{section}

\newcommand{\tr}{{\rm tr}}
\newcommand{\ti}[1]{\tilde{#1}}

\newcommand{\mL}{{\mathcal L}}
\newcommand{\mM}{{\mathcal M}}

\newcommand{\mH}{{\mathcal H}}
\newcommand{\mU}{{\mathcal U}}
\newcommand{\mO}{{\mathcal O}}
\newcommand{\mV}{{\mathcal V}}

\newcommand{\la}{\lambda}
\newcommand{\vf}{\varphi}
\newcommand{\al}{\alpha}
\newcommand{\be}{\beta}
\newcommand{\ga}{\gamma}
\newcommand{\om}{\omega}
\newcommand{\vth}{\vartheta}
\newcommand{\de}{\delta}
\newcommand{\si}{\sigma}

\newcommand{\Mat}{ {\rm Mat}(N,\mathbb C) }

\newcommand{\mC}{\mathbb C}
\newcommand{\mZ}{\mathbb Z}

\newcommand{\mS}{\mathcal S}
\newcommand{\mA}{\mathcal A}
\newcommand{\mB}{\mathcal B}
\newcommand{\mT}{\mathcal T}
\newcommand{\mD}{\mathcal D}

\newcommand{\ox}{\otimes}

\renewcommand{\b}{\mathbf}
\newcommand{\rank}{\mathop{\rm rank}}

\newcommand{\ka}{\kappa}

\newcommand{\Om}{\Omega}

\newtheorem{corollary}{Corollary}[section]

\newtheorem{predl}{Proposition}[section]

\newenvironment{proof}{\par\noindent{\bf Proof.}}{\hfill$\scriptstyle\blacksquare$}

\def\beq{\begin{equation}}
\def\eq{\end{equation}}
\def\p{\partial}

\newtheorem{theor}{Theorem}

\newcommand{\mat}[4]{\left(\begin{array}{cc}{#1}&{#2}\\ \ \\{#3}&{#4}
\end{array}\right)}

\newcommand{\mats}[4]{\left(\begin{array}{cc}{#1}&{#2}\\ {#3}&{#4}
\end{array}\right)}

\def\res{\mathop{\hbox{Res}}\limits}

\begin{document}

\begin{center}

\setcounter{page}{1}

\vspace{0mm}

{\Large{\bf Spin generalizations of}}

\vspace{3mm}

{\Large{\bf the classical elliptic Calogero--Inozemtsev model}}

\vspace{3mm}

{\Large{\bf and integrable models of interacting tops}}

 \vspace{3mm}

{\Large{\bf associated with the $\rm{BC}_n$ root system}}

 \vspace{12mm}

 {\Large {D. Domanevsky}}$\,^{\bullet\,\diamond}$
\quad\quad
 {\Large {A. Zotov}}$\,^{\diamond}$

 \vspace{8mm}


$\bullet$ --
{\it Institute for Theoretical and Mathematical Physics,\\
Lomonosov Moscow State University, Moscow, 119991, Russia}

$\diamond$ --
 {\it Steklov Mathematical Institute of the Russian
Academy of Sciences,\\ 8 Gubkina St., Moscow 119991, Russia}


 {\small\rm {e-mails: danildom09@gmail.com, zotov@mi-ras.ru}}

\end{center}

\vspace{0mm}

\begin{abstract}
We construct Lax pairs with a spectral parameter for a family of models of
interacting tops associated with the $\rm{BC}_n$ root system. The construction
uses the parametric $\rm{BC}_n$ associative Yang--Baxter equations, which involve
both quantum $R$-matrices and boundary $K$-matrices satisfying the reflection
equation. Special cases yield spin generalizations of the elliptic
Calogero--Inozemtsev model. We also show how quantum Hamiltonians of interacting
tops, combined with the Polychronakos freezing prescription, reproduce recently
introduced anisotropic quantum long-range spin chains of
Haldane--Shastry--Inozemtsev type.
\end{abstract}

\newpage
{\small{ \tableofcontents }}



\subsection*{\ \ \ \ \ \ \underline{Main notation:}}

\ \vspace{-2mm}

{\small{

QYBE -- quantum Yang--Baxter equation (\ref{t2011});

AYBE -- associative Yang--Baxter equation (\ref{t208});

${\rm BC}_n$ AYBE -- equations (\ref{t206})-(\ref{t212});

$n$ -- number of particles in many-body systems or spins in spin chains;

$N$ -- dimension of the matrix representation of the $R$- and $K$-matrices
($R\in\Mat^{\otimes 2}$, $K\in\Mat$ );

$z$ -- spectral parameter in the Lax equation; it plays the role of the
Planck constant in the $R$-matrix;

$\mL(z),\mM(z)$ (or $L(z),M(z)$ for $N=1$) -- Lax pair;

$\mS$ -- $2Nn\times 2Nn$ (or $Nn\times Nn$) matrix of classical spin variables;

$S$ -- $2n\times 2n$ (or $n\times n$)  matrix of classical spin variables: $S=\mS|_{N=1}$;

$\mA$ and $\mB$ -- $Nn\times Nn$ blocks of $\mS$ (\ref{t4011});

$A$ and $B$ -- $n\times n$ blocks of $S$ (\ref{t5011});

$\mO^{\hbox{\tiny{$N$}}}$ -- a coadjoint orbit of the Lie group ${\rm GL}_{N}$;

${\mT}^\pm$ -- spin variables $(\mA\pm\mB)/2$ (\ref{t127}) on the orbits $\mO^{\hbox{\tiny{$nN\!,\!\pm$}}}$;

${T}^\pm$ -- spin variables $(A\pm B)/2$ (\ref{t5151}) on the orbits $\mO^{\hbox{\tiny{$n\!,\!\pm$}}}$;

$R_{12}^z(q)$ -- quantum $R$-matrix in the fundamental representation of
the Lie group ${\rm GL}_N$, satisfying the QYBE (\ref{t2011}) and the
parametric ${\rm BC}_n$-type AYBEs (\ref{t206})--(\ref{t212});

$P_{12}$ -- permutation matrix (\ref{t216});

${\ti R}_{12}^z(q)$ -- solution of the parametric ${\rm BC}_n$-type AYBEs
(\ref{t206})--(\ref{t212}) and the QYBE-type equation (\ref{t2012}).
In the elliptic case, it is related to $R_{12}^\hbar(z)$ by (\ref{t248});

${\ti K}^z(q)$ -- $K$-matrix satisfying the reflection equation (\ref{t201})
and the parametric ${\rm BC}_n$-type AYBEs (\ref{t206})--(\ref{t212});

$k_N(z,q|\nu)$, $u_N(q|\nu)$, $w_N(z|\nu)$ -- scalar functions in the
unitarity relation (\ref{t220}) for ${\ti K}^z(q)$;
see also (\ref{t256}) and (\ref{t2533});

$\varsigma_N$ -- constant in the expansion (\ref{t217}) of the $K$-matrix;

$w_i$ -- the functions (\ref{t5155});

$J$ -- the involutive matrix (\ref{t246});

$F_{ij}^z(q)$ and ${\ti F}_{ij}^z(q)$ -- derivatives of $R_{ij}^z(q)$ and ${\ti R}_{ij}^z(q)$ with respect to $q$;

${\ti Y}_j^z(q)$ -- derivative of ${\ti K}_j^z(q)$  with respect to $q$;

$r_{12}(q)$ -- classical $r$-matrix appearing in the quasiclassical limit (\ref{t202});

$m_{12}(q)$ -- second nontrivial coefficient in the quasiclassical expansion
(\ref{t202});

${\ti r}_{12}(q)$, ${\ti m}_{12}(q)$ -- coefficients of the expansion (\ref{t203}) of ${\ti R}_{12}^z(q)$;

$R^{z,(0)}_{12}$, $R^{z,(1)}_{12}$,
${\ti R}^{(-1)}_{12}$, ${\ti R}^{z,(0)}_{12}$, ${\ti R}^{z,(1)}_{12}$,
$r^{(0)}_{12}$, $r^{(1)}_{12}$,
${\ti r}^{(-1)}_{12}$, ${\ti r}^{(0)}_{12}$, ${\ti r}^{(1)}_{12}$ --
coefficients in the expansions (\ref{t213})--(\ref{t215});

${\mathbb J}$ -- linear operator (\ref{t112}) defining the inverse inertia
tensor in the Euler equation (\ref{t107});

${\mathbb{FP}}_x$ and ${\mathbb{{\ti F}P}}_x$ -- the linear operators (\ref{t4214});

${\mathbb{F}}_x$ and ${\mathbb{{\ti F}}}_x$ -- the linear operators (\ref{t6051});

$E_{ij}$ -- matrix units in ${\rm Mat}(n,\mC)$ or ${\rm Mat}(2n,\mC)$;

$e_{ij}$ -- matrix units in ${\rm Mat}(N,\mC)$;

$T_a$ -- basis matrices in $\Mat$ (\ref{t231});

$Q$ and $\Lambda$ -- matrices entering the definition of $T_a$ (\ref{t232})-(\ref{t233});

$\sigma_k$ -- Pauli matrices (\ref{t259});

$\tau$ -- modulus of the elliptic curve (${\rm Im}(\tau)>0$);

$\Om_a$ -- half-periods of the elliptic curve (\ref{a18});

$\om_a$ -- arguments of the $\phi$-function in (\ref{a32});
for $a=(a_1,a_2)\in\mZ_N\times\mZ_N$, $\om_a=\frac{a_1+a_2\tau}{N}$;

$z_i$ -- equilibrium position of the $i$-th particle.

\bigskip

\underline{Notation for coupling constants:}

$\nu_a$ ($a=0,1,2,3$) -- parameters of ${\ti K}^z(q)$ (\ref{t252});
the physical boundary couplings in (\ref{t118}) are $\bar\nu_a$;

$\nu_a'$ -- the same parameters with the possible sign change specified
in (\ref{t255});

$\bar\nu_a$ or $\bar\nu_a'$ ($a=0,1,2,3$) -- linear combinations of $\nu_a$
or $\nu_a'$ defined by the transformation (\ref{a48});

$\la_a$ -- constants defined in (\ref{t2641});

$\mu_a$ -- constants defined in (\ref{t266});

$\rho_a$ -- parameters in the ${\rm GL}_2$ expression (\ref{t261})
for the $K$-matrix.

%
%
%
%
%
%
%
%

}}

\newpage

\section{Introduction}\label{sec1}
\setcounter{equation}{0}
In recent years, there has been renewed interest in spin generalizations
of integrable particle systems at both the classical and quantum levels
\cite{Arut,Ch3,ChM,Feher3,GrZ18,GSZ,KLOZ,Lam,Lam2,MZ,MMZ,MMZ2,Resh,SeZ18,TrZ,ZZ,Z18}.
This interest is motivated by the deep connections of these systems with
the representation theory of quantum algebras, the geometry of spectral curves,
gauge theories, and other areas. Recent works \cite{MMZ,MMZ2} introduced
a parametric generalization of the type-$B$ associative Yang--Baxter equation
algebra originally proposed in \cite{Kir}. This made it possible to construct
$R$-matrix-valued Lax representations for the elliptic ${\rm BC}_n$ Calogero--Inozemtsev
system and to derive the Hamiltonian of an elliptic anisotropic long-range
quantum chain of Haldane--Shastry--Inozemtsev type associated with the
${\rm BC}_n$ root system. In this paper, we use the ${\rm BC}_n$ associative
Yang--Baxter equation algebra and its matrix representations to construct
classical integrable\footnote{Throughout this paper, integrability refers to
the existence of a Lax representation with a spectral parameter.
Liouville integrability also requires the Hamiltonians obtained from the
Lax matrix to be in involution. This property can be established through
a classical $r$-matrix structure. We assert that such a structure exists
for the models considered here and will present it in a subsequent paper.}
systems that extend the elliptic ${\rm BC}_n$ Calogero--Inozemtsev system
to various spin models.
\subsection{Models associated with the \texorpdfstring{$A$}{A}-type root system}
\paragraph{Calogero--Moser model.} We consider integrable many-body systems
of Calogero--Moser type, their spin generalizations, and models of interacting
tops. We begin with the spinless elliptic Calogero--Moser model \cite{Ca,OP}.
Its Hamiltonian describes pairwise interactions between particles
in the complex plane:
  \beq\label{t102}
  \begin{array}{c}
  \displaystyle{
H^{\hbox{\tiny{CM}}}=\sum\limits_{i=1}^n\frac{p_i^2}{2}-g^2\sum\limits_{i>j}^n\wp(q_i-q_j)\,,
 }
 \end{array}
 \eq
where $p_i,q_i$ are canonically conjugate particle momenta and positions,
$g$ is a coupling constant, and $\wp$ is the Weierstrass elliptic function.
The equations of motion admit a Lax representation with a spectral
parameter $z$ \cite{Krich1}:
  \beq\label{t101}
  \begin{array}{c}
  \displaystyle{
 \dot{L}(z)=[L(z),M(z)]\,,\qquad L(z),M(z)\in{\rm Mat}(n,\mC)\,.
 }
 \end{array}
 \eq
The Lax matrices $L(z),M(z)$ have size $n\times n$ for the $n$-body system.

\paragraph{Spin Calogero--Moser model.}
The Calogero--Moser model admits a spin generalization \cite{BAB,GH}.
The phase space acquires an additional component described by the entries
$S_{ij}$ of a matrix $S\in{\rm Mat}(n,\mC)$, which represent classical
spin variables. The Hamiltonian takes the form
  \beq\label{t103}
  \begin{array}{c}
  \displaystyle{
H^{\hbox{\tiny{spin}}}=\sum\limits_{i=1}^n\frac{p_i^2}{2}-\sum\limits_{i>j}^n S_{ij}S_{ji}\wp(q_i-q_j)\,.
 }
 \end{array}
 \eq
This model is integrable on the constraint surface
$S_{ii}=S_{jj}={\rm const}$. When $S$ has rank one, the spin model
(\ref{t103}) reduces to the spinless model (\ref{t102}).

\paragraph{Generalized model of interacting tops.}
The spin Calogero--Moser model can in turn be extended to a generalized model
of interacting tops \cite{ZL}; see also the reviews \cite{ZS,TrZ}.
Both the spin matrix and the Lax matrix are enlarged to matrices of size $Nn$,
with $\mS\in{\rm Mat}(Nn,\mC)$. The Lax matrix has the block structure
 \beq\label{t104}
  \displaystyle{
 \mL^{\hbox{\tiny{Gtops}}}(z)=\sum\limits_{i,j=1}^n\sum\limits_{a,b=1}^N \mL^{ij}_{ab}(z) E_{ij}\otimes e_{ab}=
 \underbrace{
\left(
 \begin{array}{cccc}
\mL^{11}(z) &\mL^{12}(z)   & \dots &\mL^{1n}(z)
 \\
\mL^{21}(z) &\mL^{22}(z)  & \dots &\mL^{2n}(z)
 \\
\vdots & \vdots & \dots  & \vdots
 \\
\mL^{n1}(z) &\mL^{n2}(z) &  \dots &\mL^{nn}(z)
 \end{array}
 \right)
 }_{n\ {\rm blocks\ in\ a\ row\ (or\ column)\ of\ }N\times N{\rm \ size\ each}}
 \in{\rm Mat}_{nN}\,,
 }
 \eq
where $E_{ij}$ and $e_{ij}$ are matrix units in ${\rm Mat}(n,\mC)$
and ${\rm Mat}(N,\mC)$, respectively. The superscript {\it Gtops} in
(\ref{t104}) denotes the generalized model of interacting tops.
Each block has size $N\times N$, so $\mL^{ij}(z)\in{\rm Mat}_N$.
The Hamiltonian can be written as
 \beq\label{t105}
 \begin{array}{c}
  \displaystyle{
 \mH^{\hbox{\tiny{Gtops}}}=\sum\limits_{i=1}^n \frac{p_i^2}{2}+\sum\limits_{i=1}^n
 H^{\hbox{\tiny{top}}}(\mS^{ii})+\sum\limits_{i> j}^n
 \mU(\mS^{ij},\mS^{ji},q_i-q_j)\,,
 }
 \end{array}
 \eq
 where $\mU(\mS^{ij},\mS^{ji},q_i-q_j)$ is a potential.
For $N=1$, the generalized model of interacting tops reduces to
the spin Calogero--Moser model (\ref{t103}):
$\mS^{ij}|_{N=1}=S_{ij}$ and
$\mU(\mS^{ij},\mS^{ji},q_i-q_j)|_{N=1}=-S_{ij}S_{ji}\wp(q_i-q_j)$.

\paragraph{Integrable Euler--Arnold tops.}
For $n=1$, we recover the elliptic integrable Euler--Arnold top
described in \cite{LOZ}, with Hamiltonian
  \beq\label{t106}
  \begin{array}{c}
     \displaystyle{
 H^{\hbox{\tiny{top}}}(S)=\frac{1}{2}\,\tr\Big(S\,{\mathbb J}(S)\Big)\,,
 \quad S\in\Mat\,,
 }
  \end{array}
  \eq
where ${\mathbb J}(S)$ is a linear map whose explicit expression is given
in (\ref{t340}). The equations of motion take the Euler form
 \beq\label{t107}
 \begin{array}{c}
  \displaystyle{
\dot {S}=\{H^{\hbox{\tiny{top}}}(S),S\}=[{S}, {\mathbb J}({S})]\,.
 }
 \end{array}
 \eq
For $N=2$, this model reproduces the complexified Euler top in
three-dimensional space, with $J$ playing the role of the inverse inertia tensor.

\paragraph{Integrable tops from AYBE.}
The construction of \cite{LOZ161} yields integrable models of the form
(\ref{t106})--(\ref{t107}) from a class of quantum $R$-matrices in the
fundamental representation of the Lie group ${\rm GL}_N$.
By definition, a quantum $R$-matrix satisfies the quantum Yang--Baxter
equation (QYBE):\footnote{In general, an $R$-matrix may depend on two spectral
parameters, $R_{12}^\hbar(q_1,q_2)$. Here we consider $R$-matrices that depend
only on their difference, $R_{12}^\hbar(q_1,q_2)=R_{12}^\hbar(q_1-q_2)$.}
  \beq\label{t108}
  \begin{array}{l}
  \displaystyle{
 R^{\hbar}_{12}(q_1-q_2)  R^{\hbar}_{13}(q_1-q_3) R^{\hbar}_{23}(q_2-q_3) =
 R^{\hbar}_{23}(q_2-q_3)  R^{\hbar}_{13}(q_1-q_3)  R^{\hbar}_{12}(q_1-q_2)\,,
 }
 \end{array}
 \eq
where $q_i$ are spectral parameters and $\hbar$ is the anisotropy parameter,
or Planck constant.\footnote{In some of the constructions below, the spectral
parameters of the $R$-matrix are identified with particle positions, while
the Planck constant $\hbar$ serves as the spectral parameter $z$ in the
Lax equation.}
Suppose that the $R$-matrix also satisfies the associative Yang--Baxter
equation \cite{FK}, in the parametric form introduced in \cite{Pol}:
  \beq\label{t109}
  R^\hbar_{12}(q_{1}-q_2)
 R^{\eta}_{23}(q_{2}-q_3)=R^{\eta}_{13}(q_{1}-q_3)R_{12}^{\hbar-\eta}(q_{1}-q_2)+
 R^{\eta-\hbar}_{23}(q_{2}-q_3)R^\hbar_{13}(q_{1}-q_3)
  \eq
together with the additional properties specified in the next section.
Consider its quasiclassical expansion
  \beq\label{t110}
  \begin{array}{c}
      \displaystyle{
R^\hbar_{12}(z)=\frac{1}{\hbar}\,1_N\otimes 1_N+r_{12}(z)+\hbar\,
m_{12}(z)+O(\hbar^2) }\,.
  \end{array}
  \eq
Then the $N\times N$ Lax pair
 \beq\label{t111}
  \begin{array}{c}
  \displaystyle{
  L(S,z)=\tr_2\Big(r_{12}(z){S}_2\Big)\,,\qquad M(S,z)=\tr_2\Big(m_{12}(z){S}_2\Big)\,,\qquad
  {S}_2=1_N\otimes S\,,
 }
 \end{array}
 \eq
satisfies the Lax equation identically in $z$ and yields the Euler equation
(\ref{t107}), with the inverse inertia tensor ${\mathbb J}$ given by
  \beq\label{t112}
  \begin{array}{c}
      \displaystyle{
 {\mathbb J}(S)=\tr_2(m_{12}(0)S_2)\,,\qquad S_2=1_N\otimes S\,.
  }
  \end{array}
  \eq
In (\ref{t111})--(\ref{t112}), $\tr_2$ denotes the trace over the second
tensor component. Polishchuk \cite{Pol} showed that, with the appropriate
normalization, the elliptic Baxter--Belavin quantum $R$-matrix \cite{BB}
satisfies the AYBE (\ref{t109}). For this $R$-matrix, the construction
reproduces the elliptic top of \cite{LOZ}.

\paragraph{Interacting tops.}
The generalized model (\ref{t105}) has a useful interpretation in the
rank-one case. If $\mS$ has rank one, its Hamiltonian can be expressed
solely in terms of the diagonal blocks $\mS^{ii}$:
 \beq\label{t113}
 \begin{array}{c}
  \displaystyle{
 \mH^{\hbox{\tiny{tops}}}=\sum\limits_{i=1}^n \frac{p_i^2}{2}+\sum\limits_{i=1}^n
 H^{\hbox{\tiny{top}}}(\mS^{ii})+\sum\limits_{i> j}^n
 \mV(\mS^{ii},\mS^{jj},q_i-q_j)\,,
 }
 \end{array}
 \eq
since in this case
 \beq\label{t114}
 \begin{array}{c}
  \displaystyle{
 \mU(\mS^{ij},\mS^{ji},q_i-q_j)= \mV(\mS^{ii},\mS^{jj},q_i-q_j)\ \forall i\neq j\,.
 }
 \end{array}
 \eq
This gives a direct physical interpretation of the Hamiltonian:
it describes $n$ integrable ${\rm gl}_N$ tops attached to $n$ particles,
with interactions depending on the positions $q_i$ and the corresponding
spin variables. The model was introduced in \cite{Polych}.
Its Lax representation was subsequently obtained in \cite{ZL} as the
special case rank$(\mS)=1$ of the generalized model of interacting tops.

\paragraph{\texorpdfstring{$R$}{R}-matrix-valued Lax pairs and long-range spin chains.}
Another application of the AYBE is the construction of $R$-matrix-valued
Lax pairs \cite{LOZ14,GrZ18,Z18}. These are half-quantum: the particle
positions and momenta remain classical, whereas the spin variables are
quantized in the fundamental representation of the Lie algebra ${\rm gl}_N$.
This construction yields anisotropic long-range spin chains of
Haldane--Shastry--Inozemtsev type \cite{SeZ18}, with quantum Hamiltonian
  \beq\label{t117}
  \begin{array}{c}
  \displaystyle{
 {\hat\mH}^{\hbox{\tiny{long-range}}}=\sum\limits_{i>j}^n
 F^0_{ij}(z_i-z_j)\,,\qquad F^0_{ij}(x)=\p_xR^\hbar_{ij}(x)\Big|_{\hbar=0}
 \stackrel{(\ref{t110})}{=}
 \p_{x}r_{ij}(x)\,.
 }
  \end{array}
  \eq
The idea is to restrict the classical Lax equations to an equilibrium
configuration $z_j$ of the underlying spinless model (\ref{t102}).
This produces a quantum Lax equation with Hamiltonian (\ref{t117}).
As observed in \cite{GrZ18,GSZ}, this Hamiltonian can also be obtained by
directly quantizing the interaction potential.

A fully quantum $R$-matrix-valued Lax pair, in which the particle momenta
are also quantized, can be viewed as a quantum Lax representation of the
model of interacting tops. The existence of such a pair does not, however,
guarantee integrability; in the present setting, it also does not provide
a method for computing higher Hamiltonians. A resolution of this problem
was obtained through the proof of integrability of a relativistic
generalization, or $q$-deformation, of (\ref{t117}) \cite{MZ}.

\subsection{\texorpdfstring{${\rm BC}_n$}{BCn} AYBEs, related models, and aim of the paper}
\paragraph{Calogero--Inozemtsev model.}
The spinless Calogero--Moser model (\ref{t102}) admits generalizations
associated with root systems other than type $A$ \cite{OP,Phong}.
The most general of these\footnote{Here the degree of generality refers to the
number of independent coupling constants.} is the Calogero--Inozemtsev
model, introduced in \cite{OP} and subsequently extended in \cite{Inoz89}
to five independent coupling constants. We consider this five-parameter
model, with Hamiltonian
\begin{equation}
 \label{t118}
   \displaystyle{
 H^{\hbox{\tiny{CI}}} = \frac{1}{2} \sum_{k=1}^{n} p_k^2 - g^2
\sum_{i<j}^{n} \Big(\wp(q_i-q_j) + \wp(q_i+q_j)\Big) -
\frac{1}{2}\sum\limits_{a=0}^3\sum_{k=1}^{n} \bar\nu_a^2\wp(q_k+\Om_a)\,,
}
 \end{equation}
where $\Om_\ga$ are the half-periods (\ref{a18}), and
$g,\bar\nu_0,\bar\nu_1,\bar\nu_2,\bar\nu_3\in\mC$ are arbitrary
constants.\footnote{In our construction, the boundary potential in
(\ref{t118}) arises from the unitarity relation for the $K$-matrix.
Different linear combinations of its parameters therefore occur
throughout the paper.}
A $3n\times 3n$ Lax representation was given in \cite{Inoz89}.
Here we use the $2n\times 2n$ Lax pair introduced by Takasaki \cite{Ta}.
Lax pairs of size $2n\times 2n$ for the trigonometric and rational cases
were also obtained in \cite{Feher1,Marshall}.

\paragraph{\texorpdfstring{${\bf BC}_n$}{BCn}-type AYBEs and long-range spin chains.}
The parametric ${\rm BC}_n$-type AYBEs introduced in \cite{MMZ} extend
the nonparametric $B$-type AYBE of An.~Kirillov \cite{Kir}.
Besides $R_{12}^\hbar(z)$, they involve two additional generators,
${\tilde R}_{12}^\hbar(z)$ and ${\ti K}^\hbar(z)$.
The matrix ${\tilde R}_{12}^\hbar(z)$ is obtained by applying an
involution to the second tensor component of ${R}_{12}^\hbar(z)$.
Such objects arise in the reflection equation algebra \cite{KuSkl,FM}.
Boundary $K$-matrices were introduced in the study of quantum integrable
systems with boundaries \cite{Cherednik,Skl-refl}; they satisfy the reflection equation
 \beq\label{t119}
 \begin{array}{c}
 \displaystyle{
R_{12}^{\hbar}(u-v)\tilde{K}_1^\hbar(u)\tilde{R}_{12}^{\hbar}(u+v)
\tilde{K}_2^\hbar(v)=\tilde{K}_2^\hbar(v)\tilde{R}_{12}^{\hbar}(u+v)
\tilde{K}_1^\hbar(u)R_{12}^{\hbar}(u-v)\,.
}
\end{array}
 \eq
It was shown in \cite{MMZ2} that a slightly modified version of the
elliptic Komori--Hikami $K$-matrix \cite{KH2} satisfies both (\ref{t119})
and the parametric ${\rm BC}_n$-type AYBEs, including the four-term relation
 \beq\label{t120}
 \begin{array}{c}
\displaystyle{
 R_{ij}^{w+z}(q_i\!-\!q_j)\tilde{K}^w_j(q_j)=\tilde{K}^w_i(q_i)R_{ij}^{z-w}(q_i\!-\!q_j)\!+\!
 \tilde{R}_{ij}^{w-z}(q_i\!+\!q_j)\tilde{K}^z_i(q_i)\!+\!
 \tilde{K}_j^{-z}(q_j)\tilde{R}_{ij}^{w+z}(q_i\!+\!q_j)\,.
}
\end{array}
 \eq
An $R$-matrix-valued Lax pair for the Calogero--Inozemtsev model
(\ref{t118}) was constructed in \cite{MMZ} as a generalization of
Takasaki's $2n\times 2n$ Lax pair. The construction includes all five
independent coupling constants; certain special cases had previously
been obtained in \cite{GrZ18}. Combining matrix representations of
(\ref{t120}) and the other ${\rm BC}_n$ AYBEs with this Lax pair gives
the following analogue of the long-range spin-chain Hamiltonian
(\ref{t117}):
  \beq\label{t121}
  \begin{array}{c}
  \displaystyle{
 {\hat\mH}^{\hbox{\tiny{long-range}}}=\sum\limits_{i>j}^n \Big(F^0_{ij}(z_i-z_j)+{\ti F}^0_{ij}(z_i+z_j)\Big)+
 \frac12\sum\limits_{k=1}^n \tilde{Y}_k^0(z_k)\,,
 }
 \\ \ \\
  \displaystyle{
F^0_{ij}(x)=\p_xR^\hbar_{ij}(x)\Big|_{\hbar=0}\,,\qquad \tilde{Y}^0(x)=\p_x {\ti K}^\hbar(x)\Big|_{\hbar=0}\,,
 }
  \end{array}
  \eq
where $z_i$ are equilibrium positions of the classical spinless model (\ref{t118}).
The four parameters $\nu_a$ of ${\ti K}^\hbar(x)$, and hence of
$\tilde{Y}^0(x)$, enter the last term of (\ref{t118}) through the
physical couplings $\bar\nu_a$ defined by (\ref{a48}). Trigonometric isotropic
long-range spin chains with boundaries were introduced in \cite{BPS} and then generalized in \cite{EFGR},
extending the famous $A$-type Haldane-Shastry chain \cite{HS1}.

\paragraph{Purpose of the paper.}
The main contribution of this paper is a classical Hamiltonian
construction of generalized interacting tops associated with the
${\rm BC}_n$ root system. It extends the $A$-type construction of
\cite{GSZ} using the $R$- and $K$-matrix data developed in
\cite{MMZ,MMZ2}. The latter works establish the relevant
${\rm BC}_n$ associative Yang--Baxter identities and their matrix
representations, and use them to construct $R$-matrix-valued Lax
pairs and long-range quantum spin chains. Here these data are used
to define a system with dynamical classical spin variables.
For arbitrary $n$ and $N$, we construct an explicit
$2nN\times 2nN$ Lax pair and a compatible Hamiltonian structure,
as stated in Theorems \ref{th1} and \ref{th2}.

The spin degrees of freedom are described by two coadjoint orbits
of ${\rm GL}_{nN}$, with the diagonal reduction specified below.
The construction combines interactions depending on $q_i-q_j$
and $q_i+q_j$ with local terms determined by the boundary
$K$-matrix. Its four parameters enter the classical Hamiltonian
and determine the boundary couplings of the reduced models.
For elliptic $R$- and $K$-matrices, the specialization $N=1$
gives a two-spin extension of the Calogero--Inozemtsev system.
The restriction $\mA=\mB$ and the subsequent rank-one reduction
relate the generalized model to interacting ${\rm gl}_N$ tops
and, at $N=1$, to the spinless Calogero--Inozemtsev system.

Two-spin Sutherland systems were previously studied in
\cite{KLOZ,Feher3}, while related elliptic systems were constructed
from Higgs bundles in \cite{LOSZ}. The construction presented here
is formulated directly in terms of the ${\rm BC}_n$ associative
Yang--Baxter identities and boundary $K$-matrices.
We also quantize the interacting-top Hamiltonians in the fundamental
representation and identify their relation, under the freezing
prescription, to the chains of \cite{MMZ,MMZ2}, including the
required shifts of the boundary parameters.

The Lax matrix has size
$2nN\times 2nN$ and the following block structure:
\beq\label{t122}
\begin{array}{c}
  \displaystyle{
    \mL^{\hbox{\tiny{Gtops}}}(z)=
    }
    \\[3mm]
  \displaystyle{
    =\sum^{n}_{i,j=1}\sum^{N}_{a,b=1}\Big( E_{ij}\otimes e_{ab} \mL^{ij}_{ab}(z)+E_{\bar{i}j}\otimes e_{ab}
    \mL^{\bar{i}j}_{ab}(z)+E_{i\bar{j}}\otimes e_{ab}\mL^{i\bar{j}}_{ab}(z)
    +E_{\bar{i}\bar{j}}\otimes e_{ab} \mL^{\bar{i}\bar{j}}_{ab}(z)\Big)\,,
  }
\end{array}
\eq
where $\bar i=i+n$. Equivalently,
  \beq\label{t123}
  \begin{array}{c}
  \displaystyle{
\mL^{\hbox{\tiny{Gtops}}}(z)=\mat{||\mL^{ij}(z)||}{||\mL^{i{\bar j}}(z)||}{||\mL^{{\bar i}j}(z)||}{||\mL^{{\bar i}\bar{j}}(z)||}\,,
 }
  \end{array}
  \eq
Each of the four blocks has size $nN\times nN$ and the same structure
as $\mL^{\hbox{\tiny{Gtops}}}(z)$ for the $A$-type model (\ref{t104}).
The Lax representation is expressed in terms of $R$- and $K$-matrices
satisfying the ${\rm BC}_n$ AYBEs and thus defines a family of models.
We also describe their Hamiltonian structure. The spin matrix
$\mS\in{\rm Mat}_{2nN}$ has the form
  \beq\label{t124}
  \begin{array}{c}
  \displaystyle{
\mS=\mats{\mA}{\mB}{\mB}{\mA}\,,\qquad \mA\,,\mB\in{\rm Mat}_{nN}\,.
 }
  \end{array}
  \eq
%
The phase space of the model is
  \beq\label{t125}
  \begin{array}{c}
  \displaystyle{
T^*\mC^n\times\Big((\mO^{\hbox{\tiny{$nN$}}}_+\times \mO^{\hbox{\tiny{$nN$}}}_-)//{\mathcal D}\Big)\,,
 }
  \end{array}
  \eq
%
where ${\mathcal D}$ is the following subgroup of the Cartan subgroup
of ${\rm GL}_{nN}$:
  \beq\label{t126}
  \begin{array}{c}
  \displaystyle{
{\mathcal D}={\rm diag}(d_1,...,d_n)\otimes 1_N\in{\rm Mat}_{nN}\,,
 }
  \end{array}
  \eq
and $\mO^{\hbox{\tiny{$nN$}}}_+$, $\mO^{\hbox{\tiny{$nN$}}}_-$ are
coadjoint orbits of the Lie group ${\rm GL}_{nN}$.
The corresponding matrix variables are
  \beq\label{t127}
  \begin{array}{c}
  \displaystyle{
\mO^{\hbox{\tiny{$nN$}}}_\pm:\ \ {\mathcal T}^{\pm}=\frac{\mA\pm \mB}{2}\in{\rm Mat}_{nN}\,.
 }
  \end{array}
  \eq

We next consider several special cases. For $N=1$, the $R$- and
$K$-matrices become scalar functions, yielding a spin generalization of
the elliptic Calogero--Inozemtsev model with a $2n\times 2n$ Lax pair.

There are two sets of spin variables,
$A=\mA|_{N=1}$ and $B=\mB|_{N=1}$ (or $T^\pm=\mT^\pm|_{N=1}$),
associated with the two coadjoint orbits (\ref{t127}).
Trigonometric Calogero--Inozemtsev systems with two sets of spin variables
were previously studied by different methods in \cite{KLOZ,Feher3}.
The spin part of the phase space (\ref{t125}),
$(\mO^{\hbox{\tiny{$n$}}}_+\times \mO^{\hbox{\tiny{$n$}}}_-)//{\mathcal D}$,
includes reduction by $\mathcal D$. For $N=1$, $\mathcal D$ is the full
Cartan subgroup of ${\rm GL}_n$.


Another class of special cases is obtained by identifying $\mA$ and $\mB$
in (\ref{t124}):
  \beq\label{t128}
  \begin{array}{c}
  \displaystyle{
\mA=\mB\,.
 }
  \end{array}
  \eq
This constraint is compatible with the equations of motion and amounts
to setting $\mT^-=0$. Taking $N=1$ then gives a spin generalization of
the elliptic Calogero--Inozemtsev model with a single set of spin variables
on one coadjoint orbit. Its Hamiltonian is
  \beq\label{t129}
  \begin{array}{c}
  \displaystyle{
H^{\hbox{\tiny{spin}}}=\frac12\,\sum_{i=1}^n p_i^2
-\frac12\,\sum_{i\neq j}^n A_{ij}A_{ji}
\Big(\wp(q_i-q_j)+\wp(q_i+q_j)\Big)-
  }
  \\[5mm]
  \displaystyle{
-\frac12\,\sum_{i=1}^n A_{ii}^2\wp(2q_i)
-\frac12\,\sum_{i=1}^n\sum_{a=0}^3
\la_a(\la_a+A_{ii})\wp(q_i+\Om_a)\,,
 }
  \end{array}
  \eq
where $\la_a$ are constants and $A=\mA|_{N=1}\in{\rm Mat}_n$
is the spin matrix. This model is integrable on the constraint surface
$A_{ii}=A_{jj}={\rm const}$. Quantum models of this type were proposed
in \cite{Finkel} using Dunkl operators.
For $N=1$, $\mA=\mB$, and rank$(\mA)=1$, the spin model (\ref{t129})
reduces to the spinless Calogero--Inozemtsev model (\ref{t118}),
as in the $A$-type spin Calogero--Moser model. This reduction reproduces
Takasaki's $2n\times 2n$ Lax representation of (\ref{t118}).

For $\mA=\mB$ and rank$(\mA)=1$, we obtain a model of $n$ interacting
${\rm gl}_N$ tops with interactions associated with the ${\rm BC}_n$ root system.

Finally, we discuss applications to quantum long-range spin chains of
Haldane--Shastry--Inozemtsev type \cite{HS1,Inoz90}. We obtain quantum
Hamiltonians of interacting tops by direct quantization in the fundamental
representation of the Lie algebra ${\rm gl}_N$. Applying the Polychronakos
freezing prescription \cite{Poly} then yields a family of quantum
Hamiltonians, including (\ref{t121}), which was derived in \cite{MMZ,MMZ2}
using $R$-matrix-valued Lax pairs. These results provide a construction
of a new family of (presumably) integrable quantum elliptic spin chains, although
a rigorous proof of integrability remains a separate problem.
Their relation to interacting tops may also facilitate the computation
of higher Hamiltonians.


\paragraph{Outline of the paper.}
Section \ref{sec2} recalls the $R$- and $K$-matrix identities used
in the construction. Section \ref{sec3} reviews the corresponding
$A$-type models. Section \ref{sec4} presents the generalized
${\rm BC}_n$ model, its Lax pair and Hamiltonian structure.
Section \ref{sec5} studies the $N=1$ spin models and their spinless
reduction. Section \ref{sec6} describes interacting tops and the
explicit XYZ and XXZ cases. Section \ref{sec7} gives quantum Hamiltonians
in the fundamental representation and relates them to long-range spin chains.
Appendix A collects elliptic identities; Appendices B and C prove
Theorems \ref{th1} and \ref{th2}, respectively.

\section{\texorpdfstring{$R$}{R}-matrices, \texorpdfstring{$K$}{K}-matrices and Yang--Baxter equations}\label{sec2}
\setcounter{equation}{0}
Following \cite{MMZ,MMZ2}, we recall the parametric ${\rm BC}_n$ AYBEs.
They involve three sets of generators, taken here in matrix representations:
$R^z_{ij}(x)$, ${\ti R}^z_{ij}(x)$, and $\ti{K}^{z}_{j}(x)$.
The quantum Yang--Baxter equation is
\beq\label{t2011}
  \displaystyle{
R^z_{ij}(q_i-q_j)R^z_{ik}(q_i-q_k)R^z_{jk}(q_j-q_k)
=R^z_{jk}(q_j-q_k)R^z_{ik}(q_i-q_k)R^z_{ij}(q_i-q_j)\,.
}
\eq
For ${\ti R}^z_{ij}(x)$ it reads as follows:
\beq\label{t2012}
  \displaystyle{
R^z_{ij}(q_i-q_j)\tilde{R}^z_{ik}(q_i-q_k)\tilde{R}^z_{jk}(q_j-q_k)=\tilde{R}^z_{jk}(q_j-q_k)
\tilde{R}^z_{ik}(q_i-q_k)R^z_{ij}(q_i-q_j)\,.
}
\eq
The generator $\ti{K}^{z}_{j}(x)$ is a solution of the
\underline{reflection equation} \cite{Skl-refl}, which encodes boundary
conditions for quantum integrable systems on a finite lattice.\footnote{Following
\cite{Skl-refl}, the reflection equation is usually written as
$R_{ij}^{z}(q_i-q_j)K_i^z(q_i)R_{ji}^{z}(q_i+q_j)K_j^z(q_j)=
K_j^z(q_j)R_{ij}^{z}(q_i+q_j)K_i^z(q_i)R_{ji}^{z}(q_i-q_j)$.
The $K$-matrix in this form is related to ${\ti K}^z(q)$ used here by
multiplication by a constant matrix; see \cite{MMZ2} for details.}
In our notation, it reads
\beq\label{t201}
\begin{array}{c}
  \displaystyle{
    R^{z}_{ij}(q_{ij})\ti{K}^{z}_{i}(q_{i})\ti{R}^{z}_{ij}(q_{ij}^{+})
    \ti{K}^{z}_{j}(q_{j})=\ti{K}^{z}_{j}(q_{j})\ti{R}^{z}_{ij}(q_{ij}^{+})
    \ti{K}^{z}_{i}(q_{i})R^{z}_{ij}(q_{ij})\,.
  }
\end{array}
\eq
It was shown in \cite{MMZ,MMZ2} that suitably normalized solutions
$\ti{K}^{z}_{j}(x)$ of the reflection equation also satisfy the
${\rm BC}_n$ AYBEs. A general proof of this correspondence is not yet known.
The proofs in \cite{MMZ,MMZ2} use the explicit elliptic solutions of
(\ref{t201}) given in \cite{IK,Vega,KH2}.

\subsection{Parametric \texorpdfstring{$\rm BC_{n}$}{BCn} associative Yang--Baxter equations}
We assume that the $\tilde{K},\tilde{R},R$-matrices satisfy the following
relations.

\paragraph{\texorpdfstring{$\bf BC_{n}$}{BCn} AYBEs.}\

\indent \underline{Commutation relations for distinct $i,j,k,l$:}
\beq\label{t206}
\begin{array}{c}
  \displaystyle{
    [\ti{K}^{z}_{i}(q_{i}),\ti{K}^{w}_{j}(q_{j})]= 0,
    \qquad[\ti{K}^{z}_{i}(q_{i}),R^{w}_{kl}(q_{kl})]= 0,\qquad[\ti{K}^{z}_{i}(q_{i}),
    \ti{R}^{w}_{kl}(q_{kl})] = 0\,,
  }
\end{array}
\eq
\beq\label{t207}
\begin{array}{c}
  \displaystyle{
    [R^{z}_{ij}(q_{ij}),R^{w}_{kl}(q_{kl})]= 0,\qquad [R^{z}_{ij}(q_{ij}),\ti{R}^{w}_{kl}(q_{kl})]
    =0,\qquad [\ti{R}^{z}_{ij}(q_{ij}),\ti{R}^{w}_{kl}(q_{kl})] =0\,,
  }
\end{array}
\eq
\indent \underline{Three-term relations:}
\beq\label{t208}
\begin{array}{c}
  \displaystyle{
    R_{ij}^{z}(q_{ij})R^{w}_{jk}(q_{jk})=R^{w}_{ik}(q_{ik})R^{z-w}_{ij}(q_{ij})
    + R^{w-z}_{jk}(q_{jk})R^{z}_{ik}(q_{ik})\,,
  }
\end{array}
\eq
\beq\label{t209}
\begin{array}{c}
  \displaystyle{
    R^{z}_{ij}(q_{ij})\ti{R}^{w}_{jk}(q_{jk}^{+})=\ti{R}^{w}_{ik}(q_{ik}^{+})R^{z-w}_{ij}(q_{ij})
    + \ti{R}^{w-z}_{jk}(q_{jk}^{+})\tilde{R}^{z}_{ik}(q_{ik}^{+})\,,
  }
\end{array}
\eq
\beq\label{t210}
\begin{array}{c}
  \displaystyle{
    R^{z}_{jk}(q_{jk})\ti{R}^{w}_{ik}(q_{ik}^{+})=\tilde{R}^{w}_{ij}(q_{ij}^{+})R^{z-w}_{jk}(q_{jk})
    +\ti{R}^{w-z}_{ik}(q_{ik}^{+})\ti{R}^{z}_{ij}(q_{ij}^{+})\,,
  }
\end{array}
\eq
\beq\label{t211}
\begin{array}{c}
  \displaystyle{
    \ti{R}^{z}_{jk}(q_{jk}^{+})R^{w}_{ik}(q_{ik})=\tilde{R}^{w}_{ij}(q_{ij}^{+})
    \ti{R}^{z-w}_{jk}(q_{jk}^{+})+R^{w-z}_{ik}(q_{ik}){\ti R}^{z}_{ij}(q_{ij}^+)\,,
  }
\end{array}
\eq
\indent \underline{The four-term relation:}
\beq\label{t212}
\begin{array}{c}
  \displaystyle{
    R^{w+z}_{ij}(q_{ij})\ti{K}_{j}^{w}(q_{j})=\tilde{K}^{w}_{i}(q_{i})R^{z-w}_{ij}(q_{ij})
    +\ti{R}^{w-z}_{ij}(q^{+}_{ij})\ti{K}_{i}^{z}(q_{i})+\ti{K}^{-z}_{j}(q_{j})
    \ti{R}^{w+z}_{ij}(q_{ij}^{+})\,.
  }
\end{array}
\eq

\paragraph{Skew-symmetry and unitarity.}\

\indent \underline{Symmetry and skew-symmetry properties:}
\beq\label{t204}
\begin{array}{c}
  \displaystyle{
    \ti{R}^{z}_{ij}(q)=\ti{R}^{z}_{ji}(q)\,,
  }
\end{array}
\eq
\beq\label{t205}
\begin{array}{c}
  \displaystyle{
    R^{z}_{ij}(q)=-R^{-z}_{ji}(-q)\,.
  }
\end{array}
\eq

\indent \underline{Unitarity properties:}
\beq\label{t218}
\begin{array}{c}
  \displaystyle{
    R^{z}_{12}(q)R^{z}_{21}(-q)=\Big(\wp(z)-\wp(q)\Big)1_N\otimes 1_N\,,
  }
\end{array}
\eq
\beq\label{t219}
\begin{array}{c}
  \displaystyle{
    \ti{R}^{z}_{12}(q)\ti{R}^{-z}_{21}(q)=\Big(\wp(q)-\wp(z)\Big)1_N\otimes 1_N\,,
  }
\end{array}
\eq
\beq\label{t220}
\begin{array}{c}
  \displaystyle{
    \tilde{K}^{z}(q)\ti{K}^{-z}(q)=k_N(z,q|\nu)1_N\,,\quad k_N(z,q|\nu)=u_N(q|\nu)-w_N(z|\nu)\,,
  }
\end{array}
\eq
where $k_N(z,q|\nu)$ and $u_N(q|\nu),w_N(z|\nu)$ are scalar functions.
They also depend on the parameters
$\nu=\{\nu_0,\nu_1,\nu_2,\nu_3\}$ in the definition of
$\tilde{K}^{z}_{i}(q)$.\footnote{In the most general elliptic case considered
here, the $K$-matrix depends on four constants, as in \cite{KH2}.}
Explicit elliptic expressions for $u_N(q|\nu)$ and $w_N(z|\nu)$ are
given in (\ref{t256}).

The Weierstrass $\wp$-function in the unitarity relations
(\ref{t218})--(\ref{t219}) corresponds to the elliptic case.
For trigonometric or rational $R$- and $K$-matrices, it is replaced
by its corresponding degeneration.

\paragraph{Local expansions.}
We use the following local expansions of the $R$- and $K$-matrices.

\indent \underline{Classical limit (local expansion near $z=0$):}
  \beq\label{t202}
  \begin{array}{c}
      \displaystyle{
R^z_{12}(q)=\frac{1}{z}\,1_N\otimes 1_N+r_{12}(q)+z\,m_{12}(q)+O(z^2) }\,,
  \end{array}
  \eq
  \beq\label{t203}
  \begin{array}{c}
      \displaystyle{
{\ti R}^z_{12}(q)=\frac{1}{z}\,1_N\otimes 1_N+{\ti r}_{12}(q)+z\,{\ti m}_{12}(q)+O(z^2) }\,,
  \end{array}
  \eq
where $1_N$ is the $N\times N$ identity matrix.

\underline{Expansions near $q=0$:}
\beq\label{t213}
\begin{array}{c}
  \displaystyle{
    R^{z}_{12}(q)=\frac{1}{q}\,P_{12}+R^{z,(0)}_{12}+qR^{z,(1)}_{12}+O(q^{2}),
    \qquad R^{z,(0)}_{12}=\frac{1}{z}\,1_N\ox1_N+r^{(0)}_{12}+O(z) \,,
  }
\end{array}
\eq
\beq\label{t214}
\begin{array}{c}
  \displaystyle{
    \ti{R}^{z}_{12}(q)=\frac{1}{q}\,{\ti R}_{12}^{(-1)}+\ti{R}^{z,(0)}_{12}+q\ti{R}^{z,(1)}_{12}
    +O(q^{2}), \qquad\ti{R}^{z,(0)}_{12}=\frac{1}{z}\,1_N\ox1_N+\ti{r}^{(0)}_{12}+O(z)\,,
  }
\end{array}
\eq
\beq\label{t215}
\begin{array}{c}
  \displaystyle{
    r_{12}(q)=\frac{1}{q}\,P_{12}+r^{(0)}_{12}+q\,r^{(1)}_{12}+O(q^{2}),
    \qquad\ti{r}_{12}(q)=\frac{1}{q}\,{\ti r}_{12}^{(-1)}+\ti{r}^{(0)}_{12}
    +q\,\ti{r}^{(1)}_{12}+O(q^{2})\,,
  }
\end{array}
\eq
where ${\ti R}_{12}^{(-1)}$ and ${\ti r}_{12}^{(-1)}$ are constant
matrices given in (\ref{t249}), and $P_{12}$ is the permutation operator
\beq\label{t216}
\begin{array}{c}
  \displaystyle{
    P_{12}= \sum_{i,j=1}^{N}e_{ij}\ox e_{ji}\in\Mat^{\otimes 2}\,,
  }
\end{array}
\eq
with $e_{ij}$ denoting the matrix units in $\Mat$.

The symmetry and skew-symmetry properties of these expansion coefficients
follow from (\ref{t204})--(\ref{t205}). For example,
\beq\label{t2161}
\begin{array}{c}
  \displaystyle{
    r_{12}(x)=-r_{21}(-x)\,,\quad m_{12}(x)=m_{21}(-x)\,,
 }
 \\ \ \\
   \displaystyle{
 {\ti r}_{12}(x)={\ti r}_{21}(x)\,,\quad {\ti r}_{12}(-x)=-J_1J_2{\ti r}_{12}(x)J_1J_2\,,
    }
\end{array}
\eq
where $J_1J_2=J\otimes J$, and the matrix $J$ is a certain involutive matrix, see (\ref{t246}) below for elliptic case.

\indent \underline{Expansion of the $\ti{K}$-matrix near $z=0$:}
\beq\label{t217}
\begin{array}{c}
  \displaystyle{
    \ti{K}^{z}(q) = \frac{1}{z}\,\ti{K}^{(-1)}+\ti{K}^{(0)}(q)+O(z)\,,
    \qquad\ti{K}^{(-1)}=\varsigma_N\cdot 1_N
\,,
  }
\end{array}
\eq
where $\varsigma_N$ is the constant given in (\ref{t2533}), and
$\ti{K}^{(0)}(q)$ satisfies the classical reflection equation
corresponding to (\ref{t201}):
\beq\label{t2171}
  \displaystyle{
[\ti{K}^{(0)}_1(q_1)+\ti{K}^{(0)}_2(q_2),r_{12}(q_1-q_2)]-[\ti{K}^{(0)}_1(q_1)
-\ti{K}^{(0)}_2(q_2),{\ti r}_{12}(q_1+q_2)]=0\,.
}
 \eq

\subsection{Further properties and identities}

A class of $R$-matrices also satisfies the following relation,
called \underline{Fourier symmetry}:
\beq\label{t221}
\begin{array}{c}
  \displaystyle{
    R^{z}_{12}(q)P_{12}=R^{q}_{12}(z)\,.
  }
\end{array}
\eq
The local expansions above then imply
\begin{equation}
\label{t222}
\begin{array}{cc}
    R^{z, (0)}_{12} = r_{12}(z) P_{12}\,, & r_{12}^{(0)} = r_{12}^{(0)} P_{12}\,,
     \\ \ \\
    R^{z, (1)}_{12} = m_{12}(z) P_{12}\,, & r_{12}^{(1)} = m_{12}(0) P_{12}\,.
\end{array}
\end{equation}

Introduce the following notation:
\beq\label{t223}
F_{ij}^z(q):=\partial_{q}R_{ij}^z(q)\,, \qquad\tilde{F}_{ij}^z(q):=\partial_{q}
\tilde{R}_{ij}^z(q)\,,\qquad\tilde{Y}^z_i(q):=\partial_q \tilde{K}^z_i(q|\nu)\,,
\eq
and
\beq\label{t224}
F_{ij}^0(q):=\partial_{q}R_{ij}^z(q)\Big|_{z=0}\,, \quad\tilde{F}_{ij}^0(q):=\partial_{q}
\tilde{R}_{ij}^z(q)\Big|_{z=0}\,,\quad\tilde{Y}^0_i(q):=\partial_q
\tilde{K}^z_i(q|\nu)\Big|_{z=0}\,.
\eq
The local expansions (\ref{t202}), (\ref{t203}), and (\ref{t217})
imply that these derivatives are regular at $z=0$. They also give
\beq\label{t2241}
F_{ij}^0(q)=\p_qr_{ij}(q)\,, \quad\tilde{F}_{ij}^0(q)=\partial_{q}
\tilde{r}_{ij}(q)\,,\quad\tilde{Y}^0_i(q)=\partial_q
\tilde{K}^{(0)}_i(q|\nu)\,.
\eq
and from (\ref{t2161})
\beq\label{t2242}
F_{ij}^0(q)=F_{ji}^0(-q)\,, \quad\tilde{F}_{ij}^0(q)=\tilde{F}_{ji}^0(q)\,.
\eq

We next list identities used to prove the Lax equation in Section \ref{sec4}.
\begin{corollary}
The following identities hold:
\beq\label{t225}
\begin{array}{c}
  \displaystyle{
    R^{z}_{ij}F^{z}_{jk}-F_{ij}^{z}R^{z}_{jk}=F^{0}_{jk}R^{z}_{ik}-R^{z}_{ik}F^{0}_{ij}\,,
  }
\end{array}
\eq
\beq\label{t226}
\begin{array}{c}
  \displaystyle{
    R^{z}_{ij}\ti{F}^{z}_{jk}-F^{z}_{ij}\ti{R}^{z}_{jk}=\ti{F}^{0}_{jk}\ti{R}^{z}_{ik}
    -\ti{R}^{z}_{ik}F^{0}_{ij}\,,
  }
\end{array}
\eq
\beq\label{t227}
\begin{array}{c}
  \displaystyle{
    \ti{R}^{z}_{ij}F^{-z}_{jk}+\ti{F}^{z}_{ij}R^{-z}_{jk}=F^{0}_{jk}\ti{R}^{z}_{ik}
    -\ti{R}^{z}_{ik}\ti{F}^{0}_{ij}\,,
  }
\end{array}
\eq
\beq\label{t228}
\begin{array}{c}
  \displaystyle{
    \ti{R}^{z}_{ij}\ti{F}^{-z}_{jk}+\ti{F}_{ij}^{z}\ti{R}^{-z}_{jk}=\ti{F}^{0}_{jk}R^{z}_{ik}
    -R^{z}_{ik}\ti{F}^{0}_{ij}\,,
  }
\end{array}
\eq
 and
\beq\label{t229}
\begin{array}{c}
  \displaystyle{
    \ti{Y}^{0}_{j}\ti{R}^{z}_{ij}-\ti{R}^{z}_{ij}\ti{Y}^{0}_{i}=-2F^{z}_{ij}\ti{K}^{z}_{j}
    +2\ti{K}^{z}_{i}(q_{i})F^{-z}_{ij}+R^{z}_{ij}\ti{Y}^{z}_{j}+\ti{Y}^{z}_{i}R^{-z}_{ij}\,,
  }
\end{array}
\eq
\beq\label{t230}
\begin{array}{c}
  \displaystyle{
    \ti{Y}^{0}_{j}R^{z}_{ij}-R^{z}_{ij}\ti{Y}^{0}_{i}=2\ti{F}^{z}_{ij}\ti{K}^{-z}_{j}
    +2\ti{K}^{z}_{i}\ti{F}^{-z}_{ij}+\ti{R}^{z}_{ij}\ti{Y}^{-z}_{j}
    +\ti{Y}^{z}_{i}\ti{R}^{-z}_{ij}\,,
  }
\end{array}
\eq
where we use the abbreviations $R^{w}_{ij} := R^{w}_{ij}(q_{i}-q_{j}),~
 \ti{R}^{w}_{ij}:=\ti{R}^{w}_{ij}(q_{i}+q_{j}),~\ti{K}^{w}_{i}:=\ti{K}^{w}_{i}(q_{i}),~F^{w}_{ij}:
 =F^{w}_{ij}(q_{i}-q_{j})$ and $\ti{F}^{w}_{ij}:=
 \ti{F}^{w}_{ij}(q_{i}+q_{j}),~\ti{Y}^{w}_{i}:=\ti{Y}^{w}_{i}(q_{i}).$
\end{corollary}
These are matrix analogues of the elliptic identities (\ref{a15})
and (\ref{a53}). They follow by differentiating the
${\rm BC}_n$-type AYBEs (\ref{t208})--(\ref{t212});
see \cite{MMZ,MMZ2} for details.

We rewrite (\ref{t225})--(\ref{t228}) in the form used below:
\beq\label{b23}
\begin{array}{c}
  \displaystyle{
    R^{z}_{12}(q_{ij})F^{z}_{23}(q_{jk})-F_{12}^{z}(q_{ij})R^{z}_{23}(q_{jk})
    =F^{0}_{23}(q_{jk})R^{z}_{13}(q_{ik})-R^{z}_{13}(q_{ik})F^{0}_{12}(q_{ij})\,,
  }
\end{array}
\eq
\beq\label{b03}
\begin{array}{c}
  \displaystyle{
    R^{z}_{12}(q_{ij})\ti{F}^{z}_{23}(q_{jk}^{+})-F^{z}_{12}(q_{ij})\ti{R}^{z}_{23}(q_{jk}^{+})
    =\ti{F}^{0}_{23}(q_{jk}^{+})\ti{R}^{z}_{13}(q_{ik}^{+})
    -\ti{R}^{z}_{13}(q_{ik}^{+})F^{0}_{12}(q_{ij})\,,
  }
\end{array}
\eq
\beq\label{b04}
\begin{array}{c}
  \displaystyle{
    \ti{R}^{z}_{12}(q_{ij}^{+})F^{-z}_{23}(q_{jk})+\ti{F}^{z}_{12}(q_{ij}^{+})R^{-z}_{23}(q_{jk})
    =F^{0}_{23}(q_{jk})\ti{R}^{z}_{13}(q_{ik}^{+})-\ti{R}^{z}_{13}(q_{ik}^{+})
    \ti{F}^{0}_{12}(q_{ij}^{+})\,,
  }
\end{array}
\eq
\beq\label{b24}
\begin{array}{c}
  \displaystyle{
    \ti{R}^{z}_{12}(q_{ij}^{+})\ti{F}^{-z}_{23}(q_{jk}^{+})+\ti{F}_{12}^{z}(q_{ij}^{+})
    \ti{R}^{-z}_{23}(q_{jk}^{+})=\ti{F}^{0}_{23}(q_{jk}^{+})R^{z}_{13}(q_{ik})
    -R^{z}_{13}(q_{ik})\ti{F}^{0}_{12}(q_{ij}^{+})\,.
  }
\end{array}
\eq
%

\paragraph{Degenerate identities.}
We first consider degenerations of (\ref{b03}) and (\ref{b04}).
Substituting $q_{k}=q_{j}-\epsilon$ into (\ref{b04}) and taking
$\epsilon\rightarrow 0$ gives
\beq\label{b07}
\begin{array}{c}
  \displaystyle{
    \ti{F}_{12}^{z}(q_{ij}^{+})R^{-z,(0)}_{23} =r^{(1)}_{23}\ti{R}^{z}_{13}(q_{ij}^{+})
    -\ti{R}^{z}_{13}(q_{ij}^{+})\ti{F}^{0}_{12}(q_{ij}^{+})
    -\ti{R}^{z}_{12}(q_{ij}^{+})R^{-z,(1)}_{23}-\frac{1}{2}\,\partial_{x}^{2}
    \ti{R}^{z}_{12}(x)\Big|_{x=q_{ij}^{+}}P_{23}\,.
  }
\end{array}
\eq
Similarly, setting $q_{j}=q_{i}-\epsilon$ in (\ref{b03}) and taking the
limit as $\epsilon$ tends to zero gives
\beq\label{b08}
\begin{array}{c}
  \displaystyle{
    R^{z,(0)}_{12}\ti{F}^{z}_{23}(q_{ik}^{+})=R^{z,(1)}_{12}\ti{R}^{z}_{23}(q_{ik}^{+})
    +\ti{F}^{0}_{23}(q_{ik}^{+})\ti{R}_{13}^{z}(q_{ik}^{+})
    - \ti{R}^{z}_{13}(q_{ik}^{+})r^{(1)}_{12}+\frac{1}{2}\,\partial^{2}_{x}
    \ti{R}^{z}_{13}(x)\Big|_{x=q^{+}_{ik}}P_{12}\,.
  }
\end{array}
\eq
This is a matrix analogue of the elliptic identity (\ref{a173}).
Analogous degenerations of (\ref{b23}) and (\ref{b24}) give
\beq\label{b27}
\begin{array}{c}
  \displaystyle{
    R_{12}^{z,(0)}F_{23}^{z}(q_{ij}) =R^{(1),z}_{12}R_{23}^{z}(q_{ij})
    +F^{0}_{23}(q_{ij})R^{z}_{13}(q_{ij}) -R^{z}_{13}(q_{ij})r^{(1)}_{12}+\frac{1}{2}\,
    \partial_{x}^{2}R^{z}_{13}(x)\Big|_{x=q_{ij}}P_{12}\,,
  }
\end{array}
\eq
\beq\label{b28}
\begin{array}{c}
  \displaystyle{
    F^{z}_{12}(q_{ij})R_{23}^{(0),z}=R^{z}_{12}(q_{ij})R^{(1),z}_{23}-r^{(1)}_{23}R^{z}_{13}(q_{ij})
    +R^{z}_{13}(q_{ij})F^{0}_{12}(q_{ij})+\frac{1}{2}\,P_{23}
    \partial_{x}^{2}R^{z}_{13}(x)\Big|_{x=q_{ij}}\,.
  }
\end{array}
\eq
Further degenerations of (\ref{b23}) and (\ref{b24}) yield
\beq\label{b38}
\begin{array}{c}
  \displaystyle{
    R_{12}^{z}(q)F_{23}^{z}(-q)-F_{12}^{z}(q)R_{23}^{z}(-q) =[F_{32}^{0}(q),r_{13}(z)]P_{13}
    -\partial_{q}F_{32}^{0}(q)P_{13}\,,
  }
\end{array}
\eq
\beq\label{b39}
\begin{array}{c}
  \displaystyle{
    \ti{R}_{12}^{z}(q_{ij}^{+})\tilde{F}_{23}^{-z}(q_{ji}^{+})+\ti{F}_{12}^{z}(q_{ij}^{+})
    \ti{R}_{23}^{-z}(q_{ji}^{+}) =
\ti{F}^{0}_{23}(q_{ji}^{+})R_{13}^{(0),z}-\partial_{q_{ji}^{+}}
    \ti{F}^{0}_{23}(q_{ji}^{+})P_{13}-R_{13}^{(0),z}\ti{F}^{0}_{12}(q_{ij}^{+})\,.
  }
\end{array}
\eq
Relation (\ref{b39}) follows by substituting $q_{k}=q_{i}-\epsilon$
into (\ref{b24}) and taking $\epsilon\rightarrow 0$.

We also use the following degeneration of (\ref{t229}), obtained by
setting $q_{j}=q_{i}-\epsilon$ and taking $\epsilon\rightarrow 0$:
\beq\label{b20}
\begin{array}{c}
  \displaystyle{
    \ti{Y}^{0}_{2}(q_{i})\ti{R}^{z}_{12}(2q_{i})-\ti{R}^{z}_{12}(2q_{i})\ti{Y}^{0}_{1}(q_{i}) =
  }
  \\ \ \\
  \displaystyle{
    = R^{z,(0)}_{12}\ti{Y}^{z}_{2}(q_{i})+\ti{Y}_{1}^{z}(q_{i})R^{-z,(0)}_{12}
    +2\ti{K}^{z}_{1}(q_{i})R^{-z,(1)}_{12}-2R^{z,(1)}_{12}\ti{K}^{z}_{2}(q_{i})\,.
  }
\end{array}
\eq
Setting $q_{j}=q_{i}-\epsilon$ in (\ref{t230}) and taking
$\epsilon\rightarrow 0$ gives
\beq\label{b40}
\begin{array}{c}
  \displaystyle{
    \ti{Y}^{0}_{2}(q_{i})R^{z,(0)}_{12}-\partial_{q_{i}}\ti{Y}^{0}_{2}(q_{i})P_{12}
    -R^{z,(0)}_{12}\ti{Y}^{0}_{1}(q_{i}) =
  }
  \\ \ \\
  \displaystyle{
    =2\ti{F}^{z}_{12}(2q_{i})\ti{K}^{-z}_{2}(q_{i})+2\tilde{K}^{z}_{1}(q_{i})
    \ti{F}^{-z}_{12}(2q_{i})+\ti{R}^{z}_{12}(2q_{i})\ti{Y}^{-z}_{2}(q_{i})+\ti{Y}^{z}_{1}(q_{i})
    \ti{R}^{-z}_{12}(2q_{i})\,.
  }
\end{array}
\eq

We also use the degeneration of (\ref{t208}),
\beq\label{c07}
\begin{array}{c}
  \displaystyle{
    R_{12}^{z}(q)R^{z}_{23}(-q)=\Big(r_{13}(z)r_{32}(q)-r_{32}(q)r_{13}(z)\Big)P_{13}
    -F^{0}_{13}(z)P_{13}+F^{0}_{32}(q)P_{13}\,,
  }
\end{array}
\eq
the degeneration of (\ref{t210}),
\beq\label{c08}
\begin{array}{c}
  \displaystyle{
    \ti R_{12}^{z}(q)\ti R_{23}^{-z}(q) =\Big(\ti r_{23}(q)r_{13}(z)
    -r_{13}(z)\ti r_{32}(q)\Big)P_{13} -\ti F_{32}^{0}(q)P_{13} +F_{13}^{0}(z)P_{13}
  }
\end{array}
\eq
and the degeneration of the four-term relation (\ref{t212}):
\beq\label{c10}
\begin{array}{c}
  \displaystyle{
    \ti R_{12}^{z}(2q)\ti K_{2}^{-z}(q) +\ti K_{1}^{z}(q)
    \ti R_{12}^{-z}(2q)=R_{21}^{-z,(0)}\ti K_{1}^{(0)}(q) -\ti K_{2}^{(0)}(q)R_{21}^{-z,(0)}-
  }
  \\ \ \\
  \displaystyle{
    -P_{21}\partial_q\ti K_{1}^{(0)}(q)+
\partial_{x}R_{21}^{x,(0)}\Big|_{x=-z}\ti K_{1}^{(-1)}(q) +\ti K_{2}^{(-1)}(q)
    \partial_{x}R_{21}^{x,(0)}\Big|_{x=-z}\,.
  }
\end{array}
\eq
Equivalently,
\beq\label{c101}
\begin{array}{c}
  \displaystyle{
    \ti R_{12}^{z}(2q)\ti K_{2}^{-z}(q) +\ti K_{1}^{z}(q)
    \ti R_{12}^{-z}(2q)=\Big( [{\ti K}_2^{(0)}(q),r_{12}(z)]-{\ti Y}^0_2(q)+2\varsigma_N F_{12}^0(z) \Big)P_{12}\,,
  }
\end{array}
\eq
where $\varsigma_N$ is the coefficient in (\ref{t217}).

Finally, we use the matrix analogue of (\ref{a171})
\beq\label{c02}
\begin{array}{c}
  \displaystyle{
    r_{12}(z)r_{13}(z+w)-r_{23}(w)r_{12}(z) +r_{13}(z+w)r_{23}(w) =
  }
  \\ \ \\
  \displaystyle{
    =m_{12}(z)+m_{23}(w)+m_{13}(z+w)\,.
  }
\end{array}
\eq
At $w=0$, this becomes
\beq\label{c03}
\begin{array}{c}
  \displaystyle{
    r_{12}(z)r_{13}(z)={}r_{23}^{(0)}r_{12}(z)-r_{13}(z)r_{23}^{(0)}-F_{13}^{0}(z)P_{23}
    +m_{12}(z)+m_{23}(0)+m_{13}(z)\,.
  }
\end{array}
\eq

\subsection{Elliptic formulae}
\paragraph{Matrix basis \texorpdfstring{$T_a$}{T(a)}.}
In the elliptic case, we use the following basis of $\Mat$.\footnote{We
frequently use identities from \cite{MMZ2}, where a slightly different
definition of $T_\al$ was adopted. Most relations retain their form,
but some identities change, including the Fourier (argument) symmetry
of the $K$-matrix in (\ref{t2531})--(\ref{t2532}).
See the remark at the end of Section 3 of \cite{MMZ2}.}
\beq\label{t231}\begin{array}{c} \displaystyle{
        T_\al = \exp \left( \al_1 \al_2 \frac{\pi i}{N} \right) Q^{\al_1} \Lambda^{\al_2},
        \quad \al = (\al_1, \al_2)\in \mZ_N \times \mZ_N, \quad T_0=T_{(0,0)} = 1_N\,,
}\end{array}\eq
Here $\mZ_N=\mZ/N\mZ$, and $Q\in{\rm Mat}_N$ and
$\Lambda\in{\rm Mat}_N$ are defined by
 \beq\label{t232}
 \begin{array}{c}
  \displaystyle{
(Q)_{kl}=\delta_{kl}\exp\Big(\frac{2\pi\imath}{{ N}}\,k\Big)\,,
\qquad(\Lambda)_{kl}=\delta_{k-l+1=0\,{\hbox{\tiny{mod}}}\,{ N}}\,,
 }
 \end{array}
 \eq
Thus $Q$ is diagonal, with the $N$th roots of unity as its entries,
and $\Lambda$ acts by cyclic permutation:
 \beq\label{t233}
 \Lambda=\left(
 \begin{array}{cccccc}
 0 & 1 & 0 & 0 &...& 0
 \\
 0 & 0 & 1 & 0 &...& 0
 \\
 \vdots & \vdots & \vdots & \dots & \vdots & \vdots
  \\
 0 & 0 & ... & 0 & 1 & 0
 \\
 0 & 0 & ... & 0 & 0 & 1
 \\
  1 & 0 & ... & 0 & 0 & 0
 \end{array}
 \right)\,,\qquad \Lambda^{-1}=\Lambda^{t}\,,
 \eq
where $t$ denotes matrix transposition.
These matrices define a finite-dimensional representation of the
Heisenberg group:
 \beq\label{t234}
 \begin{array}{c}
  \displaystyle{
    \exp \left( \al_1 \al_2 \frac{2\pi i}{N} \right) Q^{\al_1} \Lambda^{\al_2}
    =\Lambda^{\al_2} Q^{\al_1}\,,
}
\end{array}
\eq
\beq\label{t235}
\begin{array}{c}
\displaystyle{
    T_\al T_\be = \ka_{\al, \be} T_{\al + \be}, \quad \ka_{\al, \be}
    =\exp \left( \frac{\pi i}{N}\,(\al_2 \be_1 - \al_1 \be_2) \right),
    \quad \ka_{\al, \al + \be} = \ka_{\al, \be}, \quad \ka_{- \al, \be} = \ka_{\be,\al}\,,
}
\end{array}
\eq
\beq\label{t236}
\begin{array}{c}
 \displaystyle{
    \tr \Big(T_\al T_{-\be}\Big) = N \de_{\al,\be}\,,
}
\end{array}
\eq
\beq\label{t237}
\begin{array}{c}
 \displaystyle{
    [T_\al, T_\be] = (\ka_{\al, \be} - \ka_{\be, \al}) T_{\al+ \be} =2\imath
    \sin \left( \frac{\pi}{N}\,(\al_2 \be_1 - \al_1 \be_2) \right) T_{\al +\be}\,.
}
\end{array}
\eq
The matrix permutation operator $P_{12}$ (\ref{t216}) takes the form
\begin{equation}
\begin{array}{c}
\label{t238}
     \displaystyle{
    P_{12}=\sum\limits_{k,l=1}^N e_{kl}\otimes e_{lk}=\frac{1}{N}\,\sum
    \limits_{\al\in\,\mZ_N\times\mZ_N}T_\al\otimes T_{-\al}\in\Mat^{\otimes2}\,.
    }
\end{array}
 \end{equation}
The change of basis between $T_a$ and the standard matrix units $e_{ij}$
is described, for example, in the appendix of \cite{ZZ}.

By (\ref{t236}), the coefficient $X_\al$ of a matrix $X\in\Mat$
in the basis $T_a$ is
\beq\label{t2381}
\begin{array}{c}
 \displaystyle{
    X_\al=\frac{1}{N}\,\tr\Big(X T_{-\al}\Big)\,.
}
\end{array}
\eq
The basis $T_\al$ in (\ref{t231}) is not periodic on
$\mZ_N\times\mZ_N$, since
$T_{(\al_1+N,\al_2)}=(-1)^{\al_2}T_{(\al_1,\al_2)}$ and $T_{(\al_1,\al_2+N)}=(-1)^{\al_1}T_{(\al_1,\al_2)}$.

\paragraph{\texorpdfstring{$R$}{R}-matrices.} The elliptic Baxter--Belavin $R$-matrix \cite{BB,Belavin} has the
following form:
\begin{equation}\label{t239}
\begin{array}{c}
    \displaystyle{
    R^{z}_{12} (x) =\frac{1}{N}\, \sum\limits_{\al\in\mZ_N\times\mZ_N}
    \varphi_\al (x, \frac{z}{N} +\om_\al) T_\al \otimes T_{-\al}
    }\,,
\end{array}
\end{equation}
where $\vf_\al$ are the functions defined in (\ref{a32}).

The expansion near $z = 0$ is the quasiclassical expansion,
since $z$ plays the role of the Planck constant:
\begin{equation}\label{t240}
\begin{array}{c}
\displaystyle{
    R^{z}_{12}(x) = \frac{1}{z}\, 1_N \otimes 1_N + r_{12}(x) + z \: m_{12} (x) +O(z^2)\,,
    }
\end{array}
\end{equation}
where $r_{12}$ is the Belavin--Drinfeld classical $r$-matrix \cite{BD},
which satisfies the classical Yang--Baxter equation:
\begin{equation}\label{t241}
\begin{array}{c}
    \displaystyle{
    [r_{12}, r_{13}] + [r_{12}, r_{23}] + [r_{13}, r_{23}] = 0, \quad r_{ij} =r_{ij}(q_i-q_j)\,.
    }
\end{array}\end{equation}
The classical $r$-matrix and the next coefficient, the $m$-matrix, are
\begin{equation}\label{t242}
\begin{array}{c}
    \displaystyle{
    r_{12} (x) = \frac{1}{N}\, E_1(x)\, 1_N \otimes 1_N +\frac{1}{N}\,
    \sum_{\substack{\al\neq(0,0)\\ \al\in\,\mZ_N\times\mZ_N}}
    \varphi_\al (x, \om_\al) T_\al \otimes T_{-\al}
    }\,,
\end{array}
\end{equation}
\beq\label{t243}
\begin{array}{c}
    \displaystyle{
     m_{12} (x)= \frac{1}{N^2}\, \frac{E_1^2(x)-\wp(x)}{2}\, 1_N \otimes 1_N
     +\frac{1}{N^2}\, \sum_{\substack{\al\neq(0,0)\\ \al\in\,\mZ_N\times\mZ_N}}f_\al
     (x, \om_\al) T_\al
     \otimes T_{-\al}\,,
     }
\end{array}
\eq
\beq\label{t2431}
\begin{array}{c}
    \displaystyle{
     m_{12} (0)= \frac{\vth'''(0)}{3N\vth'(0)}\, P_{12}
     -\frac{1}{N^2}\, \sum_{\substack{\al\neq(0,0)\\ \al\in\,\mZ_N\times\mZ_N}} \wp(\om_\al) T_\al\otimes T_{-\al}\,.
     }
\end{array}
\eq
The local expansion (\ref{t215}) of the classical $r$-matrix gives
\begin{equation}\label{t244}
\begin{array}{c}
    \displaystyle{ r^{(0)}_{12} = \frac{1}{N}\,
    \sum_{\substack{\al\neq(0,0)\\ \al\in\,\mZ_N\times\mZ_N}}(E_1 (\om_\al) + 2 \pi i \partial_\tau \om_\al )
    T_\al
    \otimes T_{-\al}}\,.
\end{array}
\end{equation}
The coefficients in (\ref{t203}) and (\ref{t213})--(\ref{t215}) can be
computed directly from the expansion (\ref{a10}) of $\phi$.
Alternatively, if the quantum $R$-matrix satisfies (\ref{t221}), one can
use (\ref{t222}). For simplicity, we assume (\ref{t221})--(\ref{t222})
throughout the following construction. This assumption is not essential:
the constructions can also be carried out without Fourier symmetry
(\ref{t221}).

Using (\ref{a07}), the derivative of (\ref{t239}) is
\begin{equation}\label{t2391}
\begin{array}{c}
    \displaystyle{
    F^{z}_{12} (x) =\p_x R^{z}_{12} (x) =
    }
    \\[4mm]
     \displaystyle{
    =\frac{1}{N}\, \sum\limits_{\al\in\mZ_N\times\mZ_N}
    \varphi_\al (x, \frac{z}{N} +\om_\al)\Big(E_1(x+\frac{z}{N} + \om_\al)-E_1(x)
    +2\pi\imath\p_\tau\om_\al\Big)T_\al \otimes T_{-\al}
    }\,
\end{array}
\end{equation}
and
\begin{equation}\label{t2392}
\begin{array}{c}
    \displaystyle{
    F^{0}_{12} (x) =\p_x r_{12} (x) =
    }
    \\[4mm]
     \displaystyle{
    =-\frac{1}{N}\,E_2(x)1_N\otimes 1_N+\frac{1}{N}\,
    \sum_{\substack{\al\neq(0,0)\\\al\in\,\mZ_N\times\mZ_N}}
    \varphi_\al (x, \om_\al)\Big(E_1(x+ \om_\al)-E_1(x)+2\pi\imath\p_\tau
    \om_\al\Big)T_\al \otimes T_{-\al}
    }\,.
\end{array}
\end{equation}

Introduce the matrix
\beq\label{t245}
J=e_{1,N-1}+e_{2,N-2}+\dots+e_{N-1,1}+e_{N,N}=\frac{1}{2N}\,\sum
\limits_{a\in\,\mZ_{2N}\times\mZ_{2N}} T_a\in\Mat\,,
\eq
or
 \beq\label{t246}
 J=\left(
 \begin{array}{cccccc}
 0 & ... & 0 & 0 & 1 & 0
 \\
 0 & ... & 0 & 1 & 0& 0
 \\
 \vdots & \vdots & \vdots & \dots & \vdots & \vdots
  \\
 0 & 1 & 0 & ... & 0 & 0
 \\
 1 & 0 & 0 & ... & 0 & 0
 \\
 0 & 0 & 0 & ... & 0 & 1
 \end{array}
 \right)\,,\qquad J^2=1_N\,,
 \eq
Conjugation by $J$ realizes the involutive automorphism
$\mathrm{j}$ of ${\rm End}(\mC^N)$,
\begin{equation}\label{t247}
\mathrm{j}(A)= J\, A\, J^{-1}\,,\qquad\mathrm{j}(T_\alpha)=T_{-\alpha}\,.
\end{equation}
The matrix $\tilde{R}_{12}^\hbar(\xi)$ is then defined as the operator
$(\mathrm{id}\otimes\mathrm{j})R_{12}^\hbar(\xi)$:
  \begin{equation}\label{t248}
\begin{array}{c}
    \displaystyle{
    \tilde{R}^{z}_{12} (x):=J_2\,     R^{z}_{12} (x)\,J_2 = \frac{1}{N}\,
    \sum_{\al \in \mZ_N \times \mZ_N} \varphi_\al (x, \frac{z}{N} + \om_\al)\,T_\al
    \otimes T_{\al}}\,\in\Mat^{\otimes 2}\,,
\end{array}
\end{equation}
where $J_2=1_N\otimes J$. In the notation of the expansions
(\ref{t214})--(\ref{t215}),
 \beq\label{t249}
 \begin{array}{c}
  \displaystyle{
\res\limits_{x=0}\tilde{R}^{z}_{12}(x)=\tilde{R}^{(-1)}_{12}=P_{12} J_1 J_2\,,
\qquad\res\limits_{x=0}\tilde{r}_{12}(x)=\tilde{r}^{(-1)}_{12}=P_{12} J_1 J_2\,,
  }
 \end{array}
 \eq
 and
 \beq\label{t250}
 \begin{array}{c}
  \displaystyle{
\res\limits_{z=0} \tilde{R}^{z}_{12}(x)=1_N\otimes 1_N\,.
}
 \end{array}
 \eq
The expansion coefficients of $\tilde{R}^{z}_{12}(x)$ are computed
in the same way as those of $R^{z}_{12}(x)$:
by (\ref{t248}), one replaces $T_{-\al}$ with $T_\al$ in the
second tensor component. For example,
\begin{equation}\label{t2501}
\begin{array}{c}
    \displaystyle{
    {\ti r}_{12} (x) = \frac{1}{N}\, E_1(x)\, 1_N \otimes 1_N +\frac{1}{N}\,
    \sum_{\substack{\al\neq(0,0)\\ \al\in\,\mZ_N\times\mZ_N}} \varphi_\al (x,\om_\al) T_\al
    \otimes T_{\al}\,.
    }\,
\end{array}
\end{equation}

The Baxter--Belavin $R$-matrix satisfies the QYBE (\ref{t2011}) and
the AYBE (\ref{t208}). Further properties of this matrix and of the
coefficients in (\ref{t202}) and (\ref{t213})--(\ref{t215}) can be found
in \cite{Pol,LOZ14,LOZ15}.

\paragraph{\texorpdfstring{$K$}{K}-matrix.}
We use the elliptic $K$-matrix introduced in \cite{MMZ2}:
\begin{equation}\label{t251}
    \tilde{K}^{z}(x) = \frac{1}{2N}\,\sum_{\be \in \mZ_{2N} \times \mZ_{2N}}
    \exp(\frac{2 \pi \imath \be_2 x}{N}) v (x, \frac{z}{N} + \frac{\om_\be}{2})T_\be\in\Mat \,,
\end{equation}
where $v$ is defined in (\ref{a45}). Equivalently,
\beq\label{t252}
\begin{array}{c}
\displaystyle{
\tilde{K}^{z}(x)=\sum_{a=0}^3\nu_a\,\tilde{K}_a^{z}(x)\,,\quad\tilde{K}_a^{z}(x)
=\frac{1}{2N}\,\sum_{\be\in\mZ_{2N}\times\mZ_{2N}}\exp\!\Big(\frac{2\pi\imath\,\be_2 x}{N}\Big)
\vf_a\!\Big(2x,\,\frac{z}{N}+\frac{\om_\be}{2}+\Om_a\Big)T_\be\,.
}
\end{array}
\eq
This expression is related to the Komori--Hikami $K$-matrix \cite{KH2}
by a simple transformation; see \cite{MMZ2} for details.
It satisfies the following Fourier symmetry:
\beq\label{t2531}
\begin{array}{c}
\displaystyle{
\tilde{K}^{z}(x)=\tilde{K}^{x}(z)J\qquad\text{(even $N$)}\,,
}
\end{array}
\eq
\beq\label{t2532}
\begin{array}{c}
\displaystyle{
\tilde{K}^{z}(x|\nu_0,\nu_1,\nu_2,\nu_3)=\tilde{K}^{x}(z|\nu_0,\nu_1,-\nu_2,\nu_3)J
\qquad\text{(odd $N$)}\,.
}
\end{array}
\eq

The unitarity relation (\ref{t220}) takes the form
 \beq\label{t254}
   \displaystyle{
 \tilde{K}^z_i(x) \tilde{K}^{-z}_i(x)=1_N\sum_{a=0}^3\Big(\bar\nu_a^2\wp(x+\Omega_a)
 -\bar{\nu'}_a^{\,2}\wp(z+\Omega_a)\Big)\,,
 }
 \eq
where the parameters $\nu'$ depend on the parity of $N$:
\beq\label{t255}
\begin{array}{c}
\displaystyle{
\begin{aligned}
N\text{ even:}&\qquad\nu'_k=\nu_k\ \hbox{for all}\ k=0,1,2,3\,,
\\ \ \\
N\text{ odd:}&\qquad\nu'_2=-\nu_2,\;\nu'_k=\nu_k\text{ for }k= 0,1,3\,.
\end{aligned}
}
\end{array}
\eq
In (\ref{t254}), $\bar\nu$ and $\bar{\nu'}$ are obtained from $\nu$
and $\nu'$, respectively, by the linear transformation (\ref{a48}).
Comparing with the general unitarity relation (\ref{t220}), we obtain
 \beq\label{t256}
   \displaystyle{
 u_N(q|\nu)=\sum_{a=0}^3\bar\nu_a^2\wp(q+\Omega_a)\,,\qquad w_N(z|\nu)=\sum_{a=0}^3
 \bar{\nu'}_a^{\,2}\wp(z+\Omega_a)\,.
 }
 \eq

The residue at $z=0$ in the expansion (\ref{t217}) of the $K$-matrix is
\beq\label{t2533}
\begin{array}{c}
\displaystyle{
\res\limits_{z=0}\tilde{K}^{z}(x)=\bar\nu'_0 1_N\,,\qquad \varsigma_N=\bar\nu'_0\,,
}
\end{array}
\eq
where $\nu'$ is defined in (\ref{t255}), and $\bar\nu'$ is obtained
from (\ref{a48}); explicitly,
$\bar\nu'_0=(\nu_0+\nu_1+(-1)^N\nu_2+\nu_3)/2$.

\paragraph{Case \texorpdfstring{$N=2$}{N = 2}.}
The matrices (\ref{t232})--(\ref{t233}) are $Q=-\si_3$ and
$\Lambda=\sigma_1$, where the Pauli matrices are
\beq\label{t259}
 \begin{array}{c}
  \displaystyle{
 \sigma_0=\mats{1}{0}{0}{1}\,,\quad
  \sigma_1=\mats{0}{1}{1}{0}\,,\quad
  \sigma_2=\mats{0}{-\imath}{\imath}{0}\,,\quad
   \sigma_3=\mats{1}{0}{0}{-1}
  }
 \end{array}
\eq
and $\sigma_4=\sigma_0=1_{2\times 2}$. Therefore, $T_{(1,0)}=-\si_3$, $T_{(1,1)}=\si_2$, $T_{(0,1)}=\si_1$.
In this case, (\ref{t239}) is Baxter's eight-vertex $R$-matrix:
\beq\label{t260}
 \begin{array}{c}
  \displaystyle{
 R_{12}^z(x)
 =\frac12\sum\limits_{k=0}^3\vf_{k}\Big(x,\Om_{k}+\frac{z}{2}\Big)\sigma_{4-k}\otimes\sigma_{4-k}\,,
  }
 \end{array}
\eq
The corresponding reflection $K$-matrix was obtained in \cite{IK,Vega}.
We use the form given in \cite{MMZ2}:
  \beq\label{t261}
 \begin{array}{c}
  \displaystyle{
 {\tilde K}^z(x)=\sum\limits_{k=0}^3\rho_{k}
 e^{-2\pi\imath z\p_\tau\Om_{k}}\vf_{k}(x-\Om_{k},z+\Om_{k})\sigma_{4-k}=
}
\\ \ \\\
  \displaystyle{
 =
 \sum\limits_{k=0}^3\rho_{k}
 e^{2\pi\imath (x+z+\Om_k)\p_\tau\Om_{k}}\phi(x+\Om_{k},z+\Om_{k})\sigma_{4-k}\,,
  }
 \end{array}
\eq
 where $\rho_1=-\bar\nu_1$, $\rho_2=\bar\nu_2$, $\rho_3=\bar\nu_3$,
 $\rho_4=\rho_0=\bar\nu_0$ is assumed, and $\p_\tau\Om_{k}$ is equal to either $0$ (for $\Om_0$, $\Om_1$) or
 $1/2$ (for $\Om_2$, $\Om_3$).

When $N=2$, $\tilde{R}^\hbar_{12}(z)=R^\hbar_{12}(z)$ due to $J=\sigma_0$.

\paragraph{Case \texorpdfstring{$N=1$}{N = 1}.}
All $R$- and $K$-matrices become scalar functions.
In particular, the $R$-matrix reduces to the elliptic Kronecker function:
\beq\label{t263}
 \begin{array}{c}
  \displaystyle{
 R_{12}^z(x)\Big|_{N=1}=\phi(z,x)\,,\qquad r_{12}(x)\Big|_{N=1}=E_1(x)\,,
   }
    \\ \ \\
    \displaystyle{
  F_{12}^0(x)\Big|_{N=1}=f(0,x)=-E_2(x)=-\wp(x)+\frac{\vth'''(0)}{3\vth'(0)}\,,
  }
    \\ \ \\
    \displaystyle{
m_{12}(x)\Big|_{N=1}=\frac{E_1^2(x)-\wp(x)}{2}\,,\qquad m_{12}(0)\Big|_{N=1}=m(0)=\frac{\vth'''(0)}{3\vth'(0)}\,,
  }
  \\ \ \\
    \displaystyle{
 {\ti R}_{12}^z(x)\Big|_{N=1}=\phi(z,x)\,,\qquad {\ti r}_{12}(x)\Big|_{N=1}=E_1(x)\,.
  }
 \end{array}
\eq

Put $c_\tau=\vth'''(0)/(3\vth'(0))$. The scalar limits are
$m(0)=c_\tau$, $f(0,x)=-E_2(x)=c_\tau-\wp(x)$ and
$u_1(q|\nu)=\sum_{a=0}^3\la_a^2\wp(q+\Om_a)$, since
$\la_a^2=\bar\nu_a^2$.

The remaining coefficients in (\ref{t202})--(\ref{t215}) follow from
the expansions (\ref{a10})--(\ref{a11}). For the $K$-matrix, we obtain
\beq\label{t264}
 \begin{array}{c}
  \displaystyle{
 {\tilde K}^z(x|\nu)\Big|_{N=1}=v(z,x|\la)\,,
  }
 \end{array}
\eq
where $\la=\la(\nu)$ is the set of constants
\beq\label{t2641}
 \begin{array}{c}
  \displaystyle{
\la_0=\bar\nu_0=\frac12(\nu_0+\nu_1+\nu_2+\nu_3)\,,\qquad
\la_1=\bar\nu_1=\frac12(\nu_0+\nu_1-\nu_2-\nu_3)\,,
 }
 \\ \ \\
  \displaystyle{
\la_2=-\bar\nu_2=-\frac12(\nu_0-\nu_1+\nu_2-\nu_3)\,,\qquad
 \la_3=\bar\nu_3=\frac12(\nu_0-\nu_1-\nu_2+\nu_3)\,.
  }
 \end{array}
\eq
where the constants $\bar\nu$ are defined by (\ref{a48}). Alternatively,
\beq\label{t265}
 \begin{array}{c}
  \displaystyle{
 {\tilde K}^z(x|\nu)\Big|_{N=1}=v(x,z|\mu)\,,
  }
 \end{array}
\eq
where $\mu=\mu(\nu)$ is the set of constants
\beq\label{t266}
 \begin{array}{c}
  \displaystyle{
\mu_0=\frac12(\nu_0+\nu_1-\nu_2+\nu_3)\,,\qquad \mu_1=\frac12(\nu_0+\nu_1+\nu_2-\nu_3)\,,
  }
  \\ \ \\
   \displaystyle{
\mu_2=\frac12(-\nu_0+\nu_1+\nu_2+\nu_3)\,,\qquad \mu_3=\frac12(\nu_0-\nu_1+\nu_2+\nu_3)\,.
  }
 \end{array}
\eq
Finally,
\beq\label{t267}
 \begin{array}{c}
  \displaystyle{
 {\tilde Y}^z(x|\nu)\Big|_{N=1}=\p_x v(z,x|\la)\,,
  }
 \end{array}
\eq
\beq\label{t268}
 \begin{array}{c}
  \displaystyle{
 {\tilde Y}^0(x|\nu)\Big|_{N=1}=-\sum\limits_{k=0}^3\la_k E_2(x+\Om_k)\stackrel{(\ref{a08})}{=}
 -\sum\limits_{k=0}^3\la_k \wp(x+\Om_k)+\frac{\vth'''(0)}{3\vth'(0)}\sum\limits_{k=0}^3\la_k\,.
  }
 \end{array}
\eq

\section{Models associated with the type \texorpdfstring{$A$}{A} root system}\label{sec3}
\setcounter{equation}{0}

\subsection{Calogero--Moser model and its spin generalization}
\paragraph{Spinless Calogero--Moser model.}
The spinless model describes $n$ particles in the complex plane,
with dynamics generated by the Hamiltonian \cite{Ca,OP}
  \beq\label{t301}
  \begin{array}{c}
  \displaystyle{
H^{\hbox{\tiny{CM}}}=\sum\limits_{i=1}^n\frac{p_i^2}{2}-g^2\sum\limits_{i>j}^n\wp(q_i-q_j)\,,
 }
 \end{array}
 \eq
where $g\in\mC$ is a coupling constant and $\wp$ is the Weierstrass $\wp$-function (see Appendix A).
The particle positions $q_i\in\mC$ and momenta $p_i\in\mC$
are canonically conjugate:
  \beq\label{t302}
  \begin{array}{c}
  \displaystyle{
 \{p_i,q_j\}=\delta_{ij}\,,\quad \{p_i,p_j\}=\{q_i,q_j\}=0\,.
 }
 \end{array}
 \eq
The Hamiltonian equations ${\dot f}=\{H,f\}$ admit a Lax representation
with a spectral parameter $z$. The Lax pair (\ref{t101}) for the
spinless elliptic Calogero--Moser model consists of the following
$n\times n$ matrices \cite{Krich1}:
  \beq\label{t303}
  \begin{array}{c}
  \displaystyle{
L^{\hbox{\tiny{CM}}}_{ij}(z)=\delta_{ij}\Big(p_i+{g} E_1(z)\Big)+{g}(1-\delta_{ij})\phi(z,q_{ij})\,,\quad q_{ij}=q_i-q_j\,,
 }
 \end{array}
 \eq
  \beq\label{t304}
  \begin{array}{c}
  \displaystyle{
M^{\hbox{\tiny{CM}}}_{ij}(z)={g}\delta_{ij} d_i
+{g}(1-\delta_{ij})f(z,q_{ij})\,,\quad d_i=\sum\limits_{k:k\neq i}^n\wp(q_{ik})\,,
 }
 \end{array}
 \eq
where $\phi(z,q)$ is the elliptic Kronecker function (\ref{a01}),
and $f(z,q)$ and $E_1(z)$ are defined in (\ref{a07})--(\ref{a08}).
The equations of motion
  \beq\label{t305}
  \begin{array}{c}
  \displaystyle{
 {\dot q}_i=p_i\,,\qquad  {\dot p}_i={g}^2\sum\limits_{k:k\neq
 i}^n\wp'(q_{i}-q_k)\,,\quad i=1,...,n
 }
 \end{array}
 \eq
obtained from (\ref{t301})--(\ref{t302}) follow from the Lax equation
by the identities (\ref{a15})--(\ref{a17}).

\paragraph{Spin Calogero--Moser model.}
In the spin generalization \cite{BAB,GH}, the phase space is extended
by the classical spin component
$\mO^{\hbox{\tiny{$n$}}}//{\mathcal D}$, where
$\mO^{\hbox{\tiny{$n$}}}$ is a coadjoint orbit of the Lie group
${\rm GL}_n$ and ${\mathcal D}\subset{\rm GL}_n$ is its Cartan subgroup.
The dynamical spin variables are the entries of $S\in{\rm Mat}(n,\mC)$;
see \cite{TrZ} for a review.

The model is usually described first on the unreduced phase space,
whose spin component is $\mO^{\hbox{\tiny{$n$}}}$ with the linear
Lie--Poisson structure
  \beq\label{t306}
\begin{array}{c}
  \displaystyle{
\{S_{ij},S_{ml}\}=-S_{il}\delta_{mj}+S_{mj}\delta_{il}\,,\quad i,j,m,l=1,...,n\,,
 }
 \end{array}
 \eq
 or, equivalently,
  \beq\label{t307}
  \begin{array}{c}
  \displaystyle{
\{S_{1},S_{2}\}=[P_{12},S_1]\,.
 }
 \end{array}
 \eq
The Casimir functions are fixed on a coadjoint orbit
$\mO^{\hbox{\tiny{$n$}}}$; in the semisimple case, this amounts to
fixing the eigenvalues of $S$.

The Poisson reduction by the Cartan subgroup $\mathcal D$,
$\mO^{\hbox{\tiny{$n$}}}\rightarrow
\mO^{\hbox{\tiny{$n$}}}//{\mathcal D}$, involves the first-class
constraints $\chi_i={\rm const}$, $i=1,...,n$, where
  \beq\label{t308}
  \begin{array}{c}
  \displaystyle{
 \chi_i=S_{ii}\,,\quad i=1,...,n\,.
 }
 \end{array}
 \eq
These constraints generate the action
  \beq\label{t309}
  \begin{array}{c}
  \displaystyle{
S\to \mD S\mD^{-1}\,,\quad \mD={\rm diag}(d_1,...,d_n)\,.
 }
 \end{array}
 \eq
Together with suitable gauge-fixing conditions, (\ref{t308}) forms a
set of second-class constraints. The reduced Poisson structure is then
obtained from the Dirac bracket. This reduces the dimension of the
unreduced phase space by $2n-2$:
  \beq\label{t310}
  \begin{array}{c}
  \displaystyle{
T^*\mC^n\times\mO^{\hbox{\tiny{$n$}}}\ \stackrel{\hbox{reduction}}{\longrightarrow}\ T^*\mC^n\times\Big(\mO^{\hbox{\tiny{$n$}}}//{\mathcal D}\Big)\,.
 }
 \end{array}
 \eq
Here $T^*\mC^n$ is the particle component of the phase space, with
brackets (\ref{t302}). In the center-of-mass frame, it is replaced by
$T^*\mC^{n-1}$, and the reduced phase space has the same dimension
as the original coadjoint orbit $\mO^{\hbox{\tiny{$n$}}}$.

The matrices
  \beq\label{t311}
  \begin{array}{c}
  \displaystyle{
L^{\hbox{\tiny{spin}}}_{ij}(z)=\delta_{ij}(p_i+S_{ii}E_1(z))+(1-\delta_{ij})S_{ij}\phi(z,q_{ij})\,,
 }
 \end{array}
 \eq
  \beq\label{t312}
  \begin{array}{c}
  \displaystyle{
  M^{\hbox{\tiny{spin}}}_{ij}(z)=(1-\delta_{ij})S_{ij}f(z,q_i-q_j)
 }
 \end{array}
 \eq
satisfy a Lax equation with an additional unwanted term on the
right-hand side:
  \beq\label{t313}
  \begin{array}{c}
  \displaystyle{
 \dot{L}^{\hbox{\tiny{spin}}}(z)=[L^{\hbox{\tiny{spin}}}(z),M^{\hbox{\tiny{spin}}}(z)]+\hbox{``unwanted term''}\,,
 }
 \\ \ \\
  \displaystyle{
 \hbox{``unwanted term''}=-\sum\limits_{i\neq j}^n E_{ij}\,\Big(S_{ii}-S_{jj}\Big)S_{ij}\Big(\frac12\,\p^2_x\phi(z,x)\Big|_{x=q_{ij}}\!+m(z)\phi(z,q_{ij})\Big)\,,
 }
 \end{array}
 \eq
where $E_{ij}$ are matrix units in ${\rm Mat}(n,\mC)$ and $m(z)$
is defined in (\ref{a172}). This equation yields the equations of motion
  \beq\label{t314}
  \begin{array}{c}
  \displaystyle{
{\dot q}_i=p_i\,,\quad {\dot p}_i=\sum\limits_{j:j\neq
i}^n S_{ij}S_{ji}\wp'(q_i-q_j)\,,\quad i=1,...,n
 }
 \end{array}
 \eq
and
  \beq\label{t315}
  \begin{array}{c}
  \displaystyle{
{\dot S}_{ii}=0\,,
 }
 \\ \ \\
   \displaystyle{
{\dot S}_{ij}=-S_{ij}(S_{ii}-S_{jj})\wp(q_i-q_j)+\sum\limits_{k:k\neq i,j}^n
S_{ik}S_{kj}\Big(\wp(q_i-q_k)-\wp(q_j-q_k)\Big)\,,\ i\neq j\,.
 }
 \end{array}
 \eq
These are Hamiltonian equations with Hamiltonian
  \beq\label{t316}
  \begin{array}{c}
  \displaystyle{
H^{\hbox{\tiny{spin}}}=\sum\limits_{i=1}^n\frac{p_i^2}{2}-\sum\limits_{i>j}^n S_{ij}S_{ji}\wp(q_i-q_j)\,.
 }
 \end{array}
 \eq
The first equation in (\ref{t315}) is compatible with the constraints
(\ref{t308}):
  \beq\label{t317}
  \begin{array}{c}
  \displaystyle{
S_{ii}=S_{jj}={\rm const}\ \forall i,j\,,
 }
 \end{array}
 \eq
The unwanted term in (\ref{t313}) vanishes on this constraint surface.

Consider a coadjoint orbit of minimal dimension
$\dim\mO^{\hbox{\tiny{min,$n$}}}=2n-2$. In the semisimple case,
the matrix $S$ has $n-1$ coincident eigenvalues. This orbit is known
to be associated with the spinless Calogero--Moser model \cite{OP}.
Indeed, it can be described by a rank-one matrix $S_{ij}=\xi_i\eta_j$
and embedded in $T^*\mC^n$ with the canonical Poisson structure
  \beq\label{t318}
  \begin{array}{c}
  \displaystyle{
\{\xi_i,\eta_j\}=\delta_{ij}\,.
 }
 \end{array}
 \eq
By (\ref{t317}), $\xi_i\eta_i=g={\rm const}$, so the Lax matrix becomes
  \beq\label{t319}
  \begin{array}{c}
  \displaystyle{
L^{\hbox{\tiny{spin}}}_{ij}(z)\Big|_{\hbox{on-shell}}=\delta_{ij}(p_i+gE_1(z))+(1-\delta_{ij})g\frac{\xi_i}{\xi_j}\phi(z,q_{ij})\,,
 }
 \end{array}
 \eq
The diagonal gauge transformation with matrix
${\rm diag}(\xi_1^{-1},...,\xi_n^{-1})$ gives the spinless
Calogero--Moser Lax matrix (\ref{t303}), while the Hamiltonian
(\ref{t316}) reduces to (\ref{t301}).

\subsection{A class of integrable Euler--Arnold tops}
\paragraph{Construction based on the AYBE.}
The Euler top is defined by the equation
 \beq\label{t321}
 \begin{array}{c}
   \displaystyle{
 \dot{S}=[S,{\mathbb J}(S)]\,,\quad
 \quad S=\sum\limits_{k=1}^3S_k\sigma_k\,,\quad {\mathbb J}(S)=\sum\limits_{k=1}^3S_k J_k\sigma_k\,,
 }
 \end{array}
  \eq
where $\sigma_\al$ are the Pauli matrices. All variables are
complex-valued. The vector $(S_1,S_2,S_3)$ is the complexification
of the angular momentum of a rigid body, and $J_\al$ are the components
of its inverse inertia tensor in the principal-axis frame.

The Euler top (\ref{t321}) admits integrable higher-rank generalizations.
The matrix $S\in\Mat$ may then have arbitrary size, and
${\mathbb J}(S)$ is a fixed linear map. The equation
 \beq\label{t322}
 \begin{array}{c}
   \displaystyle{
 \dot{S}=[S,{\mathbb J}(S)]
 }
 \end{array}
  \eq
for $S\in\Mat$ defines an Euler--Arnold top on the Lie group
${\rm GL}_N$; see \cite{Arnold1}.

We use the construction of integrable Euler--Arnold tops in \cite{LOZ161},
based on a class of quantum $R$-matrices satisfying the AYBE (\ref{t208}).
For a matrix of dynamical variables $S\in\Mat$, define the Lax pair
 \beq\label{t323}
  \begin{array}{c}
  \displaystyle{
  L^{\hbox{\tiny{top}}}(S,z)=\tr_2\Big(r_{12}(z){S}_2\Big)\,,
  \qquad M^{\hbox{\tiny{top}}}(S,z)=\tr_2\Big(m_{12}(z){S}_2\Big)\,,\qquad
  {S}_2=1_N\otimes S\,,
 }
 \end{array}
 \eq
where $\tr_2$ is the trace over the second tensor component, i.e. for
 \beq\label{t324}
  \begin{array}{c}
  \displaystyle{
  r_{12}(z)=\sum\limits_{ijkl=1}^N r_{ij,kl}(z)e_{ij}\otimes e_{kl}\,,
   }
 \end{array}
 \eq
we have
 \beq\label{t325}
  \begin{array}{c}
  \displaystyle{
  L^{\hbox{\tiny{top}}}(S,z)=\tr_2\Big(r_{12}(z){S}_2\Big)=\sum\limits_{ijkl=1}^N r_{ij,kl}(z)e_{ij}S_{lk}\,.
 }
 \end{array}
 \eq
The Lax equation then holds identically in the spectral parameter $z$
on solutions of (\ref{t322}), with
 \beq\label{t326}
  \begin{array}{c}
  \displaystyle{
  {\mathbb J}(S)=\tr_2\Big(m_{12}(0){S}_2\Big)=M(S,0)\,.
 }
 \end{array}
 \eq
A proof is given in \cite{LOZ161}.

The inverse inertia tensor $J(S)$ in (\ref{t326}) is determined by
$m_{12}(0)$. Unlike that of the original Euler top (\ref{t321}),
it is not arbitrary for $N>2$. The construction therefore yields
a particular class of integrable tops.

The Hamiltonian formulation uses the Lie--Poisson brackets on the dual
Lie algebra ${\rm gl}_N^*$,
 \beq\label{t327}
  \begin{array}{c}
  \displaystyle{
  \{S_{ij},S_{kl}\}=\delta_{il}S_{kj}-\delta_{kj}S_{il} \quad\hbox{or}\quad \{S_1,S_2\}=[P_{12},S_1]
 }
 \end{array}
 \eq
and the Hamiltonian
 \beq\label{t328}
  \begin{array}{c}
  \displaystyle{
  H^{\hbox{\tiny{top}}}=\frac12\,\tr\Big(S\,{\mathbb J}(S)\Big)=\frac12\,\tr_{1,2}\Big(m_{12}(0)S_1S_2\Big)
 }
 \end{array}
 \eq
which together generate (\ref{t322}). In particular,
${\dot L}^{\hbox{\tiny{top}}}(z)=
\{H^{\hbox{\tiny{top}}},L^{\hbox{\tiny{top}}}(z)\}$.
The phase space is a coadjoint orbit of the Lie group ${\rm GL}_N$,
described by a matrix $S$ with fixed eigenvalues, which are Casimir
functions of (\ref{t327}).

\paragraph{Elliptic case.}

The inverse inertia tensor $J(S)$ in (\ref{t322}) is
 \beq\label{t340}
 \begin{array}{c}
  \displaystyle{
 {\mathbb J}(S)=\frac{1}{N}\Big(\frac{\vth'''(0)}{3\vth'(0)}1_N S_{0,0}-\sum\limits_{a\neq(0,0)} T_a S_a
 E_2(\frac{a_1+a_2\tau}{N})\Big)=
 }
 \\ \ \\
   \displaystyle{
 =\frac{1}{N}\Big(\frac{\vth'''(0)}{3\vth'(0)}\,S-\sum\limits_{a\neq(0,0)} T_a S_a
 \wp(\frac{a_1+a_2\tau}{N})\Big)\,.
 }
 \end{array}
 \eq
where $E_2(x)$ is defined in (\ref{a08}).
 The Poisson structure (\ref{t327}) in the basis $T_a$ takes the form
\beq\label{t341}
\begin{array}{c}
 \displaystyle{
    \{S_\al, S_\be\} = \frac{\ka_{\al, \be} - \ka_{\be, \al}}{N} S_{\al+ \be} =\frac{2\imath}{N}\,
    \sin \left( \frac{\pi}{N}\,(\al_2 \be_1 - \al_1 \be_2) \right) S_{\al +\be}\,.
}
\end{array}
\eq

For $N=2$, the model is described by the Baxter's eight-vertex $R$-matrix (\ref{t260}).

\subsection{Generalized interacting tops}
The spin matrix has size $nN\times nN$:
 \beq\label{t350}
 \begin{array}{c}
  \displaystyle{
\mS=\sum\limits_{i,j=1}^n\sum\limits_{a,b=1}^N E_{ij}\otimes e_{ab} \mS^{ij}_{ab}=\sum\limits_{i,j=1}^n
E_{ij}\otimes \mS^{ij}\in{\rm Mat}_{nN}\,,
 }
 \end{array}
 \eq
where $\mS^{ij}$ are $N\times N$ blocks:
 \beq\label{t351}
  \displaystyle{
 \mS=
\left(
 \begin{array}{cccc}
\mS^{11} &\mS^{12}   & \dots &\mS^{1n}
 \\
\mS^{21} &\mS^{22}  & \dots &\mS^{2n}
 \\
\vdots & \vdots & \dots  & \vdots
 \\
\mS^{n1} &\mS^{n2} &  \dots &\mS^{nn}
 \end{array}
 \right)
 \in{\rm Mat}_{nN}\,,
 }
 \eq
The Poisson structure is
 \beq\label{t352}
 \begin{array}{c}
  \displaystyle{
  \{\mS^{ij}_{ab},\mS^{kl}_{cd}\}=\mS^{kj}_{cb}\,\delta^{il}\,\delta_{ad}-\mS^{il}_{ad}\,\delta^{kj}\,\delta_{bc}
 }
 \end{array}
 \eq
for $i,j,k,l=1...n$ and
$a,b,c,d=1...N$. Equivalently,
 \beq\label{t353}
 \begin{array}{c}
  \displaystyle{
  \{\mS_1^{ij},\mS_2^{kl}\}=P_{12}\,\mS_1^{kj}\,\delta^{il}-\mS_{1}^{il}\,P_{12}\,\delta^{kj}\,,
 }
 \end{array}
 \eq
For the diagonal blocks we have
 \beq\label{t354}
 \begin{array}{c}
  \displaystyle{
  \{\mS_1^{ii},\mS_2^{kk}\}=[P_{12},\mS_1^{ii}]\,\delta^{ik}\,.
 }
 \end{array}
 \eq
These brackets form a direct sum of copies of (\ref{t327}).
The particle variables carry the canonical Poisson structure (\ref{t302}).

The Hamiltonian is \cite{GSZ}; see also \cite{ZL}:
 \beq\label{t355}
 \begin{array}{c}
  \displaystyle{
 \mH^{\hbox{\tiny{Gtops}}}=\sum\limits_{i=1}^n \frac{p_i^2}{2}+\frac{1}{2}\sum\limits_{i=1}^n
 \tr_{12}\Big(m_{12}(0)\mS^{ii}_1\mS^{ii}_2\Big)+\sum\limits_{i>j}^n
 \tr_{12}\Big(
 F^0_{21}(q_{ij}) P_{12}\,\mS_1^{ij}\mS_2^{ji}\Big)\,.
 }
 \end{array}
 \eq
It generates the following equations for the off-diagonal
$N\times N$ blocks of $\mS$:
  \beq\label{t356}
  \begin{array}{c}
  \displaystyle{
 {\dot
 \mS}^{ij}=\sum\limits_{k:k\neq i,j}^n
 \Big(\mS^{ik}\tr_{2}(\mS^{kj}_2F_{12}^0(q_{kj})P_{12})-\tr_2(\mS_2^{ik}F_{12}^0(q_{ik})P_{12})\mS^{kj}\Big)+
 }
 \\[6mm]
 \displaystyle{
+\mS^{ij}\tr_2\Big(m_{12}(0)\mS_2^{jj}\Big)-\tr_2\Big(m_{12}(0)\mS^{ii}_2\Big)\mS^{ij}+
 }
  \\ \ \\
 \displaystyle{
+\mS^{ii}\tr_2\Big(\mS_2^{ij}F^0_{12}(q_{ij})P_{12}\Big)  -\tr_2\Big(\mS_2^{ij}{ F}^0_{12}(q_{ij})P_{12}\Big)\mS^{jj}\,.
 }
  \end{array}
 \eq
For the diagonal $N\times N$ blocks of $\mS$, we obtain
  \beq\label{t357}
  \begin{array}{c}
  \displaystyle{
 {\dot \mS}^{ii}=[\mS^{ii},\tr_2\Big(m_{12}(0)\mS^{ii}_2\Big)]+\sum\limits_{k:k\neq
 i}^n\Big( \mS^{ik}\tr_2\Big(\mS_2^{ki}F_{12}^0(q_{ki})P_{12}\Big)-\tr_2\Big(\mS^{ik}_2F_{12}^0(q_{ik})P_{12}\Big)\mS^{ki}
 \Big)\,.
 }
  \end{array}
 \eq
For the momenta,
  \beq\label{t358}
  \begin{array}{c}
  \displaystyle{
 {\dot p}_i=-\sum\limits_{k:k\neq
 i}^n\tr_{23}\Big(\p_{q_i}F^0_{32}(q_{ik})P_{23}\,\mS^{ik}_2\mS^{ki}_3\Big)\,,
 }
  \end{array}
 \eq
and ${\dot q}_i=p_i$.

It follows from (\ref{t357}) that
  \beq\label{t359}
  \begin{array}{c}
  \displaystyle{
 \tr\Big(\dot\mS^{ii}\Big)=0\,,\ \forall i=1,...,n\,.
 }
  \end{array}
 \eq
This conservation law is associated with the subgroup
$\mD={\rm diag}(d_1,...,d_n)\otimes 1_N$ of the Cartan subgroup
of ${\rm GL}_{nN}$. Reduction by its coadjoint action imposes the constraints
  \beq\label{t360}
  \begin{array}{c}
  \displaystyle{
 \chi_i=\tr\Big(\mS^{ii}\Big)={\rm const}\,,\ \forall i=1,...,n
 }
  \end{array}
 \eq
as in the spin Calogero--Moser model (\ref{t308}).
Thus (\ref{t356})--(\ref{t358}) are written on the unreduced phase space
  \beq\label{t361}
  \begin{array}{c}
  \displaystyle{
 T^*\mC^n\times \mO^{\hbox{\tiny{$nN$}}}\,,
 }
  \end{array}
 \eq
and reduction by $\mD$, acting as $\mS\to \mD \mS \mD^{-1}$,
gives the analogue of (\ref{t310}):
  \beq\label{t362}
  \begin{array}{c}
  \displaystyle{
T^*\mC^n\times \mO^{\hbox{\tiny{$nN$}}}\ \stackrel{\hbox{reduction}}{\longrightarrow}
\ T^*\mC^n\times\Big(\mO^{\hbox{\tiny{$nN$}}}//{\mathcal D}\Big)\,.
 }
 \end{array}
 \eq

On the unreduced phase space (\ref{t361}), the Lax equation again
contains unwanted terms, which vanish under the constraints
  \beq\label{t363}
  \begin{array}{c}
  \displaystyle{
 \tr\Big(\mS^{ii}\Big)=\tr\Big(\mS^{jj}\Big)={\rm const}\,,\ \forall i,j=1,...,n\,.
 }
  \end{array}
 \eq
The Lax pair has the block structure (\ref{t104}):
  \beq\label{t364}
  \begin{array}{c}
  \displaystyle{
 \mL^{\hbox{\tiny{Gtops}}}(z)=\sum\limits_{i,j=1}^n E_{ij}\otimes \mL^{ij}(z)\,,\quad\quad
 \mL^{ij}(z)\in{\rm Mat}_{N}\quad\quad \mL(z)\in{\rm Mat}_{nN}\,,
 }
  \end{array}
 \eq
  \beq\label{t365}
  \begin{array}{c}
  \displaystyle{
 \mL^{ij}(z)=
 \delta_{ij}\Big(p_i1_N+\tr_2\Big(\mS^{\,ii}_2\,R_{12}^{z,(0)}P_{12}\Big)\Big)
  +(1-\delta_{ij})\,\tr_2\Big(\mS^{\,ij}_2\,R^z_{12}(q_{ij})P_{12}\Big)\,.
 }
  \end{array}
 \eq
The blocks $\mM^{ij}(z)\in{\rm Mat}_{N}$ are given similarly by
  \beq\label{t366}
  \begin{array}{c}
  \displaystyle{
 \mM^{ij}(z)=
 \delta_{ij}\,\tr_2\Big(\mS^{\,ii}_2\,R_{12}^{z,(1)}P_{12}\Big)
  +(1-\delta_{ij})\,\tr_2\Big(\mS^{\,ij}_2\,F^z_{12}(q_{ij})P_{12}\Big)\,.
 }
  \end{array}
 \eq
Using Fourier symmetry (\ref{t222}), these expressions can be rewritten as
  \beq\label{t367}
  \begin{array}{c}
  \displaystyle{
 \mL^{ij}(z)=
 \delta_{ij}\Big(p_i1_N+\tr_2\Big(\mS^{\,ii}_2\,r_{12}(z)\Big)\Big)
  +(1-\delta_{ij})\,\tr_2\Big(\mS^{\,ij}_2\,R^z_{12}(q_{ij})P_{12}\Big)
 }
  \end{array}
 \eq
and
  \beq\label{t368}
  \begin{array}{c}
  \displaystyle{
 \mM^{ij}(z)=
 \delta_{ij}\,\tr_2\Big(\mS^{\,ii}_2\,m_{12}(z)\Big)
  +(1-\delta_{ij})\,\tr_2\Big(\mS^{\,ij}_2\,F^z_{12}(q_{ij})P_{12}\Big)\,.
 }
  \end{array}
 \eq
For $n=1$, we recover the top (\ref{t322})--(\ref{t328}).
For $N=1$, we recover the spin Calogero--Moser model
(\ref{t311})--(\ref{t316}). Indeed, in the elliptic case, (\ref{a08}) gives
  \beq\label{t3681}
  \begin{array}{c}
  \displaystyle{
H^{\hbox{\tiny{Gtops}}}\Big|_{N=1}=H^{\hbox{\tiny{spin}}}+\frac{\vth'''(0)}{6\vth'(0)}\,\tr(S^2)\,,
 }
  \end{array}
 \eq
and
  \beq\label{t3682}
  \begin{array}{c}
  \displaystyle{
M^{\hbox{\tiny{Gtops}}}\Big|_{N=1}=M^{\hbox{\tiny{spin}}}(z)+m(z){\rm diag}(S_{11},...,S_{nn})\,,
 }
  \end{array}
 \eq
where $m(z)$ is defined in (\ref{a172}), and
${\rm diag}(S_{11},...,S_{nn})$ is proportional to $1_n$ on the
constraint surface (\ref{t363}).

The Lax equation on the unreduced phase space (\ref{t361}) reads
(cf.~(\ref{t313}))
%
  \beq\label{t369}
  \begin{array}{c}
  \displaystyle{
 \dot{\mL}^{\hbox{\tiny{Gtops}}}(z)=[\mL^{\hbox{\tiny{Gtops}}}(z),\mM^{\hbox{\tiny{Gtops}}}(z)]+\hbox{``unwanted term''}\,,
 }
 \\ \ \\
  \displaystyle{
 \hbox{``unwanted term''}=
 -\frac{1}{2} \sum\limits_{i\neq j}^n  \Big(\tr(\mS^{ii})-\tr(\mS^{jj})\Big) E_{ij} \otimes
 \tr_{2} \Big( \mS^{ij}_2  \partial_{q_{i}} F_{12}^{z} (q_{ij}) P_{12} \Big)\,.
 }
 \end{array}
 \eq

\subsection{Interacting tops}
Suppose that $\mS$ has rank one, so that
  \beq\label{t370}
  \begin{array}{c}
  \displaystyle{
\mS^{ij}_{ab}=\xi^i_a\eta^j_b\,,\quad i,j=1,...,n;\ a,b=1,...,N\,.
 }
 \end{array}
 \eq
 Consider the potential from (\ref{t355})
 \beq\label{t371}
  \begin{array}{c}
  \displaystyle{
 \tr_{12}\Big(
 F^0_{21}(q_{ij})
 P_{12}\,\mS_1^{ij}\mS_2^{ji}\Big)=\sum\limits_{a,b,c,d=1}^N (F^0_{12}(q_{ji})
 P_{12})_{ab,cd}\mS^{ij}_{ba}\mS^{ji}_{dc}\,.
 }
 \end{array}
 \eq
Right multiplication of
$T_{12}=\sum\limits_{i,j,k,l=1}^N
T_{ijkl}e_{ij}\otimes e_{kl}\in\Mat^{\otimes 2}$
by the permutation operator $P_{12}$ sends
$T_{ijkl}\rightarrow T_{ilkj}$. Hence
 \beq\label{t372}
  \begin{array}{c}
  \displaystyle{
 \tr_{12}\Big(
 F^0_{21}(q_{ij})
 P_{12}\,\mS_1^{ij}\mS_2^{ji}\Big)=\sum\limits_{a,b,c,d=1}^N (F^0_{12}(q_{ji})
 )_{ad,cb}\mS^{ij}_{ba}\mS^{ji}_{dc}\,,
 }
 \end{array}
 \eq
where we used (\ref{t2242}).
In the rank-one case,
 \beq\label{t373}
  \begin{array}{c}
  \displaystyle{
\mS^{ij}_{ba}\mS^{ji}_{dc}=\xi^i_b\eta^j_a\xi^j_d\eta^i_c=\mS^{ii}_{bc}\mS^{jj}_{da}\,.
 }
 \end{array}
 \eq
Thus
 \beq\label{t374}
  \begin{array}{c}
  \displaystyle{
 \tr_{23}\Big(
 F^0_{32}(q_{ij})
 P_{23}\,\mS_2^{ij}\mS_3^{ji}\Big)=\tr_{12}\Big(F^0_{12}(q_{ji})\mS^{jj}_1\mS^{ii}_2\Big)=
 \tr_{12}\Big(F^0_{12}(q_{ij})\mS^{ii}_1\mS^{jj}_2\Big)\,.
 }
 \end{array}
 \eq
The Hamiltonian (\ref{t355}) therefore becomes
 \beq\label{t375}
  \begin{array}{c}
  \displaystyle{
\mH^{\hbox{\tiny{tops}}}=\sum\limits_{i=1}^n
\frac{p_i^2}{2}+\frac{1}{2}\sum\limits_{i=1}^n
 \tr_{12}\Big(m_{12}(0)\mS^{ii}_1\mS^{ii}_2\Big)+\sum\limits_{i>j}^n
 \tr_{12}\Big(F^0_{12}(q_{ij})\mS^{ii}_1\mS^{jj}_2\Big)\,.
 }
 \end{array}
 \eq
In this form, it describes $n$ interacting ${\rm gl}_N$ tops attached
to $n$ particles. The equations of motion follow from the spin brackets
(\ref{t354}) and the canonical brackets (\ref{t302}):
  \beq\label{t376}
  \begin{array}{c}
  \displaystyle{
 {\dot \mS}^{ii}=[\mS^{ii},\tr_2(m_{12}(0)\mS^{ii}_2)]+\sum\limits_{k:k\neq
 i}^n\, [\mS^{ii},\tr_2(F^0_{12}(q_{ik})\mS^{kk}_2)]\,,
 }
  \end{array}
 \eq
  \beq\label{t377}
  \begin{array}{c}
  \displaystyle{
 {\dot p}_i=-\sum\limits_{k:k\neq
 i}^n\tr_{12}\Big(\p_{q_i}F^0_{12}(q_{ik})\mS^{ii}_1\mS^{kk}_2\Big)\,.
 }
  \end{array}
 \eq

The condition rank$(\mS)=1$ corresponds to a coadjoint orbit
of minimal dimension,
  \beq\label{t378}
  \begin{array}{c}
  \displaystyle{
 \dim\mO^{\hbox{\tiny{min,$nN$}}}=2(nN-1)\,.
 }
  \end{array}
 \eq
with all but one eigenvalue coinciding. Reduction (\ref{t362}) decreases
the dimension by $2n-2$:
  \beq\label{t379}
  \begin{array}{c}
  \displaystyle{
 \dim \Big(\mO^{\hbox{\tiny{min,$nN$}}}//{\mathcal D}\Big)=2(nN-1)-(2n-2)=2n(N-1)\,.
 }
  \end{array}
 \eq
This calculation reflects the fact that
  \beq\label{t380}
  \begin{array}{c}
  \displaystyle{
 \Big(\mO^{\hbox{\tiny{min,$nN$}}}//{\mathcal D}\Big)\cong
  \underbrace{ \mO^{\hbox{\tiny{min,$N$}}}\times ...\times \mO^{\hbox{\tiny{min,$N$}}} }_{\hbox{\tiny{$n$ times}}}\,,\qquad
  \dim\mO^{\hbox{\tiny{min,$N$}}}=2(N-1)\,.
 }
  \end{array}
 \eq
For $N=1$, this model reduces to the spinless Calogero--Moser model
(\ref{t301}), as explained in (\ref{t318})--(\ref{t319}).

\section{Generalized model of interacting tops for the \texorpdfstring{${\rm BC}_n$}{BCn} root system}\label{sec4}
\setcounter{equation}{0}
We construct a $2Nn\times2Nn$ Lax pair $\mL(z),\mM(z)$ satisfying
\beq\label{t401}
\begin{array}{c}
  \displaystyle{
    \dot{\mL}^{\hbox{\tiny{Gtops}}}(z)=[\mL^{\hbox{\tiny{Gtops}}}(z),\mM^{\hbox{\tiny{Gtops}}}(z)]\,.
  }
\end{array}
\eq
The matrix of classical spin variables has the form
\beq\label{t4011}
\begin{array}{c}
  \displaystyle{
    \mS=\mats{\mA}{\mB}{\mB}{\mA}\,,\quad \mA,\mB\in{\rm Mat}_{nN}\,.
  }
\end{array}
\eq
Equivalently,
\beq\label{t414}
\begin{array}{c}
  \displaystyle{
    \mS=\sum^{n}_{i,j=1}E_{ij}\otimes \mS^{ij}+\sum^{n}_{i,j=1}E_{\bar{i}j}\otimes
    \mS^{\bar{i}j}+\sum^{n}_{i,j=1}E_{i\bar{j}}\otimes\mS^{i\bar{j}}
    +\sum^{n}_{i,j=1}E_{\bar{i}\bar{j}}\otimes \mS^{\bar{i}\bar{j}}\in \text{Mat}_{2nN}\,,
  }
\end{array}
\eq
where
\beq\label{tt403}
\begin{array}{c}
  \displaystyle{
    \bar{i}=i+n,~ \qquad\bar{\bar{i}}=i\,,
     \qquad \mS^{ij}\,,\mS^{\bar{i}j}\,,\mS^{i\bar{j}}\,,\mS^{\bar{i}\bar{j}}\in \text{Mat}_N
  }
\end{array}
\eq
and
\beq\label{t4012}
\begin{array}{c}
  \displaystyle{
    \mS^{ij}=\mS^{\bar{i}\bar{j}}:=\mA^{ij}\in{\rm Mat}_{N}\,,\qquad
    \mS^{i\bar{j}}=\mS^{\bar{i}j}:
    =\mB^{ij}\in{\rm Mat}_{N}\,,~~ \forall i,j=1,...,n\,.
  }
\end{array}
\eq

 \subsection{Lax pair and equations of motion}
The Lax matrix $\mL^{\hbox{\tiny{Gtops}}}(z)\in \text{Mat}_{2nN}$
and its companion $\mM^{\hbox{\tiny{Gtops}}}(z)$ have the block
structure (\ref{t122})--(\ref{t123}). We write
\beq\label{tt402}
\begin{array}{c}
  \displaystyle{
    \mL^{\hbox{\tiny{Gtops}}}(z) =\sum^{n}_{i,j=1}E_{ij}\otimes \mL^{ij}(z)+\sum^{n}_{i,j=1}E_{\bar{i}j}\otimes
    \mL^{\bar{i}j}(z)+\sum^{n}_{i,j=1}E_{i\bar{j}}\otimes\mL^{i\bar{j}}(z)
    +\sum^{n}_{i,j=1}E_{\bar{i}\bar{j}}\otimes \mL^{\bar{i}\bar{j}}(z)\,.
  }
\end{array}
\eq
Using the abbreviations
\beq\label{tt4021}
\begin{array}{c}
  \displaystyle{
    q_{ij}=q_i-q_j\,,\quad q^+_{ij}=q_i+q_j
  }
\end{array}
\eq
the $N\times N$ blocks in (\ref{tt402}) are
\beq\label{tt404}
\begin{array}{c}
  \displaystyle{
    \mL^{ij}(z)=\delta_{ij}\left(p_{i}1_N+\tr_{2}\Big(\mA^{ii}_{2}R^{z,(0)}_{12}P_{12}\Big)\right)
    +(1-\de_{ij})\tr_{2}\Big(\mA^{ij}_{2}R^{z}_{12}(q_{ij})P_{12}\Big)=
  }
  \\ \ \\
    \displaystyle{
  =\delta_{ij}\left(p_{i}1_N+\tr_{2}\Big(\mA^{ii}_{2}r_{12}(z)\Big)\right)
    +(1-\de_{ij})\tr_{2}\Big(\mA^{ij}_{2}R^{z}_{12}(q_{ij})P_{12}\Big)\,,
    }
\end{array}
\eq
\beq\label{tt405}
\begin{array}{c}
  \displaystyle{
    \mL^{i\bar{j}}(z)=\de_{ij}\left( \ti{K}^{z}_{1}(q_{i})
    +\tr_{2}\Big(\mB^{ii}_{2}\ti{R}^{z}_{12}(2q_{i})P_{12}\Big)\right) +(1-\de_{ij})
    \tr_{2}\Big(\mB^{ij}_2\ti{R}^{z}_{12}(q_{ij}^{+})P_{12}\Big)\,,
  }
\end{array}
\eq
\beq\label{tt407}
\begin{array}{c}
  \displaystyle{
    \mL^{\bar{i}j}(z)=-\de_{ij}\left( \ti{K}^{-z}_{1}(q_{i})
    +\tr_{2}\Big(\mB^{ii}_{2}\ti{R}^{-z}_{12}(2q_{i})P_{12}\Big)\right)-(1-\de_{ij})
    \tr_{2}\Big(\mB^{ij}_{2}\ti{R}^{-z}_{12}(q_{ij}^{+})P_{12}\Big)\,,
  }
\end{array}
\eq
\beq\label{tt406}
\begin{array}{c}
  \displaystyle{
    \mL^{\bar{i}\bar{j}}(z)=-\de_{ij}\left(p_{i}1_{N}
    +\tr_{2}\Big(\mA^{ii}_{2}R^{-z,(0)}_{12}P_{12}\Big)\right)-(1-\de_{ij})
    \tr_{2}\Big(\mA^{ij}_{2}R^{-z}_{12}(q_{ij})P_{12}\Big)=
  }
  \\ \ \\
  \displaystyle{
    =-\de_{ij}\left(p_{i}1_{N}
    +\tr_{2}\Big(\mA^{ii}_{2}r_{12}(-z)\Big)\right)-(1-\de_{ij})
    \tr_{2}\Big(\mA^{ij}_{2}R^{-z}_{12}(q_{ij})P_{12}\Big)\,,
  }
\end{array}
\eq
Similarly, the blocks of $\mM^{\hbox{\tiny{Gtops}}}(z)\in \text{Mat}_{2nN}$ are
\beq\label{tt408}
\begin{array}{c}
  \displaystyle{
    \mM^{ij}(z)=\de_{ij}\tr_{2}\Big(\mA^{ii}_{2}R^{z,(1)}_{12}P_{12}\Big)+(1-\de_{ij})
    \tr_{2}\Big(\mA^{ij}_{2}F^{z}_{12}(q_{ij})P_{12}\Big)=
  }
  \\ \ \\
  \displaystyle{
    =\de_{ij}\tr_{2}\Big(\mA^{ii}_{2}m_{12}(z)\Big)+(1-\de_{ij})
    \tr_{2}\Big(\mA^{ij}_{2}F^{z}_{12}(q_{ij})P_{12}\Big)\,,
  }
\end{array}
\eq
\beq\label{tt409}
\begin{array}{c}
  \displaystyle{
    \mM^{i\bar{j}}(z)=\de_{ij}\left(\frac{1}{2}\,\ti{Y}^{z}_{1}(q_{i})
    +\tr_{2}\Big(\mB^{ii}_2\ti{F}^{z}_{12}(2q_{i})P_{12}\Big)\right) +(1-\de_{ij})
    \tr_{2}\Big(\mB^{ij}_{2}\ti{F}^{z}_{12}(q_{ij}^{+})P_{12}\Big)\,,
  }
\end{array}
\eq
\beq\label{tt411}
\begin{array}{c}
  \displaystyle{
    \mM^{\bar{i}j}(z)=\de_{ij}\left(\frac{1}{2}\,\ti{Y}^{-z}_{1}(q_{i})
    +\tr_{2}\Big(\mB^{ii}_{2}\ti{F}^{-z}_{12}(2q_{i})P_{12}\Big)\right)+(1-\de_{ij})
    \tr_{2}\Big(\mB^{ij}_2\ti{F}^{-z}_{12}(q_{ij}^{+})P_{12}\Big)\,,
  }
\end{array}
\eq
\beq\label{tt410}
\begin{array}{c}
  \displaystyle{
    \mM^{\bar{i}\bar{j}}(z)=\de_{ij}\tr_{2}\Big(\mA^{ii}_2R_{12}^{-z,(1)}P_{12}\Big)
    +(1-\de_{ij})\tr_{2}\Big(\mA^{ij}_{2}F^{-z}_{12}(q_{ij})P_{12}\Big)=
  }
  \\ \ \\
    \displaystyle{
    =\de_{ij}\tr_{2}\Big(\mA^{ii}_2 m_{12}(-z)\Big)
    +(1-\de_{ij})\tr_{2}\Big(\mA^{ij}_{2}F^{-z}_{12}(q_{ij})P_{12}\Big)\,,
  }
\end{array}
\eq
Here we use the notation of (\ref{t202})--(\ref{t215}) and
(\ref{t223})--(\ref{t2242}). In (\ref{tt404}), (\ref{tt406}),
(\ref{tt408}), and (\ref{tt410}), the two expressions are related by
Fourier symmetry (\ref{t222}), which we assume throughout.\footnote{The
Lax equation can be proved independently in either form. Fourier symmetry
is therefore convenient but not essential for the proof.}

\begin{theor}\label{th1}
Let the meromorphic $R,\tilde R,\tilde K$-matrices satisfy the ${\rm BC}_n$ AYBEs
(\ref{t206})--(\ref{t212}) and all the properties stated in Section 2.2.
In addition, assume the symmetry and skew-symmetry relations
(\ref{t204})--(\ref{t205}), the unitarity relations
(\ref{t218})--(\ref{t220}), and the normalized local expansions
(\ref{t202})--(\ref{t215}) and (\ref{t217}) from Section 2.1.
In particular, the residue of $R^z_{12}(q)$ at $q=0$ is the permutation
operator $P_{12}$ in (\ref{t216}), the coefficient $m_{12}(q)$ is regular
at $q=0$, and the residue
$\ti K^{(-1)}=\varsigma_N1_N$ is scalar and independent of $q$.
The coupling constants are fixed, and the particle coordinates are
restricted to the domain where the Lax matrices are defined.

Then the pair (\ref{tt404})--(\ref{tt410}) satisfies the Lax equation
(\ref{t401}) identically in $z$ on the constraint surface
\beq\label{t416}
\begin{array}{c}
  \displaystyle{
    \tr\Big(\mS^{ii}\Big)=\tr\Big(\mS^{jj}\Big)={\rm const}, ~~\forall i,j\,,
    }
\end{array}
\eq
and yields the following equations of motion for the off-diagonal
$N\times N$ blocks of
$\mA$:
\beq\label{tt417}
\begin{array}{c}
  \displaystyle{
    \dot{\mA}^{ij} =
     \sum_{k:k\neq i,j}^{n}\Big( \mA^{ik}\tr_2\Big(
    \mA_2^{kj}F_{12}^{0}(q_{kj})P_{12}\Big) -
    \tr_2\Big(F_{12}^{0}(q_{ik})P_{12}\mA_2^{ik}\Big)\mA^{kj} \Big) +
  }
  \\ \ \\
    \displaystyle{
    + \sum_{k:k\neq i,j}^{n}\Big(\mB^{ik}\tr_{2}\Big(\mB_{2}^{kj}
    \ti{F}_{12}^{0}(q_{kj}^{+})P_{12}\Big) -
    \tr_2\Big(\ti{F}_{12}^{0}(q_{ik}^{+})P_{12}\mB_{2}^{ik}\Big)\mB^{kj}\Big) +
  }
\end{array}
\eq
$$
\begin{array}{c}
  \displaystyle{
    + \mA^{ii}\tr_2\Big(\mA_{2}^{ij}F_{12}^{0}(q_{ij})P_{12}\Big)
    -\tr_2\Big(F_{12}^{0}(q_{ij})P_{12}\mA_{2}^{ij}\Big)\mA^{jj}+
  }
  \\ \ \\
  \displaystyle{
    + \mA^{ij}\tr_2\Big(\mA_2^{jj}m_{12}(0)\Big)
    -\tr_2\Big(m_{12}(0)\mA_2^{ii}\Big)\mA^{ij}
     +\frac{1}{2}\,\mB^{ij}\ti{Y}^{0}(q_{j})-\frac{1}{2}\,\ti{Y}^{0}(q_{i})\mB^{ij}+
  }
  \\ \ \\
  \displaystyle{
     + \mB^{ij}\tr_{2}\Big(\mB_{2}^{jj}\ti{F}^{0}_{12}(2q_{j})P_{12}\Big)
      -\tr_{2}\Big(\ti{F}^{0}_{12}(2q_{i})P_{12}\mB_{2}^{ii}\Big)\mB^{ij}+
  }
  \\ \ \\
  \displaystyle{
+\mB^{ii}\tr_{2}\Big(\mB_{2}^{ij}\ti{F}^{0}_{12}(q_{ij}^{+})P_{12}\Big)
    -\tr_{2}\Big(\ti{F}^{0}_{12}(q_{ij}^{+})P_{12}\mB_{2}^{ij}\Big)\mB^{jj}\,,
  }
\end{array}
$$
for the off-diagonal $N\times N$ blocks of
$\mB$:
\beq\label{tt418}
\begin{array}{c}
  \displaystyle{
    \dot{\mB}^{ij} = \sum_{k:k\neq i,j}
    \Big(\mA^{ik}\tr_2\Big(\mB_2^{kj}
    \ti{F}_{12}^{0}(q_{kj}^{+})P_{12}\Big) -
    \tr_2\Big(\ti{F}_{12}^{0}(q_{ik}^{+})P_{12}\mB_2^{ik}\Big)\mA^{kj} \Big)
  }
  \\ \ \\
  \displaystyle{
    + \sum_{k:k\neq i,j} \Big( \mB^{ik}\tr_{2}\Big(\mA_{2}^{kj}F_{12}^{0}(q_{kj})P_{12}\Big)
    -\tr_2\Big(F_{12}^{0}(q_{ik})P_{12}\mA_2^{ik}\Big)\mB^{kj}\Big) +
  }
\end{array}
\eq
$$
\begin{array}{c}
  \displaystyle{
    + \mA^{ii}\tr_2\Big(\mB_2^{ij}\ti{F}_{12}^{0}(q_{ij}^{+})P_{12}\Big)
    -\tr_2\Big(\ti{F}_{12}^{0}(q_{ij}^{+})P_{12}\mB_2^{ij}\Big)\mA^{jj}+
  }
  \\ \ \\
  \displaystyle{
    + \mB^{ij}\tr_2\Big(\mA_2^{jj}m_{12}(0)\Big)
    -\tr_2\Big(m_{12}(0)\mA_2^{ii}\Big)\mB^{ij}
    +\frac{1}{2}\,\mA^{ij}\ti{Y}^{0}(q_{j}) -\frac{1}{2}\,\ti{Y}^{0}(q_{i})\mA^{ij} +
  }
  \\ \ \\
  \displaystyle{
      + \mB^{ii}\tr_{2}\Big(\mA_{2}^{ij}F_{12}^{0}(q_{ij})P_{12}\Big)
    -\tr_2\Big(F_{12}^{0}(q_{ij})P_{12}\mA_2^{ij}\Big)\mB^{jj} +
  }
  \\ \ \\
  \displaystyle{
    + \mA^{ij}\tr_2\Big(\mB_2^{jj}\ti{F}_{12}^{0}(2q_{j})P_{12}\Big)
    -\tr_2\Big(\ti{F}_{12}^{0}(2q_{i})P_{12}\mB_2^{ii}\Big)\mA^{ij}\,,
  }
\end{array}
$$
for the diagonal $N\times N$ blocks of $\mA$:
\beq\label{tt419}
\begin{array}{c}
  \displaystyle{
    \dot{\mA}^{ii} = [\mA^{ii},\tr_2\Big(m_{12}(0)\mA_2^{ii}\Big)] +\frac{1}{2}\,[\mB^{ii},
    \ti{Y}^{0}(q_{i})] +
  }
  \\ \ \\
  \displaystyle{
    + \sum_{k:k\neq i}^n\left(\mA^{ik}\tr_2\Big(\mA_2^{ki}F_{12}^{0}(q_{ki})P_{12}\Big)
    -\tr_2\Big(F_{12}^{0}(q_{ik})P_{12}\mA_2^{ik}\Big)\mA^{ki}\right) +
  }
  \end{array}
\eq
$$
\begin{array}{c}
  \displaystyle{
    + \sum_{k:k\neq i}^n\left(\mB^{ik}\tr_2\Big(\mB_2^{ki}\ti{F}_{12}^{0}(q_{ki}^{+})P_{12}\Big)
    -\tr_2\Big(\ti{F}_{12}^{0}(q_{ik}^{+})P_{12}\mB_2^{ik}\Big)\mB^{ki}\right) +
  }
  \\ \ \\
  \displaystyle{
    + \mB^{ii}\tr_{2}\Big(\mB_{2}^{ii}\ti{F}^{0}_{12}(2q_{i})P_{12}\Big)
    -\tr_{2}\Big(\ti{F}^{0}_{12}(2q_{i})P_{12}\mB_{2}^{ii}\Big)\mB^{ii}\,,
  }
\end{array}
$$
for the diagonal $N\times N$ blocks of $\mB$:
\beq\label{tt420}
\begin{array}{c}
  \displaystyle{
    \dot{\mB}^{ii} =
     [\mB^{ii},\tr_2\Big(m_{12}(0)\mA_2^{ii}\Big)] +\frac{1}{2}\,[\mA^{ii},
    \ti{Y}^{0}(q_{i})] +
  }
  \\ \ \\
    \displaystyle{
     +\sum_{k:k\neq i}^n \Big(  \mA^{ik}\tr_2\Big(\mB_2^{ki}
    \ti{F}_{12}^{0}(q_{ik}^{+})P_{12}\Big)
    -\tr_2\Big(\ti{F}_{12}^{0}(q_{ik}^{+})P_{12}\mB_2^{ik}\Big)\mA^{ki} \Big)+
  }
  \end{array}
\eq
$$
\begin{array}{c}
  \displaystyle{
    + \sum_{k:k\neq i}^n \Big(  \mB^{ik}\tr_2\Big(\mA_2^{ki}F_{12}^{0}(q_{ki})P_{12}\Big)
    -\tr_2\Big(F_{12}^{0}(q_{ik})P_{12}\mA_2^{ik}\Big)\mB^{ki} \Big)+
  }
  \\ \ \\
  \displaystyle{
    + \mA^{ii}\tr_2\Big(\mB_2^{ii}\ti{F}_{12}^{0}(2q_i)P_{12}\Big)
    -\tr_2\Big(\ti{F}_{12}^{0}(2q_i)P_{12}\mB_2^{ii}\Big)\mA^{ii}\,,
  }
\end{array}
$$
and for the momenta:
\beq\label{tt421}
\begin{array}{c}
  \displaystyle{
    \dot{p}_{i} = \frac{1}{2}\, \p_{q_i}u_N(q_i|\nu) -\frac{1}{2}\,
    \tr_{2}\Big(\partial\ti{Y}^{0}_{2}(q_{i})\mB^{ii}_{2}\Big)
-\tr_{23}\Big(\partial_{2q_{i}}\ti{F}^{0}_{32}(2q_{i})P_{23}
    \mB_{2}^{ii}\mB_{3}^{ii}\Big)-
  }
  \\ \ \\
  \displaystyle{
    - \sum_{k:k \neq i} \Big( \tr_{23}\Big(\partial_{q_{ik}^{+}} \ti{F}_{32}^{0} (q_{ik}^{+})P_{23}
    \mB_2^{ik} \mB_3^{ki}\Big) + \tr_{23} \Big(\partial F_{32}^{0}(q_{ik}) P_{23}
    \mA_2^{ik}\mA_3^{ki}\Big)  \Big)\,,
  }
\end{array}
\eq
where $u_N(q_i|\nu)$ is defined by the unitarity relation
(\ref{t220})\footnote{In the elliptic case (\ref{t254})--(\ref{t256}),
$\p_{q_i}k_N(z,q_i|\nu)=\p_{q_i}u_N(q_i|\nu)=
\sum_{b=0}^{3}\bar\nu_{b}^{2}\wp'(q_{i}+\Omega_{b})$,
where $\bar\nu_b$ are defined in (\ref{a48}).} and
\beq\label{tt4215}
\begin{array}{c}
  \displaystyle{
  \partial\ti{Y}^{0}(x)=\p_x\ti{Y}^{0}(x) \,,\quad \partial F_{32}^{0}(x)=\partial_x F_{32}^{0}(x)\,,
  \quad \partial {\ti F}_{32}^{0}(x)=\partial_x {\ti F}_{32}^{0}(x)\,.
  }
\end{array}
\eq
We assume $p_i={\dot q}_i$, $i=1,...,n$.
\end{theor}
More precisely, the following relation holds:
\beq\label{t4211}
\begin{array}{c}
  \displaystyle{
    \dot{\mL}^{\hbox{\tiny{Gtops}}}(z)=[\mL^{\hbox{\tiny{Gtops}}}(z),\mM^{\hbox{\tiny{Gtops}}}(z)]
    +\hbox{``unwanted terms''}\,,
  }
\end{array}
\eq
where
\beq\label{t4212}
\begin{array}{c}
  \displaystyle{
\hbox{``unwanted terms''}=
 }
 \\ \ \\
  \displaystyle{
=-\frac12\,\sum\limits_{i\neq j}^n (\chi_i-\chi_j) E_{ij}\otimes\tr_2\Big( \mA_2^{ij}
\p^2_x R_{12}^z(x)\Big|_{x=q_i-q_j}P_{12} \Big)-
 }
 \\ \ \\
  \displaystyle{
-\frac12\,\sum\limits_{i\neq j}^n (\chi_i-\chi_j) E_{i{\bar j}}\otimes\tr_2\Big( \mB_2^{ij}
\p^2_x {\ti R}_{12}^z(x)\Big|_{x=q_i+q_j}P_{12} \Big)+
  }
   \\ \ \\
  \displaystyle{
+\frac12\,\sum\limits_{i\neq j}^n (\chi_i-\chi_j) E_{{\bar i}{\bar j}}\otimes\tr_2\Big( \mA_2^{ij}
\p^2_x R_{12}^{-z}(x)\Big|_{x=q_i-q_j}P_{12} \Big)+
 }
 \\ \ \\
   \displaystyle{
+\frac12\,\sum\limits_{i\neq j}^n (\chi_i-\chi_j) E_{{\bar i} j}\otimes\tr_2\Big( \mB_2^{ij}
\p^2_x {\ti R}_{12}^{-z}(x)\Big|_{x=q_i+q_j}P_{12} \Big)\,.
  }
\end{array}
\eq
with the notation
\beq\label{t4213}
\begin{array}{c}
  \displaystyle{
\chi_i=\tr\Big(\mS^{ii}\Big)=\tr\Big(\mS^{{\bar i}{\bar i}}\Big)=\tr\Big(\mA^{ii}\Big)\,.
  }
\end{array}
\eq
The unwanted term vanishes under (\ref{t416}), since all $\chi_i$
are then equal to the same constant.
The theorem is proved in Appendix B.

The existence of a Lax pair with a spectral parameter is not sufficient
for Liouville integrability. One must also establish the Poisson
commutativity of the Hamiltonians obtained from
$\tr(\mL^{\hbox{\tiny{Gtops}}}(z)^k)$ and determine the number
of functionally independent Hamiltonians. These questions will be
addressed elsewhere.

\paragraph{Compact form of the equations of motion.}
The spin equations (\ref{tt417})--(\ref{tt420}) for $\mA,\mB$
can be written in a compact form.
For $T\in\Mat$, define the linear operators
\beq\label{t4214}
\begin{array}{c}
  \displaystyle{
{\mathbb J}(T)=\tr_2\Big(m_{12}(0)T_2\Big)\,,
 }
 \\ \ \\
   \displaystyle{
{\mathbb{FP}}_x(T)=\tr_2\Big(F_{12}^0(x)P_{12}T_2\Big)\,,
\qquad
{\mathbb {\ti FP}}_x(T)=\tr_2\Big({\ti F}_{12}^0(x)P_{12}T_2\Big)\,.
  }
\end{array}
\eq
The operator ${\mathbb J}$ was introduced in (\ref{t112}).

Define ${\mathcal X},{\mathcal Y}\in{\rm Mat}_{nN}$ by their
$N\times N$ blocks:
\beq\label{t4215}
\begin{array}{c}
  \displaystyle{
{\mathcal X}^{ij}=\delta_{ij}\,{\mathbb J}(\mA^{ii})+(1-\delta_{ij}){\mathbb{FP}}_{q_i-q_j}(\mA^{ij})\in{\rm Mat}_{N}\,,
\quad {\mathcal X}=\sum\limits_{i,j=1}^n E_{ij}\otimes {\mathcal X}^{ij}\,,
 }
 \\ \ \\
   \displaystyle{
{\mathcal Y}^{ij}=\frac12\,\delta_{ij}\,{\ti Y}^0(q_i)+{\mathbb {\ti FP}}_{q_i+q_j}(\mB^{ij})\in{\rm Mat}_{N}\,,
\quad {\mathcal Y}=\sum\limits_{i,j=1}^n E_{ij}\otimes {\mathcal Y}^{ij}\,,
  }
\end{array}
\eq
where $E_{ij}$ are matrix units in ${\rm Mat}_n$.

\begin{predl}
The equations of motion (\ref{tt417})--(\ref{tt420}) take the form
\beq\label{t4216}
\begin{array}{c}
  \displaystyle{
\dot\mA=[\mA,{\mathcal X}]+[\mB,{\mathcal Y}]\,,
 }
 \\ \ \\
   \displaystyle{
\dot\mB=[\mA,{\mathcal Y}]+[\mB,{\mathcal X}]\,.
  }
\end{array}
\eq
\end{predl}
This follows by direct calculation.

\subsection{Hamiltonian description}

\paragraph{The Hamiltonian function.}
We obtain the Hamiltonian of the model (\ref{tt404})--(\ref{tt421})
from the generating function $\tr \mL^2(z)$.

\begin{theor}\label{th2}
Assume the matrix hypotheses of Theorem \ref{th1} and the block
identifications (\ref{t4012}).
Assume, in addition, that the partial traces of $r_{12}(z)$ and
$m_{12}(z)$ are scalar and satisfy (\ref{t424})--(\ref{t4241})
with scalar functions $\widetilde E_1(z)$ and $\widetilde\wp(z)$.
Then the following statements hold.

1. The Hamiltonian obtained from the Lax matrix
$\mL^{\hbox{\tiny{Gtops}}}(z)$ in (\ref{tt402})--(\ref{tt406}) is
\beq
\label{tt422}
\begin{array}{c}
  \displaystyle{
    H^{\hbox{\tiny{Gtops}}}=\frac{1}{2}\,\sum_{i=1}^{n}p_i^2-\frac{1}{2}\,\sum_{i=1}^{n}u_N(q_i|\nu)
    +\frac{1}{2}\,\sum_{i=1}^{n}\tr_{23}\Big(m_{23}(0)\mA_2^{ii}\mA_3^{ii}\Big)+
  }
\end{array}
\eq
$$
\begin{array}{c}
  \displaystyle{
    +\frac{1}{2}\,\sum_{i\neq j}^{n}\tr_{23}\Big(F_{32}^{0}(q_{ij})P_{23}\mA_2^{ij}\mA_3^{ji}\Big)
    +\frac{1}{2}\,\sum_{i\neq j}^{n}\tr_{23}\Big(\ti{F}_{32}^{0}(q_{ij}^{+})P_{23}
    \mB_2^{ij}\mB_3^{ji}\Big)+
  }
  \\ \ \\
  \displaystyle{
    +\frac{1}{2}\,\sum_{i=1}^{n}\tr_{23}\Big(\ti{F}_{32}^{0}(2q_i)P_{23}\mB_2^{ii}\mB_3^{ii}\Big)
    +\frac{1}{2}\,\sum_{i=1}^{n}
    \tr_2\Big(\mB_2^{ii}\ti{Y}_2^{(0)}(q_i)\Big)\,,
  }
\end{array}
$$
where $u_N(q_i|\nu)$ is defined by the unitarity relation
(\ref{t220}). This Hamiltonian appears in the expansion
\beq\label{tt423}
\begin{array}{c}
   \displaystyle{
    \frac{1}{4N}\,\tr\Big(\mL^{2}(z)\Big)=H^{\hbox{\tiny{Gtops}}}-\frac{\partial_{z}\widetilde{E}_{1}(z)}{2N}\,\Big[2
    \varsigma_N\tr\Big(\mB\Big)+\tr\Big(\mA^2\Big)+\tr\Big(\mB^2\Big)\Big]+
  }
  \\ \ \\
   \displaystyle{
   +\frac{{\widetilde E}_1^2(z)-\widetilde\wp(z)}{2N^2}\,\sum
   \limits_{i=1}^n\Big(\tr\Big(\mA^{ii}\Big)\Big)^2+\frac{n}{2}\,w_N(z|\nu)\,,
   }
\end{array}
\eq
\beq\label{t4242}
\begin{array}{c}
  \displaystyle{
\tr\Big(\mA^2\Big)=\tr\Big(\sum\limits_{i,j=1}^n \mS^{ij}\mS^{ji}\Big)\,,
\qquad\tr\Big(\mB^2\Big)=\tr\Big(\sum\limits_{i,{j}=1}^n \mS^{i{\bar j}}\mS^{{\bar j}i}\Big)\,,
  }
\end{array}
\eq
where $\varsigma_N$ and $w_N(z|\nu)$ are defined in
(\ref{t217}) and (\ref{t220}), respectively.
The scalar partial-trace conditions assumed above are\footnote{For the
elliptic Baxter--Belavin matrix (\ref{t239}), these conditions hold with
$\widetilde E_1(z)=E_1(z)$ and $\widetilde\wp(z)=\wp(z)$;
see (\ref{t242})--(\ref{t243}) and (\ref{a172}).}
\beq\label{t424}
\begin{array}{c}
  \displaystyle{
\tr_{1}\Big(r_{12}(z)\Big)=\widetilde E_{1}(z)1_N\,,
  }
\end{array}
\eq
and
\beq\label{t4241}
\begin{array}{c}
  \displaystyle{
\tr_{1}\Big(m_{12}(z)\Big)=\frac{{\widetilde E}_1^2(z)-\widetilde\wp(z)}{2N}\,1_N\,.
  }
\end{array}
\eq
Together with (\ref{t2161}), these conditions imply
$\widetilde E_1(-z)=-\widetilde E_1(z)$ and
$\widetilde\wp(-z)=\widetilde\wp(z)$. The same scalar expressions are
obtained by tracing over the second tensor component.
In particular, differentiating (\ref{t424}) gives
$\tr_1 F^0_{12}(z)=\partial_z\widetilde E_1(z)1_N$.

2. With the following Poisson brackets, the Hamiltonian
$H^{\hbox{\tiny{Gtops}}}$ in (\ref{tt422}) generates precisely the
equations of motion (\ref{tt417})--(\ref{tt421}) obtained from the Lax equation:
\beq\label{t425}
\begin{array}{c}
  \displaystyle{
    \{p_i,q_j\}=\delta_{ij}, \qquad \{p_i,p_j\}=\{q_i,q_j\}=0, \qquad i,j=1,\ldots,n\,,
  }
\end{array}
\eq
and\footnote{Recall the barred-index convention (\ref{tt403}):
for $1\leq i\leq n$ and $I=i+n$,
${\bar I}={\overline{i+n}}=\bar{\bar{i}}=i$.}
\beq\label{t426}
\begin{array}{c}
  \displaystyle{
    \{\mS^{IJ}_{1},\mS^{KL}_{2}\} = \delta^{IL} P_{12}\mS^{KJ}_{1} -\delta^{KJ}\mS^{IL}_{1}P_{12}+
  }
  \\ \ \\
  \displaystyle{
    +\delta^{I\bar L} P_{12}\mS^{\bar K J}_{1} -\delta^{\bar K J} \mS^{I\bar L}_{1}P_{12}\,,
  }
\end{array}
\eq
where $I,J,K,L=1,\ldots,2n$ and $P_{12}$, as before, is the
permutation operator in $\text{Mat}(N,\mathbb{C})^{\otimes2}$.
All mixed Poisson brackets between $p_i,q_i$ and the spin variables
$\mS^{IJ}_{ab}$ vanish.
\end{theor}
The second statement follows by direct calculation; see also the
description of (\ref{t426}) below. The first statement, including
(\ref{tt423}), is proved in Appendix C using $R$- and $K$-matrix identities.

\paragraph{Poisson brackets.}
We describe the spin part of the phase space in more detail.
Consider the dynamical matrices $\mA\in{\rm Mat}_{nN}$ and
$\mB\in{\rm Mat}_{nN}$ entering $\mS\in{\rm Mat}_{2nN}$ in
(\ref{t4011}). While $P_{12}$ denotes the permutation operator
(\ref{t216}) in $\Mat^{\otimes 2}$, let
$\Pi_{12}\in{\rm Mat}_{nN}^{\otimes 2}$ denote the corresponding
operator on $\mC^{nN}\otimes\mC^{nN}$.
The brackets (\ref{t426}) then become
\beq\label{t432}
\begin{array}{c}
  \displaystyle{
    \{\mA_1,\mA_2\}=[\Pi_{12},\mA_1],\qquad\{\mA_1,\mB_2\}=[\Pi_{12},\mB_1]\,,
  }
  \\ \ \\
  \displaystyle{
    \{\mB_1,\mA_2\}=[\Pi_{12},\mB_1],\qquad\{\mB_1,\mB_2\}=[\Pi_{12},\mA_1]\,,
  }
\end{array}
\eq
where $1,2$ label the tensor components in ${\rm Mat}_{nN}^{\otimes 2}$.
Equivalently, cf.~(\ref{t353})--(\ref{t354}),
\beq\label{t4321}
\begin{array}{c}
  \displaystyle{
    \{\mA^{ij}_1,\mA^{kl}_2\}=P_{12}\mA^{kj}_1\delta^{il}-\mA^{il}_1P_{12}\delta^{kj}\,,
    \qquad\{\mA^{ij}_1,\mB^{kl}_2\}=P_{12}\mB^{kj}_1\delta^{il}-\mB^{il}_1P_{12}\delta^{kj}\,,
  }
  \\ \ \\
  \displaystyle{
   \{\mB^{ij}_1,\mA^{kl}_2\}=P_{12}\mB^{kj}_1\delta^{il}-\mB^{il}_1P_{12}\delta^{kj}\,,
    \qquad
    \{\mB^{ij}_1,\mB^{kl}_2\}=P_{12}\mA^{kj}_1\delta^{il}-\mA^{il}_1P_{12}\delta^{kj}\,,
  }
\end{array}
\eq
where $\mA^{ij},\mB^{ij}\in {\rm Mat}_{N}$, $i,j,k,l=1,...,n$,
and $1,2$ now label the components in ${\rm Mat}_{N}^{\otimes 2}$.
In components, as in (\ref{t352}),
\beq\label{t4322}
\begin{array}{c}
  \displaystyle{
    \{\mA^{ij}_{ab},\mA^{kl}_{cd}\}=\mA^{kj}_{cb}\,\delta^{il}\,\delta_{ad}-\mA^{il}_{ad}\,\delta^{kj}\,\delta_{bc}\,,
    \qquad
    \{\mA^{ij}_{ab},\mB^{kl}_{cd}\}=\mB^{kj}_{cb}\,\delta^{il}\,\delta_{ad}-\mB^{il}_{ad}\,\delta^{kj}\,\delta_{bc}\,,
  }
  \\ \ \\
  \displaystyle{
   \{\mB^{ij}_{ab},\mA^{kl}_{cd}\}=\mB^{kj}_{cb}\,\delta^{il}\,\delta_{ad}-\mB^{il}_{ad}\,\delta^{kj}\,\delta_{bc}\,,
    \qquad
    \{\mB^{ij}_{ab},\mB^{kl}_{cd}\}=\mA^{kj}_{cb}\,\delta^{il}\,\delta_{ad}-\mA^{il}_{ad}\,\delta^{kj}\,\delta_{bc}\,,
  }
\end{array}
\eq
where $i,j,k,l=1,...,n$ and $a,b,c,d=1,...,N$.

Introduce the matrices $\mT^\pm\in{\rm Mat}_{nN}$:
\beq\label{t433}
\begin{array}{c}
  \displaystyle{
    \mT^+=\frac{\mA+\mB}{2}\,,\qquad \mT^-=\frac{\mA-\mB}{2}\,.
  }
\end{array}
\eq
A direct calculation transforms (\ref{t432}) into
\beq\label{t434}
\begin{array}{c}
  \displaystyle{
    \{\mT^+_1,\mT^+_2\}=[\Pi_{12},\mT^+_1]\,,
      }
  \\ \ \\
  \displaystyle{
    \{\mT^-_1,\mT^-_2\}=[\Pi_{12},\mT^-_1]\,,
      }
  \\ \ \\
  \displaystyle{
    \{\mT^+_1,\mT^-_2\}=0\,.
  }
\end{array}
\eq
These are two mutually commuting copies of the standard Lie--Poisson
bracket on ${\rm gl}(nN,\mC)^*$, giving the Lie--Poisson structure on
\beq\label{t435}
\begin{array}{c}
  \displaystyle{
    ({\rm gl}(nN,\mC)\oplus{\rm gl}(nN,\mC))^*\,.
  }
\end{array}
\eq
We may fix the Casimir functions
\beq\label{t439}
\begin{array}{c}
  \displaystyle{
    C_k^+=\operatorname{tr}\bigl((\mT^+)^k\bigr),\qquad C_k^-=\operatorname{tr}\bigl((\mT^-)^k\bigr)\,,
    \qquad k=1,\ldots,nN\,.
  }
\end{array}
\eq
The phase space of the generalized ${\rm BC}_n$ model is therefore
\beq\label{t427}
\begin{array}{c}
  \displaystyle{
   T^*\mC^{n}\times(\mO^{\hbox{\tiny{$nN\!,\!+$}}}\times \mO^{\hbox{\tiny{$nN\!,\!-$}}})//\mD\,,
  }
\end{array}
\eq
where $\mO^{\hbox{\tiny{$nN\!,\!\pm$}}}$ are coadjoint orbits
of the Lie group ${\rm GL}_{nN}$, and
$\mD={\rm diag}(d_1,...,d_n)\otimes 1_N$ is the same subgroup
of the Cartan subgroup of ${\rm GL}_{nN}$ as in (\ref{t362}).
Reduction by $\mD$ imposes the first-class
constraints (\ref{t416}),
\beq\label{t428}
\begin{array}{c}
  \displaystyle{
    \tr\Big(\mS^{ii}\Big)=\tr\Big(\mA^{ii}\Big)=\tr\Big(\mT^{+,ii}+\mT^{-,ii}\Big)={\rm const}
  }
\end{array}
\eq
together with suitable gauge-fixing conditions. For generic regular
orbits, $\dim\mO^{\hbox{\tiny{$nN\!,\!\pm$}}}=nN(nN-1)$.
The reduced phase space therefore has generic dimension
\beq\label{t4281}
\begin{array}{c}
  \displaystyle{
   \dim\Big(T^*\mC^{n}\times(\mO^{\hbox{\tiny{$nN\!,\!+$}}}\times \mO^{\hbox{\tiny{$nN\!,\!-$}}})//\mD\Big)=
   2n+2nN(nN-1)-2(n-1)=2(n^2N^2-nN+1)\,.
  }
\end{array}
\eq

Notice also that
\beq\label{t438}
\begin{array}{c}
  \displaystyle{
    \tr\Big(\mA\Big)=\tr\Big(\mT^++\mT^-\Big)\,,\qquad \tr\Big(\mB\Big)=\tr\Big(\mT^+-\mT^-\Big)
  }
\end{array}
\eq
and
\beq\label{t429}
\begin{array}{c}
  \displaystyle{
    \tr\Big(\mA^2+\mB^2\Big)=2\tr\Big((\mT^+)^2\Big)+2\tr\Big((\mT^-)^2\Big)\,.
  }
\end{array}
\eq
These functions occur in (\ref{tt423}) and are constant on the chosen orbits.

\subsection{Reduction to \texorpdfstring{$\mA=\mB$}{A = B}: one-orbit model}
Impose the additional condition
\beq\label{t451}
\begin{array}{c}
  \displaystyle{
    \mA=\mB\quad{\rm or}\quad \mT^-=0\,.
  }
\end{array}
\eq
This amounts to choosing the zero orbit for $\mT^-$. It is therefore
an admissible restriction rather than a further Poisson reduction.
The brackets (\ref{t432}) are compatible with this restriction,
leaving only
\beq\label{t452}
\begin{array}{c}
  \displaystyle{
    \{\mA^{ij}_1,\mA^{kl}_2\}=P_{12}\mA^{kj}_1\delta^{il}-\mA^{il}_1P_{12}\delta^{kj}\,,
  }
\end{array}
\eq
and
\beq\label{t453}
\begin{array}{c}
  \displaystyle{
    \mS=\mats{\mA}{\mA}{\mA}{\mA}\,,\quad \mA\in{\rm Mat}_{nN}\,.
  }
\end{array}
\eq
The equations of motion (\ref{tt417})--(\ref{tt421}) are also compatible
with (\ref{t451}).

The phase space in this case is
\beq\label{t454}
\begin{array}{c}
  \displaystyle{
   T^*\mC^{n}\times\mO^{\hbox{\tiny{$nN$}}}//\mD\,,\quad
   \dim\Big(T^*\mC^{n}\times\mO^{\hbox{\tiny{$nN$}}}//\mD\Big)=nN(nN-1)+2\,,
  }
\end{array}
\eq
where the dimension is given for a generic orbit.
The Hamiltonian follows by setting $\mB=\mA$ in (\ref{tt422}):
\beq
\label{tt455}
\begin{array}{c}
  \displaystyle{
    H^{\hbox{\tiny{Gtops}}}=\frac{1}{2}\,\sum_{i=1}^{n}p_i^2-\frac{1}{2}\,\sum_{i=1}^{n}u_N(q_i|\nu)
    +\frac{1}{2}\,\sum_{i=1}^{n}\tr_{23}\Big(m_{23}(0)\mA_2^{ii}\mA_3^{ii}\Big)+
  }
  \\ \ \\
  \displaystyle{
    +\frac{1}{2}\,\sum_{i\neq j}^{n}\tr_{23}\Big(F_{32}^{0}(q_{ij})P_{23}\mA_2^{ij}\mA_3^{ji}\Big)
    +\frac{1}{2}\,\sum_{i\neq j}^{n}\tr_{23}\Big(\ti{F}_{32}^{0}(q_{ij}^{+})P_{23}
    \mA_2^{ij}\mA_3^{ji}\Big)+
  }
  \\ \ \\
  \displaystyle{
    +\frac{1}{2}\,\sum_{i=1}^{n}\tr_{23}\Big(\ti{F}_{32}^{0}(2q_i)P_{23}\mA_2^{ii}\mA_3^{ii}\Big)
    +\frac{1}{2}\,\sum_{i=1}^{n}
    \tr_2\Big(\mA_2^{ii}\ti{Y}_2^{(0)}(q_i)\Big)\,.
  }
\end{array}
\eq
The Lax pair is obtained from (\ref{tt402})--(\ref{tt410}) by setting
$\mB=\mA$. The equations of motion similarly follow from
(\ref{tt417})--(\ref{tt421}) by substituting $\mB\to\mA$.
For the off-diagonal blocks of $\mA$, this gives
\beq\label{tt456}
\begin{array}{c}
  \displaystyle{
    \dot{\mA}^{ij} =
     \sum_{k:k\neq i,j}^{n}\Big( \mA^{ik}\tr_2\Big(
    \mA_2^{kj}F_{12}^{0}(q_{kj})P_{12}\Big) -
    \tr_2\Big(F_{12}^{0}(q_{ik})P_{12}\mA_2^{ik}\Big)\mA^{kj} \Big) +
  }
  \\ \ \\
    \displaystyle{
    + \sum_{k=1}^{n}\Big(\mA^{ik}\tr_{2}\Big(\mA_{2}^{kj}
    \ti{F}_{12}^{0}(q_{kj}^{+})P_{12}\Big) -
    \tr_2\Big(\ti{F}_{12}^{0}(q_{ik}^{+})P_{12}\mA_{2}^{ik}\Big)\mA^{kj}\Big) +
  }
\end{array}
\eq
$$
\begin{array}{c}
  \displaystyle{
    + \mA^{ii}\tr_2\Big(\mA_{2}^{ij}F_{12}^{0}(q_{ij})P_{12}\Big)
    -\tr_2\Big(F_{12}^{0}(q_{ij})P_{12}\mA_{2}^{ij}\Big)\mA^{jj}+
  }
  \\ \ \\
  \displaystyle{
    + \mA^{ij}\tr_2\Big(\mA_2^{jj}m_{12}(0)\Big)
    -\tr_2\Big(m_{12}(0)\mA_2^{ii}\Big)\mA^{ij}
     +\frac{1}{2}\,\mA^{ij}\ti{Y}^{0}(q_{j})-\frac{1}{2}\,\ti{Y}^{0}(q_{i})\mA^{ij}\,.
  }
  \end{array}
$$
For the diagonal $N\times N$ blocks of $\mA$,
\beq\label{tt457}
\begin{array}{c}
  \displaystyle{
    \dot{\mA}^{ii} = [\mA^{ii},\tr_2\Big(m_{12}(0)\mA_2^{ii}\Big)] +\frac{1}{2}\,[\mA^{ii},
    \ti{Y}^{0}(q_{i})] +
  }
  \\ \ \\
  \displaystyle{
    + \sum_{k:k\neq i}^n\left(\mA^{ik}\tr_2\Big(\mA_2^{ki}F_{12}^{0}(q_{ki})P_{12}\Big)
    -\tr_2\Big(F_{12}^{0}(q_{ik})P_{12}\mA_2^{ik}\Big)\mA^{ki}\right) +
  }
  \\ \ \\
  \displaystyle{
    + \sum_{k=1}^n\left(\mA^{ik}\tr_2\Big(\mA_2^{ki}\ti{F}_{12}^{0}(q_{ki}^{+})P_{12}\Big)
    -\tr_2\Big(\ti{F}_{12}^{0}(q_{ik}^{+})P_{12}\mA_2^{ik}\Big)\mA^{ki}\right)\,.
  }
\end{array}
\eq
and for the momenta,
\beq\label{tt458}
\begin{array}{c}
  \displaystyle{
    \dot{p}_{i} = \frac{1}{2}\, \p_{q_i}u_N(q_i|\nu) -\frac{1}{2}\,
    \tr_{2}\Big(\partial\ti{Y}^{0}_{2}(q_{i})\mA^{ii}_{2}\Big)
-\tr_{23}\Big(\partial_{2q_{i}}\ti{F}^{0}_{32}(2q_{i})P_{23}
    \mA_{2}^{ii}\mA_{3}^{ii}\Big)-
  }
  \\ \ \\
  \displaystyle{
    - \sum_{k:k \neq i} \Big( \tr_{23}\Big(\partial_{q_{ik}^{+}} \ti{F}_{32}^{0} (q_{ik}^{+})P_{23}
    \mA_2^{ik} \mA_3^{ki}\Big) + \tr_{23} \Big(\partial F_{32}^{0}(q_{ik}) P_{23}
    \mA_2^{ik}\mA_3^{ki}\Big)  \Big)\,.
  }
\end{array}
\eq

\section{Case \texorpdfstring{$N=1$}{N = 1}: spin generalizations of the \texorpdfstring{${\rm BC}_n$}{BCn} Calogero--Inozemtsev model}\label{sec5}
\setcounter{equation}{0}

We now specialize the elliptic models to $N=1$ and use the notation
 \beq\label{t501}
 \begin{array}{c}
  \displaystyle{
\mS^{ij}|_{N=1}:=S_{ij}\,,\quad \mA^{ij}|_{N=1}:=A_{ij}\,,\quad \mB^{ij}|_{N=1}:=B_{ij}\,,
  }
 \end{array}
\eq
so that the spin variables form a matrix $S\in{\rm Mat}_{2n}$:
 \beq\label{t5011}
 \begin{array}{c}
  \displaystyle{
S=\mats{A}{B}{B}{A}\,,\quad A,B\in{\rm Mat}_n\,.
  }
 \end{array}
\eq
We use the scalar expressions (\ref{t263})--(\ref{t268}) for the
$R$- and $K$-matrices at $N=1$.
Here $\la=\la(\nu)$ is given by (\ref{t2641}), $f(z,x)=\p_x\phi(z,x)$,
and $v'(z,x|\la)=\p_xv(z,x|\la)$, as in (\ref{a06}) and (\ref{a46}).

\subsection{Two-spin generalization}
We call the general case $A\neq B$ the two-spin model, since its
spin variables lie on two coadjoint orbits.
Setting $N=1$ in (\ref{tt402})--(\ref{tt410}) gives the $2n\times 2n$ Lax pair
\beq\label{t502}
\begin{array}{c}
  \displaystyle{
   L^{\hbox{\tiny{2-spin}}}(z) =\sum^{n}_{i,j=1}E_{ij} {L}_{ij}(z)+\sum^{n}_{i,j=1}E_{\bar{i}j}
    {L}_{\bar{i}j}(z)+\sum^{n}_{i,j=1}E_{i\bar{j}}{L}_{i\bar{j}}(z)
    +\sum^{n}_{i,j=1}E_{\bar{i}\bar{j}}{L}_{\bar{i}\bar{j}}(z)\,.
  }
\end{array}
\eq
with entries
\beq\label{t503}
\begin{array}{c}
  \displaystyle{
    {L}_{ij}(z)=\de_{ij}\left(p_i+A_{ii}E_1(z)\right)
    +(1-\de_{ij})A_{ij}\phi(z,q_{ij})\,,
  }
\end{array}
\eq
\beq\label{t504}
\begin{array}{c}
  \displaystyle{
    {L}_{i\bar{j}}(z)=\de_{ij}\left(v(z,q_i|\la)+B_{ii}\phi(z,2q_i)\right)
    +(1-\de_{ij})B_{ij}\phi(z,q_{ij}^{+})\,,
  }
\end{array}
\eq
\beq\label{t505}
\begin{array}{c}
  \displaystyle{
    {L}_{\bar{i}j}(z)=-\de_{ij}\left(v(-z,q_i|\la)+B_{ii}\phi(-z,2q_i)\right)
    -(1-\de_{ij})B_{ij}\phi(-z,q_{ij}^{+})\,,
  }
\end{array}
\eq
\beq\label{t506}
\begin{array}{c}
  \displaystyle{
    {L}_{\bar{i}\bar{j}}(z)=-\de_{ij}\left(p_i+A_{ii}E_1(-z)\right)
    -(1-\de_{ij})A_{ij}\phi(-z,q_{ij})\,,
  }
\end{array}
\eq
The entries of ${M}^{\hbox{\tiny{2-spin}}}(z)\in \text{Mat}_{2n}$ are
\beq\label{t507}
\begin{array}{c}
  \displaystyle{
    {M}_{ij}(z)=\de_{ij}A_{ii}\frac{E_1^2(z)-\wp(z)}{2}
    +(1-\de_{ij})A_{ij}f(z,q_{ij})\,,
  }
\end{array}
\eq
\beq\label{t508}
\begin{array}{c}
  \displaystyle{
    {M}_{i\bar{j}}(z)=\de_{ij}\left(\frac{1}{2}\,v'(z,q_i|\la)+B_{ii}f(z,2q_i)\right)
    +(1-\de_{ij})B_{ij}f(z,q_{ij}^{+})\,,
  }
\end{array}
\eq
\beq\label{t509}
\begin{array}{c}
  \displaystyle{
    {M}_{\bar{i}j}(z)=\de_{ij}\left(\frac{1}{2}\,v'(-z,q_i|\la)+B_{ii}f(-z,2q_i)\right)
    +(1-\de_{ij})B_{ij}f(-z,q_{ij}^{+})\,,
  }
\end{array}
\eq
\beq\label{t510}
\begin{array}{c}
  \displaystyle{
    {M}_{\bar{i}\bar{j}}(z)=\de_{ij}A_{ii}\frac{E_1^2(-z)-\wp(-z)}{2}
    +(1-\de_{ij})A_{ij}f(-z,q_{ij})\,,
  }
\end{array}
\eq
where $v'(z,x)=\p_x v(z,x)$, as in (\ref{a46}), and
$\la=\la(\nu)$ is defined in (\ref{t2641}).

\paragraph{Hamiltonian and Poisson brackets.}
Substituting (\ref{t268})
into (\ref{tt422}), we obtain
\beq\label{t511}
\begin{array}{c}
  \displaystyle{
H^{\hbox{\tiny{2-spin}}}=\frac12\,\sum_{i=1}^n p_i^2
-\frac12\,\sum_{i\neq j}^n A_{ij}A_{ji}\wp(q_{i}-q_j)
-\frac12\,\sum_{i,j=1}^n B_{ij}B_{ji}\wp(q_{i}+q_j)-
  }
  \\ \ \\
  \displaystyle{
-\frac12\,\sum_{i=1}^n\sum_{a=0}^3
\la_a(\la_a+B_{ii})\wp(q_i+\Om_a)
+\frac{m(0)}{2}\,{\rm const}\,,
  }
\end{array}
\eq
where
\beq\label{t512}
\begin{array}{c}
  \displaystyle{
m(0)=\frac{\vth'''(0)}{3\vth'(0)}\,,\qquad {\rm const}=\sum_{i,j=1}^n\Big(A_{ij}A_{ji}+B_{ij}B_{ji}\Big)
+\Big(\sum_{a=0}^3\la_a\Big)\sum_{i=1}^n B_{ii}\,.
  }
\end{array}
\eq
The latter expression is a Casimir of the spin Poisson brackets below.
Thus, the last term in (\ref{t511}) can be omitted without changing the
equations of motion; it is retained here to give the exact scalar limit
of (\ref{tt422}).

The scalar form of the Poisson brackets (\ref{t426}) is
\beq\label{t513}
\begin{array}{c}
  \displaystyle{
\{S_{ij},S_{kl}\}=\de_{il}S_{kj}-\de_{kj}S_{il}
+\de_{i\bar l}S_{\bar k j}-\de_{\bar k j}S_{i\bar l}\,,
\qquad i,j,k,l=1,\ldots,2n\,,
  }
\end{array}
\eq
where $S_{ij}=S_{\bar i\bar j}=A_{ij}$ and
$S_{i\bar j}=S_{\bar i j}=B_{ij}$. Equivalently,
\beq\label{t514}
\begin{array}{c}
  \displaystyle{
\{A_{ij},A_{kl}\}=\de_{il}A_{kj}-\de_{kj}A_{il}\,,\qquad
\{A_{ij},B_{kl}\}=\de_{il}B_{kj}-\de_{kj}B_{il}\,,
  }
  \\ \ \\
  \displaystyle{
\{B_{ij},A_{kl}\}=\de_{il}B_{kj}-\de_{kj}B_{il}\,,\qquad
\{B_{ij},B_{kl}\}=\de_{il}A_{kj}-\de_{kj}A_{il}\,,
  }
\end{array}
\eq
with $i,j,k,l=1,\ldots,n$. Together with the canonical brackets
\beq\label{t515}
\begin{array}{c}
  \displaystyle{
\{p_i,q_j\}=\de_{ij}\,,\qquad
\{p_i,p_j\}=\{q_i,q_j\}=0\,,
  }
\end{array}
\eq
and vanishing mixed brackets between $p_i,q_i$ and $A_{jk},B_{jk}$,
these brackets define the Hamiltonian dynamics
$\frac{d}{dt} *=\{H^{\hbox{\tiny{2-spin}}},*\}$.

As in (\ref{t433}), introduce $T^\pm\in{\rm Mat}_{n}$:
\beq\label{t5151}
\begin{array}{c}
  \displaystyle{
    T^+=\frac{A+B}{2}\,,\qquad T^-=\frac{A-B}{2}\,.
  }
\end{array}
\eq
The brackets then become
\beq\label{t5152}
\begin{array}{c}
  \displaystyle{
\{T^\pm_{ij},T^\pm_{kl}\}=\de_{il}T^\pm_{kj}-\de_{kj}T^\pm_{il}
      }
  \\ \ \\
  \displaystyle{
    \{T^-_1,T^+_2\}=0\,.
      }
\end{array}
\eq
The phase space of the two-spin model is
\beq\label{t5153}
\begin{array}{c}
  \displaystyle{
   T^*\mC^{n}\times(\mO^{\hbox{\tiny{$n\!,\!+$}}}\times \mO^{\hbox{\tiny{$n\!,\!-$}}})//\mD\,,
  }
\end{array}
\eq
where $\mO^{\hbox{\tiny{$n\!,\!\pm$}}}$ are coadjoint orbits
of ${\rm GL}_{n}$, and $\mD={\rm diag}(d_1,...,d_n)$ is its
Cartan subgroup, as in (\ref{t309}). Reduction by $\mD$ imposes
the first-class constraints (\ref{t416}),
\beq\label{t5154}
\begin{array}{c}
  \displaystyle{
    A_{ii}=A_{jj}={\rm const}
  }
\end{array}
\eq
together with gauge-fixing conditions.

\paragraph{Equations of motion.}
For brevity, denote
\beq\label{t5155}
\begin{array}{c}
  \displaystyle{
    w_i=\sum_{a=0}^3\la_a\wp(q_i+\Om_a)\,.
  }
\end{array}
\eq
The scalar counterparts of (\ref{tt417}) and (\ref{tt418}), for $i\neq j$,
are
\beq\label{t516}
\begin{array}{c}
  \displaystyle{
\dot A_{ij}=\sum_{k:k\neq i,j}^n A_{ik}A_{kj}
\Big(\wp(q_{ik})-\wp(q_{kj})\Big)+\sum_{k:k\neq i,j}^n B_{ik}B_{kj}
\Big(\wp(q_{ik}^{+})-\wp(q_{kj}^{+})\Big)+
  }
  \\ \ \\
  \displaystyle{
+B_{ij}\Big((B_{jj}-B_{ii})\wp(q_{ij}^{+})
+B_{ii}\wp(2q_i)-B_{jj}\wp(2q_j)+\frac12\,(w_i-w_j)\Big)-(A_{ii}-A_{jj})A_{ij}\wp(q_{ij})\,,
  }
\end{array}
\eq
\beq\label{t517}
\begin{array}{c}
  \displaystyle{
\dot B_{ij}=\sum_{k:k\neq i,j}^n A_{ik}B_{kj}
\Big(\wp(q_{ik})-\wp(q_{kj}^{+})\Big)+\sum_{k:k\neq i,j}^n B_{ik}A_{kj}
\Big(\wp(q_{ik}^{+})-\wp(q_{kj})\Big)+
  }
  \\ \ \\
  \displaystyle{
+A_{ij}\Big((B_{jj}-B_{ii})\wp(q_{ij})
+B_{ii}\wp(2q_i)-B_{jj}\wp(2q_j)+\frac12\,(w_i-w_j)\Big)-(A_{ii}-A_{jj})B_{ij}\wp(q_{ij}^{+})\,.
  }
\end{array}
\eq
For the diagonal entries, (\ref{tt419}) and (\ref{tt420}) reduce to
\beq\label{t518}
\dot A_{ii}=0\,,
\eq
\beq\label{t519}
\begin{array}{c}
  \displaystyle{
\dot B_{ii}=\sum_{k:k\neq i}^n
\Big(A_{ik}B_{ki}-B_{ik}A_{ki}\Big)
\Big(\wp(q_{ik})-\wp(q_{ik}^{+})\Big)\,.
  }
\end{array}
\eq
Finally, the particle equations, including the scalar form of (\ref{tt421}),
are\footnote{Here $\wp'(x)=\p_x\wp(x)$; in particular, the derivative in
$\wp'(2q_i)$ is with respect to the argument $2q_i$.}
\beq\label{t520}
\begin{array}{c}
  \displaystyle{
\dot q_i=p_i\,,\qquad i=1,\ldots,n\,,
  }
  \\ \ \\
  \displaystyle{
\dot p_i=\sum_{k:k\neq i}^n
\Big(A_{ik}A_{ki}\wp'(q_{ik})+B_{ik}B_{ki}\wp'(q_{ik}^{+})\Big)
+B_{ii}^2\wp'(2q_i)+\frac12\,\sum_{a=0}^3\la_a(\la_a+B_{ii})\wp'(q_i+\Om_a)\,.
  }
\end{array}
\eq

\paragraph{Lax equation.}
On the unreduced phase space the Lax pair (\ref{t502})-(\ref{t510})
satisfies
\beq\label{t521}
\begin{array}{c}
  \displaystyle{
\dot L^{\hbox{\tiny{2-spin}}}(z)
=[L^{\hbox{\tiny{2-spin}}}(z),M^{\hbox{\tiny{2-spin}}}(z)]
+\hbox{unwanted term}\,,
  }
\end{array}
\eq
where the scalar counterpart of (\ref{t4212}) is
\beq\label{t522}
\begin{array}{c}
  \displaystyle{
\hbox{unwanted term}=
  }
  \\ \ \\
  \displaystyle{
=-\frac12\,\sum_{i\neq j}^n
(A_{ii}-A_{jj})A_{ij}E_{ij}
\p_x^2\phi(z,x)\Big|_{x=q_{ij}}-\frac12\,\sum_{i\neq j}^n
(A_{ii}-A_{jj})B_{ij}E_{i\bar j}
\p_x^2\phi(z,x)\Big|_{x=q_{ij}^{+}}+
  }
  \\ \ \\
  \displaystyle{
+\frac12\,\sum_{i\neq j}^n
(A_{ii}-A_{jj})A_{ij}E_{\bar i\bar j}
\p_x^2\phi(-z,x)\Big|_{x=q_{ij}}+
\frac12\,\sum_{i\neq j}^n
(A_{ii}-A_{jj})B_{ij}E_{\bar i j}
\p_x^2\phi(-z,x)\Big|_{x=q_{ij}^{+}}\,.
  }
\end{array}
\eq
%
The unwanted term vanishes on the constraints
\beq\label{t523}
A_{11}=\cdots=A_{nn}=\chi\,,\qquad \chi={\rm const}\,,
\eq
which are preserved by (\ref{t518}). Hence (\ref{t521}) becomes the usual
Lax equation on this invariant constraint surface.

\subsection{One-spin generalization}

The condition $A=B$ is equivalent to choosing the zero orbit for $T^-$
in (\ref{t5151}). It is preserved by (\ref{t516})-(\ref{t519}), and
the surviving spin variables form a single matrix $A\in{\rm Mat}_n$:
\beq\label{t524}
A=B=T^+\,,\qquad T^-=0\,,\qquad S=\mats{A}{A}{A}{A}\,.
\eq
This is the scalar counterpart of the one-orbit restriction (\ref{t451}).

\paragraph{Lax pair.}
The matrices $L^{\hbox{\tiny{1-spin}}}(z),M^{\hbox{\tiny{1-spin}}}(z)
\in{\rm Mat}_{2n}$ have the same structure as (\ref{t502}).
Their entries are obtained by substituting $B_{ij}=A_{ij}$ into
(\ref{t503})-(\ref{t510}). Explicitly, the entries of the Lax matrix are
\beq\label{t525}
\begin{array}{c}
  \displaystyle{
L_{ij}(z)=\de_{ij}\Big(p_i+A_{ii}E_1(z)\Big)
+(1-\de_{ij})A_{ij}\phi(z,q_{ij})\,,
  }
\end{array}
\eq
\beq\label{t526}
\begin{array}{c}
  \displaystyle{
L_{i\bar j}(z)=\de_{ij}v(z,q_i|\la)+A_{ij}\phi(z,q_{ij}^{+})\,,
  }
\end{array}
\eq
\beq\label{t527}
\begin{array}{c}
  \displaystyle{
L_{\bar i j}(z)=-\de_{ij}v(-z,q_i|\la)-A_{ij}\phi(-z,q_{ij}^{+})\,,
  }
\end{array}
\eq
\beq\label{t528}
\begin{array}{c}
  \displaystyle{
L_{\bar i\bar j}(z)=-\de_{ij}\Big(p_i+A_{ii}E_1(-z)\Big)
-(1-\de_{ij})A_{ij}\phi(-z,q_{ij})\,.
  }
\end{array}
\eq
The entries of $M^{\hbox{\tiny{1-spin}}}(z)$ are
\beq\label{t529}
\begin{array}{c}
  \displaystyle{
M_{ij}(z)=\de_{ij}A_{ii}\frac{E_1^2(z)-\wp(z)}{2}
+(1-\de_{ij})A_{ij}f(z,q_{ij})\,,
  }
\end{array}
\eq
\beq\label{t530}
\begin{array}{c}
  \displaystyle{
M_{i\bar j}(z)=\frac12\,\de_{ij}v'(z,q_i|\la)
+A_{ij}f(z,q_{ij}^{+})\,,
  }
\end{array}
\eq
\beq\label{t531}
\begin{array}{c}
  \displaystyle{
M_{\bar i j}(z)=\frac12\,\de_{ij}v'(-z,q_i|\la)
+A_{ij}f(-z,q_{ij}^{+})\,,
  }
\end{array}
\eq
\beq\label{t532}
\begin{array}{c}
  \displaystyle{
M_{\bar i\bar j}(z)=\de_{ij}A_{ii}\frac{E_1^2(-z)-\wp(-z)}{2}
+(1-\de_{ij})A_{ij}f(-z,q_{ij})\,.
  }
\end{array}
\eq
The functions $f$, $v'$ and the constants $\la$ have the same meaning
as in the two-spin case.

\paragraph{Hamiltonian and Poisson brackets.}
The restriction of (\ref{t511}), equivalently the scalar limit of
(\ref{tt455}), is
\beq\label{t533}
\begin{array}{c}
  \displaystyle{
H^{\hbox{\tiny{1-spin}}}=\frac12\,\sum_{i=1}^n p_i^2
-\frac12\,\sum_{i\neq j}^n A_{ij}A_{ji}
\Big(\wp(q_i-q_j)+\wp(q_i+q_j)\Big)-
  }
  \\ \ \\
  \displaystyle{
-\frac12\,\sum_{i=1}^n A_{ii}^2\wp(2q_i)
-\frac12\,\sum_{i=1}^n\sum_{a=0}^3
\la_a(\la_a+A_{ii})\wp(q_i+\Om_a)
+\frac{m(0)}{2}\,{\rm const}\,,
  }
\end{array}
\eq
where the constant in (\ref{t512}) is
\beq\label{t534}
\begin{array}{c}
  \displaystyle{
{\rm const}=2\sum_{i,j=1}^n A_{ij}A_{ji}
+\Big(\sum_{a=0}^3\la_a\Big)\sum_{i=1}^n A_{ii}\,.
  }
\end{array}
\eq
It is a linear combination of the Casimirs $\tr(A^2)$ and $\tr(A)$.
Consequently, the last term in (\ref{t533}) does not affect the dynamics.
Since $T^+=A$ on (\ref{t524}), the surviving brackets in (\ref{t5152})
have the standard normalization:
\beq\label{t535}
\begin{array}{c}
  \displaystyle{
\{A_{ij},A_{kl}\}=\de_{il}A_{kj}-\de_{kj}A_{il}\,,
\qquad i,j,k,l=1,\ldots,n\,,
  }
  \\ \ \\
  \displaystyle{
\{p_i,q_j\}=\de_{ij}\,,\qquad
\{p_i,p_j\}=\{q_i,q_j\}=0\,,\qquad
\{p_i,A_{jk}\}=\{q_i,A_{jk}\}=0\,.
  }
\end{array}
\eq
Thus, before the reduction by $\mD$, the spin variables lie on a single
coadjoint orbit $\mO^{\hbox{\tiny{$n$}}}$ of ${\rm GL}_n$.
Imposing the first-class constraints
\beq\label{t536}
A_{11}=\cdots=A_{nn}=\chi\,,\qquad \chi={\rm const}\,,
\eq
and quotienting by the diagonal action of $\mD$ gives the phase space
\beq\label{t537}
T^*\mC^n\times\mO^{\hbox{\tiny{$n$}}}//\mD\,.
\eq
On a fixed orbit, $\chi=\tr(A)/n$.

\paragraph{Equations of motion.}
The equations below are written before imposing (\ref{t536}), with
the convention $\frac{d}{dt}*=\{H^{\hbox{\tiny{1-spin}}},*\}$.
Substitution of $B=A$ into (\ref{t518})-(\ref{t519}) gives
\beq\label{t538}
\dot A_{ii}=0\,,\qquad i=1,\ldots,n\,.
\eq
For $i\neq j$, both (\ref{t516}) and (\ref{t517}) give the same equation:
\beq\label{t539}
\begin{array}{c}
  \displaystyle{
\dot A_{ij}=\sum_{k:k\neq i,j}^n A_{ik}A_{kj}
\Big(\wp(q_{ik})-\wp(q_{kj})+\wp(q_{ik}^{+})
-\wp(q_{kj}^{+})\Big)-
  }
  \\ \ \\
  \displaystyle{
-(A_{ii}-A_{jj})A_{ij}\Big(\wp(q_{ij})+\wp(q_{ij}^{+})\Big)+A_{ij}\Big(A_{ii}\wp(2q_i)-A_{jj}\wp(2q_j)
+\frac12\,(w_i-w_j)\Big)\,,
  }
\end{array}
\eq
where $w_i$ is defined by (\ref{t5155}). The particle equations are
\beq\label{t540}
\begin{array}{c}
  \displaystyle{
\dot q_i=p_i\,,\qquad i=1,\ldots,n\,,
  }
  \\ \ \\
  \displaystyle{
\dot p_i=\sum_{k:k\neq i}^n A_{ik}A_{ki}
\Big(\wp'(q_{ik})+\wp'(q_{ik}^{+})\Big)
+A_{ii}^2\wp'(2q_i)
+\frac12\,\sum_{a=0}^3\la_a(\la_a+A_{ii})\wp'(q_i+\Om_a)\,.
  }
\end{array}
\eq
%

\paragraph{Lax equation.}
For arbitrary diagonal entries $A_{ii}$ the equations of motion imply
\beq\label{t541}
\begin{array}{c}
  \displaystyle{
\dot L^{\hbox{\tiny{1-spin}}}(z)
=[L^{\hbox{\tiny{1-spin}}}(z),M^{\hbox{\tiny{1-spin}}}(z)]
+\hbox{unwanted term}\,,
  }
\end{array}
\eq
where
\beq\label{t542}
\begin{array}{c}
  \displaystyle{
\hbox{unwanted term}=-\frac12\,\sum_{i\neq j}^n
(A_{ii}-A_{jj})A_{ij}\Big(E_{ij}
\p_x^2\phi(z,x)\Big|_{x=q_{ij}}+E_{i\bar j}
\p_x^2\phi(z,x)\Big|_{x=q_{ij}^{+}}\Big)+
  }
  \\ \ \\
  \displaystyle{
+\frac12\,\sum_{i\neq j}^n
(A_{ii}-A_{jj})A_{ij}\Big(E_{\bar i\bar j}
\p_x^2\phi(-z,x)\Big|_{x=q_{ij}}+E_{\bar i j}
\p_x^2\phi(-z,x)\Big|_{x=q_{ij}^{+}}\Big)\,.
  }
\end{array}
\eq
%
The constraints (\ref{t536}) are preserved by (\ref{t538}) and make
the unwanted term vanish. On this invariant constraint surface,
(\ref{t541}) is the usual Lax equation.

\subsection{Reduction to the spinless \texorpdfstring{${\rm BC}_n$}{BCn} Calogero--Inozemtsev model}
We now recover the spinless Calogero--Inozemtsev model (\ref{t118}).
We retain $\nu_a$ as the parameters of the original $K$-matrix
and $\la_a=\la_a(\nu)$ as defined in (\ref{t2641}). The reduction
shifts the physical boundary couplings to $\la_a+g/2$.
In the Lax pair below we write
$v(z,q)=v(z,q|\la+g/2)$ and
$v'(z,q)=\p_qv(z,q|\la+g/2)$; the shift of the four components
of $\la$ is understood componentwise. No redefinition of $\nu_a$
or $\bar\nu_a$ is made.
The Lax representation with a spectral parameter is due to
Takasaki \cite{Ta}; see also \cite{BoSa}:
  \beq\label{t560}
  \begin{array}{c}
  \displaystyle{
L_{ij}(z)=\delta_{ij}p_i+g(1-\delta_{ij})\phi(z,q_{ij})\,,
 }
 \end{array}
 \eq
  \beq\label{t561}
  \begin{array}{c}
  \displaystyle{
L_{i{\bar j}}(z)=\delta_{ij}v(z,q_i)+g(1-\delta_{ij})\phi(z,q^+_{ij})\,,
 }
 \end{array}
 \eq
  \beq\label{t562}
  \begin{array}{c}
  \displaystyle{
L_{{\bar i}j}(z)=-\delta_{ij}v(-z,q_i)-g(1-\delta_{ij})\phi(-z,q^+_{ij})\,,
 }
 \end{array}
 \eq
  \beq\label{t563}
  \begin{array}{c}
  \displaystyle{
L_{{\bar i}{\bar j}}(z)=-\delta_{ij}p_i-g(1-\delta_{ij})\phi(-z,q_{ij})
 }
 \end{array}
 \eq
and
  \beq\label{t564}
  \begin{array}{c}
  \displaystyle{
M_{ij}(z)=\delta_{ij}a_i+g(1-\delta_{ij})f(z,q_{ij})\,,
 }
 \end{array}
 \eq
  \beq\label{t565}
  \begin{array}{c}
  \displaystyle{
M_{i{\bar j}}(z)=\frac12\,\delta_{ij}v'(z,q_i)+g(1-\delta_{ij})f(z,q^+_{ij})\,,
 }
 \end{array}
 \eq
  \beq\label{t566}
  \begin{array}{c}
  \displaystyle{
M_{{\bar i}j}(z)=\frac12\,\delta_{ij}v'(-z,q_i)+g(1-\delta_{ij})f(-z,q^+_{ij})\,,
 }
 \end{array}
 \eq
  \beq\label{t567}
  \begin{array}{c}
  \displaystyle{
M_{{\bar i}{\bar j}}(z)=\delta_{ij}a_i+g(1-\delta_{ij})f(-z,q_{ij})\,,
 }
 \end{array}
 \eq
where
  \beq\label{t568}
  \begin{array}{c}
  \displaystyle{
a_i=g\sum\limits_{k:k\neq i}^n\Big(\wp(q_{i}-q_k)+\wp(q_i+q_k)\Big)
+\frac12\sum\limits_{a=0}^3(\la_a+g/2)\wp(q_i+\Omega_a)\,.
 }
 \end{array}
 \eq

\paragraph{Minimal orbit and reduction.}
For $g\neq0$, choose in (\ref{t537}) the rank-one semisimple
orbit\footnote{The factorization $A_{ij}=\xi_i\eta_j$ is defined up to the common rescaling
$(\xi,\eta)\mapsto(c\xi,c^{-1}\eta)$, $c\in\mC^\times$.}
\beq\label{t569}
\begin{array}{c}
  \displaystyle{
\mO_g=\Big\{A=\xi\eta^{\mathsf T}\,:\,
\eta^{\mathsf T}\xi=ng\Big\}\,,\qquad
A_{ij}=\xi_i\eta_j\,,\qquad \xi,\eta\in\mC^n\,.
  }
\end{array}
\eq
Its eigenvalues are $0,\ldots,0,ng$, and $\dim\mO_g=2n-2$.
The canonical brackets
\beq\label{t570}
\{\xi_i,\eta_j\}=\de_{ij}\,,\qquad
\{\xi_i,\xi_j\}=\{\eta_i,\eta_j\}=0
\eq
induce exactly the brackets for $A$ in (\ref{t535}). Fixing
$\sum_i\xi_i\eta_i=ng$ and quotienting by the common rescaling
gives the orbit (\ref{t569}). The variables $q_i,p_i$ commute with
$\xi_j,\eta_j$.

The Cartan subgroup acts by $A\mapsto DAD^{-1}$, or equivalently
by $\xi\mapsto D\xi$ and $\eta\mapsto D^{-1}\eta$, where $D$ is diagonal.
The moment-map constraints (\ref{t536}) become
\beq\label{t571}
A_{ii}=\xi_i\eta_i=g\,,\qquad
\eta_i=\frac{g}{\xi_i}\,,\qquad
A_{ij}=g\frac{\xi_i}{\xi_j}\,,\qquad i,j=1,\ldots,n\,.
\eq
%
Since $g\neq0$, all $\xi_i$ are nonzero on this level set.
The diagonal gauge transformation
\beq\label{t572}
\begin{array}{c}
  \displaystyle{
D={\rm diag}(\xi_1^{-1},\ldots,\xi_n^{-1})\,,\qquad
A\to DAD^{-1}=g\,e e^{\mathsf T}\,,\qquad
e=(1,\ldots,1)^{\mathsf T}
  }
\end{array}
\eq
therefore fixes all spin variables: $A_{ij}\to A_{ij}=g$.
The spin quotient is a point, and the reduced phase space is locally
\beq\label{t573}
\Big(T^*\mC^n\times\mO_g\Big)//\mD_{\rm eff}
\simeq T^*\mC^n\,,\qquad
\mD_{\rm eff}=\mD/\mC^*\,.
\eq
The dimension of the spin quotient (the spin part of the phase space) is
$(2n-2)-2(n-1)=0$. Since the constraints and the gauge conditions
involve only spin variables, the reduced brackets of $p_i,q_i$
remain canonical, as in (\ref{t515}).

\paragraph{The reduced Hamiltonian.}
The first component of (\ref{a21}), with its two arguments interchanged,
and the symmetry of the Kronecker function give
\beq\label{t574}
\phi(z,2q)=\frac12\,\sum_{a=0}^3\vf_a(2z,q+\Om_a)\,.
\eq
Consequently, the shifted boundary couplings satisfy
\beq\label{t575}
\Big(\la_a+\frac{g}{2}\Big)^2=\la_a(\la_a+g)+\frac{g^2}{4}\,,\qquad a=0,1,2,3\,.
\eq
Using the definition (\ref{a45}) and differentiating with respect to $q$
yields
\beq\label{t576}
\begin{array}{c}
  \displaystyle{
v(z,q|\la)+g\phi(z,2q)=v(z,q|\la+g/2)\,,
  }
  \\ \ \\
  \displaystyle{
\frac12\,v'(z,q|\la)+g f(z,2q)=\frac12\,v'(z,q|\la+g/2)\,.
  }
\end{array}
\eq
These identities also hold with $z$ replaced by $-z$.

On the constraint surface (\ref{t571}), $A_{ij}A_{ji}=g^2$,
$\tr(A)=ng$, and $\tr(A^2)=n^2g^2$.
Substituting into (\ref{t533}) and using (\ref{a24}) gives
\beq\label{t577}
H^{\hbox{\tiny{1-spin}}}\Big|_{(\ref{t571})}
=H^{\hbox{\tiny{CI}}}
+\frac{m(0)}{2}\,\Big(2n^2g^2+ng\sum_{a=0}^3\la_a(\nu)\Big)\,,
\eq
where
\beq\label{t578}
\begin{array}{c}
  \displaystyle{
H^{\hbox{\tiny{CI}}}=\frac12\,\sum_{i=1}^n p_i^2
-g^2\sum_{i<j}^n\Big(\wp(q_i-q_j)+\wp(q_i+q_j)\Big)-
\frac12\,\sum_{i=1}^n\sum_{a=0}^3\big(\la_a+g/2\big)^2\wp(q_i+\Om_a)\,.
  }
\end{array}
\eq
Indeed, the boundary terms combine because
$g^2\wp(2q)=\frac{g^2}{4}\sum_a\wp(q+\Om_a)$ and
$\la_a(\la_a+g)+g^2/4=\big(\la_a+g/2\big)^2$.
Thus, (\ref{t578}) has the form (\ref{t118}) with its physical
boundary couplings replaced by $\la_a+g/2$. Here $\la$ is still
determined by the original $K$-matrix parameters $\nu$ through
(\ref{t2641}); this replacement does not identify $\bar\nu_a$
with $\la_a+g/2$ for the same original parameters.
The second term on the right-hand side of (\ref{t577}) is a constant.
The parameters $g,\nu_0,\nu_1,\nu_2,\nu_3$ remain independent.
The reduced equations are
\beq\label{t579}
\begin{array}{c}
  \displaystyle{
\dot q_i=p_i\,,\qquad
\dot p_i=g^2\sum_{k:k\neq i}^n
\Big(\wp'(q_i-q_k)+\wp'(q_i+q_k)\Big)
+\frac12\,\sum_{a=0}^3\big(\la_a+g/2\big)^2\wp'(q_i+\Om_a)\,.
  }
\end{array}
\eq

\paragraph{Gauge transformation of the Lax pair.}
On the level set (\ref{t571}), the diagonal coefficients (\ref{t568})
can also be written in terms of the original scalar parameters as
\beq\label{t580}
\begin{array}{c}
  \displaystyle{
a_i=g\sum_{k:k\neq i}^n
\Big(\wp(q_i-q_k)+\wp(q_i+q_k)\Big)
+g\wp(2q_i)+\frac12\,w_i\,,
  }
\end{array}
\eq
where $w_i$ is given by (\ref{t5155}). The equality between
(\ref{t568}) and (\ref{t580}) follows from (\ref{t575}) and
(\ref{a24}). Since $A_{ik}A_{kj}=gA_{ij}$ on (\ref{t571}),
the spin equations (\ref{t538})-(\ref{t539}) reduce to
\beq\label{t581}
\dot A_{ij}=(a_i-a_j)A_{ij}\,,\qquad
\dot A=[{\rm diag}(a_1,\ldots,a_n),A]\,.
\eq
In this calculation the $k=i,j$ terms cancel in the difference of
the two sums in (\ref{t580}), by the evenness of $\wp$.
Equation (\ref{t581}) determines the time derivatives of the ratios
$\xi_i/\xi_j$. The common rescaling of $\xi$ can be chosen so that
$\dot\xi_i/\xi_i=a_i$. For the gauge matrix (\ref{t572}) and its
lift to the $2n\times2n$ Lax matrices, this gives
\beq\label{t582}
\begin{array}{c}
  \displaystyle{
G=\mats{D}{0}{0}{D}\,,\qquad
-\dot{G}G^{-1}
={\rm diag}(a_1,\ldots,a_n,a_1,\ldots,a_n)\,.
  }
\end{array}
\eq
The transformed Lax pair is
%
\beq\label{t583}
\begin{array}{c}
  \displaystyle{
L^{\hbox{\tiny{1-spin}}}(z)\to
G L^{\hbox{\tiny{1-spin}}}(z) G^{-1}\Big|_{\hbox{\tiny{on-shell (\ref{t571})}}}=L^{\hbox{\tiny{red}}}(z)\,,
  }
  \\ \ \\
  \displaystyle{
M^{\hbox{\tiny{1-spin}}}(z)\to \Big(G M^{\hbox{\tiny{1-spin}}}(z)G^{-1}
-\dot{G}G^{-1}\Big)\Big|_{\hbox{\tiny{on-shell (\ref{t571})}}}=M^{\hbox{\tiny{red}}}(z)\,.
  }
\end{array}
\eq
 By (\ref{t571}), the unwanted term (\ref{t542})
vanishes.
Finally, remove the scalar spectral terms:
\beq\label{t584}
\begin{array}{c}
  \displaystyle{
L^{\hbox{\tiny{red}}}(z)\to L^{\hbox{\tiny{red}}}(z)-gE_1(z)1_{2n}\,,
  }
  \\ \ \\
  \displaystyle{
M^{\hbox{\tiny{red}}}(z)\to  M^{\hbox{\tiny{red}}}(z)-g m(z)1_{2n}\,,\qquad
m(z)=\frac{E_1^2(z)-\wp(z)}{2}\,.
  }
\end{array}
\eq
Here $E_1(-z)=-E_1(z)$ and $m(-z)=m(z)$, so the subtracted terms
are scalar in the full $2n$-dimensional space.
The diagonal entries become
\beq\label{t585}
\begin{array}{c}
  \displaystyle{
L^{\hbox{\tiny{red}}}_{ii}(z)=p_i\,,\qquad
L^{\hbox{\tiny{red}}}_{\bar i\bar i}(z)=-p_i\,,\qquad
M^{\hbox{\tiny{red}}}_{ii}(z)
=M^{\hbox{\tiny{red}}}_{\bar i\bar i}(z)=a_i\,.
  }
\end{array}
\eq
Every off-diagonal spin coefficient becomes $g$ by (\ref{t572}),
while (\ref{t576}) transforms the entries with indices $i,\bar i$
and $\bar i,i$ into the corresponding $v$- and $v'$-terms.
Thus, all entries of (\ref{t584}) coincide exactly with
(\ref{t560})-(\ref{t567}), with $a_i$ given by (\ref{t568}), and
\beq\label{t586}
\dot L^{\hbox{\tiny{red}}}(z)
=[L^{\hbox{\tiny{red}}}(z),M^{\hbox{\tiny{red}}}(z)]\,.
\eq
In particular, the diagonal term $a_i$ in Takasaki's $M$-matrix
arises from the time derivative of the gauge transformation.

\section{Special case: \texorpdfstring{$n$}{n} interacting tops}\label{sec6}
\setcounter{equation}{0}

We consider the rank-one case of the one-orbit model
(\ref{t451})--(\ref{tt458}). Let $\mS$ in (\ref{t453}) have rank one,
so that $\mA$ also has rank one. The Lax pair is that of the
one-orbit model, obtained from (\ref{tt402})--(\ref{tt410})
by replacing $\mB$ with $\mA$.

\subsection{Hamiltonian description}
As for the $A$-type model (\ref{t370})--(\ref{t380}), parameterize $\mA$ as
  \beq\label{t601}
  \begin{array}{c}
  \displaystyle{
\mA^{ij}_{ab}=\xi^i_a\eta^j_b\,,\quad i,j=1,...,n;\ a,b=1,...,N\,.
 }
 \end{array}
 \eq
This corresponds to the minimal coadjoint orbit
$\mO^{\hbox{\tiny{min,$n\!N$}}}$, of dimension
  \beq\label{t602}
  \begin{array}{c}
  \displaystyle{
\dim\mO^{\hbox{\tiny{min,$n\!N$}}}=2nN-2\,.
 }
 \end{array}
 \eq
The Poisson brackets for the components of $\xi^i$ and $\eta^j$ are canonical:
  \beq\label{t603}
  \begin{array}{c}
  \displaystyle{
\{\xi^i_a,\eta^j_b\}=\delta^{ij}\delta_{ab}\,,\quad i,j=1,...,n;\ a,b=1,...,N\,.
 }
 \end{array}
 \eq
The subtraction of two dimensions in (\ref{t602}) comes from Hamiltonian
reduction by the common rescaling
$(\xi,\eta)\to(\epsilon\xi,\epsilon^{-1}\eta)$, $\epsilon\in\mC^*$.
This imposes $\tr(\mA)=ng$ and quotients by the $\mC^*$ action.

The subsequent reduction
$\mO^{\hbox{\tiny{min,$n\!N$}}}\to
\mO^{\hbox{\tiny{min,$n\!N$}}}//\mD$ decreases the dimension by
$2n-2$.\footnote{The common $\mC^*$ acts trivially by conjugation,
so the effective group is $\mD/\mC^*$, as in (\ref{t573}).}
The phase space is therefore
  \beq\label{t604}
  \begin{array}{c}
  \displaystyle{
T^*\mC^n\times\mO^{\hbox{\tiny{min,$n\!N$}}}//\mD\,,\qquad
\dim\Big(T^*\mC^n\times\mO^{\hbox{\tiny{min,$n\!N$}}}//\mD\Big)=2nN\,.
 }
 \end{array}
 \eq
The Poisson structure for the spin part of the phase space is
  \beq\label{t6041}
  \begin{array}{c}
  \displaystyle{
\{\mA^{ii}_1,\mA^{jj}_2\}
   =\delta^{ij}[P_{12},\mA^{ii}_1]\quad {\rm or}\quad
   \{\mA^{ii}_{ab},\mA^{jj}_{cd}\}=\delta^{ij}\Big(\delta_{ad}\mA^{ii}_{cb}-\delta_{bc}\mA^{ii}_{ad}\Big)\,.
 }
 \end{array}
 \eq

Applying the same rearrangement as in (\ref{t374}) to the one-orbit
Hamiltonian (\ref{tt455}) gives
\beq\label{t605}
\begin{array}{c}
  \displaystyle{
    H^{\hbox{\tiny{tops}}}=\frac{1}{2}\,\sum_{i=1}^{n}p_i^2-\frac{1}{2}\,\sum_{i=1}^{n}u_N(q_i|\nu)
    +
  }
  \\ \ \\
  \displaystyle{
  +\frac{1}{2}\,\sum_{i=1}^{n}\Big(
  \tr_{23}\Big(m_{23}(0)\mA_2^{ii}\mA_3^{ii}\Big)
   +\tr_{23}\Big(\ti{F}_{23}^{0}(2q_i)\mA_2^{ii}\mA_3^{ii}\Big)
    +
    \tr\Big(\mA^{ii}\ti{Y}^{(0)}(q_i)\Big)\Big)+
  }
  \\ \ \\
  \displaystyle{
    +\frac{1}{2}\,\sum_{i\neq j}^{n}\Big( \tr_{23}\Big(F_{23}^{0}(q_i-q_j)\mA_2^{ii}\mA_3^{jj}\Big)
    +\tr_{23}\Big(\ti{F}_{23}^{0}(q_i+q_j)
    \mA_2^{ii}\mA_3^{jj}\Big) \Big),
  }
\end{array}
\eq
where we used the identity valid for every rank-one matrix $\mA^{ii}$:
\beq\label{t6052}
\begin{array}{c}
  \displaystyle{
P_{23}\mA_2^{ii}\mA_3^{ii}=\mA_2^{ii}\mA_3^{ii}\,.
  }
\end{array}
\eq
This Hamiltonian describes $n$ interacting ${\rm gl}_N$ tops attached
to $n$ particles. Each top is described by
$\mO^{\hbox{\tiny{min,$N$}}}$, as in (\ref{t380}).
For the pair interactions, introduce the linear operators
\beq\label{t6051}
\begin{array}{c}
  \displaystyle{
 {\mathbb F}_x(T)=\tr_2\Big(F_{12}^{0}(x)T_2\Big)\,,\qquad
 \ti{\mathbb F}_x(T)=\tr_2\Big(\ti{F}_{12}^{0}(x)T_2\Big)\,.
  }
\end{array}
\eq
Unlike ${\mathbb{FP}}_x$ and ${\mathbb{{\ti F}P}}_x$ in (\ref{t4214}), these operators
do not contain $P_{12}$: the permutation has already been eliminated from the
pair interactions by the rank-one substitution (\ref{t601}). Keeping
${\mathbb J}$ as defined in (\ref{t4214}), the Hamiltonian
(\ref{t605}) takes the form
\beq\label{t606}
\begin{array}{c}
  \displaystyle{
    H^{\hbox{\tiny{tops}}}=\frac{1}{2}\,\sum_{i=1}^{n}p_i^2-\frac{1}{2}\,\sum_{i=1}^{n}u_N(q_i|\nu)
    +
  }
  \\ \ \\
  \displaystyle{
  +\frac{1}{2}\,\sum_{i=1}^{n}\Big(
  \tr\Big({\mathbb J}(\mA^{ii})\mA^{ii}\Big)
   +\tr\Big({\mathbb{{\ti F}}}_{2q_i}(\mA^{ii})\mA^{ii}\Big)
    +
    \tr\Big(\mA^{ii}\ti{Y}^{(0)}(q_i)\Big)\Big)+
  }
  \\ \ \\
  \displaystyle{
    +\frac{1}{2}\,\sum_{i\neq j}^{n}\Big(\tr\Big({{\mathbb F}}_{q_i-q_j}(\mA^{jj})\mA^{ii}\Big)+
     \tr\Big({\ti{\mathbb F}}_{q_i+q_j}(\mA^{jj})\mA^{ii}\Big)\Big)\,.
  }
\end{array}
\eq
The equations of motion are
\beq\label{t608}
\begin{array}{c}
 \displaystyle{
 \dot{\mA}^{ii}
 =\Big[\mA^{ii},\,{\mathbb J}(\mA^{ii})
       +{\mathbb{{\ti F}}}_{2q_i}(\mA^{ii})
       +\frac12\,\ti{Y}^{(0)}(q_i)\Big]
 +\sum_{k:k\neq i}^{n}\Big[\mA^{ii},\,
       {\mathbb F}_{q_i-q_k}(\mA^{kk})
       +\ti{\mathbb F}_{q_i+q_k}(\mA^{kk})\Big]\,,
 }
\end{array}
\eq
\beq\label{t609}
\begin{array}{c}
 \displaystyle{
 \dot p_i=\frac12\,\partial_{q_i}u_N(q_i|\nu)
    -\frac12\,\tr\Big(\mA^{ii}\partial_{q_i}\ti{Y}^{(0)}(q_i)\Big)
    -\tr_{12}\Big(\p\ti{F}_{12}^{0}(2q_i)
                      \mA^{ii}_1\mA^{ii}_2\Big)
 }
 \\ \ \\
 \displaystyle{
 \hphantom{\dot p_i=}
    -\sum_{k:k\neq i}^{n}\tr_{12}\Big(
       \big(\p F_{12}^{0}(q_i-q_k)
          +\p\ti{F}_{12}^{0}(q_i+q_k)\big)
        \mA^{ii}_1\mA^{kk}_2\Big)\,.
 }
\end{array}
\eq
and ${\dot q}_i=p_i$. Here $\p F_{12}^{0}(x)=\p_xF_{12}^{0}(x)$ and similarly for ${\ti F}_{12}^{0}(x)$.

\subsection{\texorpdfstring{$N=2$}{N = 2} XYZ case: interacting Euler tops}
For $N=2$, the Hamiltonian (\ref{t605}) describes interacting Euler tops
with position-dependent self-interaction and a boundary field. Below,
all its terms are evaluated in the Pauli basis (\ref{t259}) using the
$R$- and $K$-matrices (\ref{t260}) and (\ref{t261}).
Since $J=\sigma_0$, one has $\ti R=R$ and $\ti F^0=F^0$.

Write the diagonal blocks as\footnote{$\tr(\sigma_a\sigma_b)=2\delta_{ab}$, including $a,b=0$.}
\beq\label{t610}
\begin{array}{c}
 \displaystyle{
 \mA^{ii}=s_0^i\sigma_0+\sum_{a=1}^3s_a^i\sigma_a\,,\qquad
 s_a^i=\frac12\,\tr\Big(\mA^{ii}\sigma_a\Big)\,,\quad a=0,1,2,3\,.
 }
\end{array}
\eq
The Poisson brackets (\ref{t6041}) give
\beq\label{t611}
\begin{array}{c}
 \displaystyle{
 \{s_a^i,s_b^j\}=\imath\delta^{ij}\sum_{c=1}^3
 \epsilon_{abc}s_c^i\,,\qquad \{s_0^i,s_a^j\}=0\,,\qquad a,b=1,2,3\,.
 }
\end{array}
\eq
%
The rank-one condition implies\footnote{On the reduction level $\tr(\mA^{ii})=g$ one has $s_0^i=g/2$ for every
$i$.}
$
 \det\mA^{ii}=(s_0^i)^2-\sum_{a=1}^3(s_a^i)^2=0\,.
$

\paragraph{Coefficients of the \texorpdfstring{$R$}{R}-matrix.}
Recall that $\Om_0=0$, $\Om_1=1/2$, $\Om_2=(1+\tau)/2$ and
$\Om_3=\tau/2$. All coefficients below are expressed through the
functions $\vf_a$, $E_1$, $E_2$ and $\wp$ used in Section \ref{sec2}.
We use the following abbreviations and identities, which follow
from (\ref{a30}) and (\ref{a31}):
\beq\label{t612}
\begin{array}{c}
 \displaystyle{
 c_\tau=\frac{\vth'''(0)}{3\vth'(0)}\,,\qquad
 E_2(x)=\wp(x)-c_\tau\,,\qquad
 E_1(\Om_a)=-2\pi\imath\p_\tau\Om_a\,,
 }
 \\ \ \\
 \displaystyle{
 \vf_a(x):=\vf_a(x,\Om_a)\,,\qquad
 \vf'_a(x):=\p_x\vf_a(x)\,,\qquad a=1,2,3\,,
 }
 \\ \ \\
 \displaystyle{
 \vf'_a(x)
   =\vf_a(x)\Big(E_1(x+\Om_a)-E_1(x)-E_1(\Om_a)\Big)=-\frac{\vf_1(x)\vf_2(x)\vf_3(x)}{\vf_a(x)}\,.
 }
\end{array}
\eq
The prime denotes differentiation with respect to the argument $x$.
The functions $\vf_a(x)$ are odd, so $\vf'_a(x)$ are even.
Specializing (\ref{t2431}) and (\ref{t2392}) to $N=2$ gives
\beq\label{t613}
\begin{array}{c}
 \displaystyle{
 m_{12}(0)=\frac14\Big(c_\tau\sigma_0\otimes\sigma_0
             -\sum_{a=1}^3E_2(\Om_a)
                            \sigma_{4-a}\otimes\sigma_{4-a}\Big)=\frac{c_\tau}{2}\,P_{12}-\frac14\sum_{a=1}^3\wp(\Om_a)
                            \sigma_{4-a}\otimes\sigma_{4-a}\,,
 }
 \\ \ \\
 \displaystyle{
 F_{12}^{0}(x)=\ti{F}_{12}^{0}(x)
   =\frac12\Big(-E_2(x)\sigma_0\otimes\sigma_0
             +\sum_{a=1}^3\vf'_a(x)
                        \sigma_{4-a}\otimes\sigma_{4-a}\Big)\,.
 }
\end{array}
\eq
For a matrix $T=t_0\sigma_0+\sum_{a=1}^3t_a\sigma_a$, the operators
entering (\ref{t606}) are therefore
\beq\label{t614}
\begin{array}{c}
 \displaystyle{
 {\mathbb J}(T)=\frac{c_\tau}{2}\,t_0\sigma_0
       -\frac12\sum_{a=1}^3E_2(\Om_a)t_{4-a}\sigma_{4-a}\,,
 }
 \\ \ \\
 \displaystyle{
 {\mathbb F}_x(T)=\ti{\mathbb F}_x(T)
       =-E_2(x)t_0\sigma_0
        +\sum_{a=1}^3\vf'_a(x)
                         t_{4-a}\sigma_{4-a}\,.
 }
\end{array}
\eq
\paragraph{Quadratic spin terms.}
The Euler-top contribution is
\beq\label{t615}
\begin{array}{c}
 \displaystyle{
 \frac12\,\tr\Big({\mathbb J}(\mA^{ii})\mA^{ii}\Big)
   =\frac{c_\tau}{2}(s_0^i)^2
      -\frac12\sum_{a=1}^3E_2(\Om_a)(s_{4-a}^i)^2
   =-\frac12\sum_{a=1}^3\wp(\Om_a)(s_{4-a}^i)^2
       +c_\tau(s_0^i)^2\,.
 }
\end{array}
\eq
Up to a constant on the orbit,
this is the Euler-top Hamiltonian with quadratic coefficients
$-\wp(\Om_a)/2$ for the components $s_{4-a}^i$.

The self-interaction term follows directly from the second line of
(\ref{t614}):
\beq\label{t616}
\begin{array}{c}
 \displaystyle{
 \frac12\,\tr\Big(\ti{\mathbb F}_{2q_i}(\mA^{ii})\mA^{ii}\Big)
   =-E_2(2q_i)(s_0^i)^2
     +\sum_{a=1}^3\vf'_a(2q_i)
                       (s_{4-a}^i)^2\,.
 }
\end{array}
\eq
For the interaction between distinct tops, we obtain
\beq\label{t617}
\begin{array}{c}
 \displaystyle{
 \frac12\,\tr\Big({\mathbb F}_x(\mA^{jj})\mA^{ii}\Big)
  =\frac12\,\tr\Big(\ti{\mathbb F}_x(\mA^{jj})\mA^{ii}\Big)
  =-E_2(x)s_0^is_0^j
    +\sum_{a=1}^3\vf'_a(x)
                       s_{4-a}^is_{4-a}^j\,,
 }
\end{array}
\eq
with $x=q_i-q_j$ and $x=q_i+q_j$, respectively. Thus the ordered sum
in (\ref{t606}) can be replaced by twice the sum over $i<j$.

\paragraph{Boundary terms.}
The constants in (\ref{t261}) are explicitly
\beq\label{t618}
\begin{array}{c}
 \displaystyle{
 \rho_0=\bar\nu_0=\frac12(\nu_0+\nu_1+\nu_2+\nu_3)\,,\qquad
 \rho_1=-\bar\nu_1=\frac12(-\nu_0-\nu_1+\nu_2+\nu_3)\,,
 }
 \\ \ \\
 \displaystyle{
 \rho_2=\bar\nu_2=\frac12(\nu_0-\nu_1+\nu_2-\nu_3)\,,\qquad
 \rho_3=\bar\nu_3=\frac12(\nu_0-\nu_1-\nu_2+\nu_3)\,.
 }
\end{array}
\eq
Consequently, (\ref{t256}) becomes
\beq\label{t619}
 u_2(q|\nu)=\sum_{k=0}^3\rho_k^2\wp(q+\Om_k)=\sum_{k=0}^3{\bar\nu}_k^2\wp(q+\Om_k)\,.
\eq
Expanding the first expression in (\ref{t261}) at $z=0$ gives
$\ti K^z(q)=\rho_0\sigma_0/z+\ti K^{(0)}(q)+O(z)$, where
\beq\label{t621}
\begin{array}{c}
 \displaystyle{
 \ti K^{(0)}(q)=\rho_0E_1(q)\sigma_0
          +\sum_{a=1}^3\rho_a\vf_a(q-\Om_a)\sigma_{4-a}\,,
 }
 \\ \ \\
 \displaystyle{
 \ti Y^{(0)}(q)=\p_q\ti K^{(0)}(q)
       =-\rho_0E_2(q)\sigma_0
          +\sum_{a=1}^3\rho_a\vf'_a(q-\Om_a)\sigma_{4-a}\,.
 }
\end{array}
\eq
For comparison with the coefficients of (\ref{t261}), quasiperiodicity
and $E_1(\Om_a)=-2\pi\imath\p_\tau\Om_a$ give
\beq\label{t6211}
\begin{array}{c}
 \displaystyle{
 \vf'_a(q-\Om_a)
   =-\vf_a(q-\Om_a)
        \Big(E_1(q+\Om_a)-E_1(q)-E_1(\Om_a)\Big)\,.
        }
\end{array}
\eq
Thus the boundary contribution in (\ref{t606}) is
\beq\label{t622}
 \frac12\,\tr\Big(\mA^{ii}\ti Y^{(0)}(q_i)\Big)
   =-\rho_0E_2(q_i)s_0^i
      +\sum_{a=1}^3\rho_a\vf'_a(q_i-\Om_a)s_{4-a}^i\,.
\eq
The ordering of the components $s_{4-a}^i$ is determined by the Pauli
matrices $\sigma_{4-a}$ in (\ref{t261}); the coupling constants are
$\rho_a$, as defined in (\ref{t618}).

\paragraph{Explicit expression for the Hamiltonian.}
Collecting (\ref{t615})--(\ref{t622}) gives the following exact expression
on the rank-one orbits:
\beq\label{t623}
\begin{array}{c}
 \displaystyle{
 H^{\hbox{\tiny{Euler tops}}}
  =\frac12\sum_{i=1}^np_i^2
    -\frac12\sum_{i=1}^n\sum_{k=0}^3\rho_k^2\wp(q_i+\Om_k)-
 }
 \\ \ \\
 \displaystyle{
-\sum_{i=1}^n\Big(
       \wp(2q_i)(s_0^i)^2+\rho_0\wp(q_i)s_0^i\Big)-
       2\sum\limits_{i>j}^n s_0^is_0^j
       \Big(\wp(q_i-q_j)+\wp(q_i+q_j)\Big)
  }
 \\ \ \\
 \displaystyle{
        +\sum_{i=1}^n\sum_{a=1}^3\Big(
       \vf'_a(2q_i)-\frac12\,\wp(\Om_a)\Big)
                      (s_{4-a}^i)^2 +\sum_{i=1}^n\sum_{a=1}^3\rho_a\vf'_a(q_i-\Om_a)s_{4-a}^i
 }
 \\ \ \\
 \displaystyle{
 +2\sum_{1\leq i<j\leq n}\sum_{a=1}^3\Big(
       \vf'_a(q_i-q_j)+\vf'_a(q_i+q_j)\Big)
                       s_{4-a}^is_{4-a}^j+C_\tau\,.
 }
\end{array}
\eq
The resulting additive term, independent of the positions and the
noncentral spin components, is
\beq\label{t624}
\begin{array}{c}
 \displaystyle{
 C_\tau=c_\tau\Big(2\Big(\sum_{i=1}^ns_0^i\Big)^2
                         +\rho_0\sum_{i=1}^ns_0^i\Big)\,.
 }
\end{array}
\eq
It is constant on the chosen reduction level. The pure Euler term is
$-\frac12\sum_{i,a}\wp(\Om_a)(s_{4-a}^i)^2$; the other spin terms describe the
position-dependent self-interaction, the boundary field and the pair
interactions.

\subsection{\texorpdfstring{$N=2$}{N = 2} XXZ case: trigonometric model}
The trigonometric XXZ $R$-matrix associated with the quantum affine
algebra ${\hat{\mathcal U}}_q({\rm gl}_2)$ has the form \cite{Jimbo}:
\beq\label{t640}
\begin{array}{c}
 \displaystyle{
 R_{12}^{\hbar}(z)
 =\frac12\Big(\pi\cot(\pi\hbar)+\pi\cot(\pi z)\Big)
       \Big(\sigma_0\otimes\sigma_0+\sigma_3\otimes\sigma_3\Big)+
 }
 \\ \ \\
 \displaystyle{
 +\frac{\pi}{2\sin(\pi\hbar)}
       \Big(\sigma_0\otimes\sigma_0-\sigma_3\otimes\sigma_3\Big)
 +\frac{\pi}{2\sin(\pi z)}
       \Big(\sigma_1\otimes\sigma_1+\sigma_2\otimes\sigma_2\Big)\,.
 }
\end{array}
\eq
It is obtained from Baxter's eight-vertex $R$-matrix in the limit
${\rm Im}\,\tau\to\infty$. The corresponding $K$-matrix is
\beq\label{t641}
\begin{array}{c}
 \displaystyle{
 \ti K^{\hbar}(z)
 =\pi\nu_0\Big(\cot(\pi z)+\cot(\pi\hbar)\Big)\sigma_0
 +\pi\nu_3\Big(\tan(\pi z)+\tan(\pi\hbar)\Big)\sigma_3+
 }
 \\ \ \\
 \displaystyle{
 +\pi\sin\big(\pi(\hbar+z)\big)
       \Big(\nu_1\sigma_1+\nu_2\sigma_2\Big)\,.
 }
\end{array}
\eq
Up to multiplication by a scalar function, this is the trigonometric
$K$-matrix obtained by T. Inami and H. Konno in \cite{IK}.
Throughout this subsection, the four parameters $\nu_a$ are those
appearing directly in (\ref{t641}).
As in the XYZ case, $J=\sigma_0$, so $\ti R=R$ and $\ti F^0=F^0$.
We use the spin variables (\ref{t610}) with the Poisson brackets
(\ref{t6041}); in particular,
$\sum_{a=1}^3(s_a^i)^2=(s_0^i)^2$ on the rank-one orbits.

\paragraph{Coefficients of the \texorpdfstring{$R$}{R}-matrix.}
Expanding (\ref{t640}) as
$R^{\hbar}_{12}(x)=\hbar^{-1}\sigma_0\otimes\sigma_0
+r_{12}(x)+\hbar m_{12}(x)+O(\hbar^2)$ gives
\beq\label{t642}
\begin{array}{c}
 \displaystyle{
 r_{12}(x)=\frac{\pi}{2}\cot(\pi x)
       \Big(\sigma_0\otimes\sigma_0+\sigma_3\otimes\sigma_3\Big)
 +\frac{\pi}{2\sin(\pi x)}
       \Big(\sigma_1\otimes\sigma_1+\sigma_2\otimes\sigma_2\Big)\,,
 }
 \\ \ \\
 \displaystyle{
 m_{12}(x)=m_{12}(0)
 =-\frac{\pi^2}{12}\,\sigma_0\otimes\sigma_0
  -\frac{\pi^2}{4}\,\sigma_3\otimes\sigma_3\,.
 }
\end{array}
\eq
Thus $m_{12}$ is independent of $x$, and
\beq\label{t643}
\begin{array}{c}
 \displaystyle{
 F_{12}^0(x)=\ti F_{12}^0(x)
 =-\frac{\pi^2}{2\sin^2(\pi x)}\Big(
    \sigma_0\otimes\sigma_0+\cos(\pi x)
       \big(\sigma_1\otimes\sigma_1+\sigma_2\otimes\sigma_2\big)
       +\sigma_3\otimes\sigma_3
    \Big)\,.
 }
\end{array}
\eq
For $T=t_0\sigma_0+\sum_{a=1}^3t_a\sigma_a$, the corresponding
operators are
\beq\label{t644}
\begin{array}{c}
 \displaystyle{
 {\mathbb J}(T)=-\frac{\pi^2}{6}\,t_0\sigma_0
                -\frac{\pi^2}{2}\,t_3\sigma_3\,,
 }
 \\ \ \\
 \displaystyle{
 {\mathbb F}_x(T)=\ti{\mathbb F}_x(T)
 =-\frac{\pi^2}{\sin^2(\pi x)}\Big(
     t_0\sigma_0+\cos(\pi x)(t_1\sigma_1+t_2\sigma_2)+t_3\sigma_3
\Big)\,.
 }
\end{array}
\eq

\paragraph{Quadratic spin terms.}
The one-top contribution is
\beq\label{t645}
 \frac12\,\tr\Big({\mathbb J}(\mA^{ii})\mA^{ii}\Big)
 =-\frac{\pi^2}{6}\,(s_0^i)^2
  -\frac{\pi^2}{2}\,(s_3^i)^2\,.
\eq
Up to the central term, this is the Hamiltonian of an axially symmetric
Euler top. The self-interaction can be simplified using the rank-one
condition:
\beq\label{t646}
\begin{array}{c}
 \displaystyle{
 \frac12\,\tr\Big(\ti{\mathbb F}_{2q_i}(\mA^{ii})\mA^{ii}\Big)
 =-\frac{\pi^2}{\sin^2(2\pi q_i)}\Big(
     (s_0^i)^2+(s_3^i)^2
     +\cos(2\pi q_i)\big((s_1^i)^2+(s_2^i)^2\big)\Big)=
 }
 \\ \ \\
 \displaystyle{
 =-\frac{\pi^2}{2}\Big(
       \frac{(s_0^i)^2}{\sin^2(\pi q_i)}
       +\frac{(s_3^i)^2}{\cos^2(\pi q_i)}\Big)\,.
 }
\end{array}
\eq
For the interaction between distinct tops, we obtain
\beq\label{t647}
\begin{array}{c}
 \displaystyle{
 \frac12\,\tr\Big({\mathbb F}_x(\mA^{jj})\mA^{ii}\Big)
 =\frac12\,\tr\Big(\ti{\mathbb F}_x(\mA^{jj})\mA^{ii}\Big)=
 }
 \\ \ \\
 \displaystyle{
 =-\frac{\pi^2}{\sin^2(\pi x)}\Big(
       s_0^is_0^j
       +\cos(\pi x)(s_1^is_1^j+s_2^is_2^j)+s_3^is_3^j\Big)\,.
 }
\end{array}
\eq
Here $x=q_i-q_j$ or $x=q_i+q_j$.

\paragraph{Boundary terms.}
The expansion
$\ti K^{\hbar}(q)=\nu_0\sigma_0/\hbar+\ti K^{(0)}(q)+O(\hbar)$
of (\ref{t641}) gives
\beq\label{t648}
\begin{array}{c}
 \displaystyle{
 \ti K^{(0)}(q)
 =\pi\nu_0\cot(\pi q)\sigma_0+\pi\sin(\pi q)(\nu_1\sigma_1+\nu_2\sigma_2)+\pi\nu_3\tan(\pi q)\sigma_3
\,,
 }
 \\ \ \\
 \displaystyle{
 \ti Y^{(0)}(q)=\p_q\ti K^{(0)}(q)
 =\pi^2\Big(-\frac{\nu_0}{\sin^2(\pi q)}\,\sigma_0+\cos(\pi q)(\nu_1\sigma_1+\nu_2\sigma_2)
             +\frac{\nu_3}{\cos^2(\pi q)}\,\sigma_3
\Big)\,.
 }
\end{array}
\eq
To determine the scalar potential, factorize (\ref{t641}) as
\beq\label{t6481}
\begin{array}{c}
 \displaystyle{
 \ti K^{\hbar}(q)=\pi\sin\big(\pi(q+\hbar)\big)\Big(
     \frac{\nu_0\sigma_0}{\sin(\pi q)\sin(\pi\hbar)}+\nu_1\sigma_1+\nu_2\sigma_2
    +\frac{\nu_3\sigma_3}{\cos(\pi q)\cos(\pi\hbar)}
\Big)\,.
 }
\end{array}
\eq
The Pauli relations and
$\sin\big(\pi(q+\hbar)\big)\sin\big(\pi(q-\hbar)\big)
=\sin^2(\pi q)-\sin^2(\pi\hbar)$ then yield
\beq\label{t649}
 \ti K^{\hbar}(q)\ti K^{-\hbar}(q)
   =\Big(u_2(q|\nu)-w_2(\hbar|\nu)\Big)\sigma_0\,,
\eq
where we choose the common additive constant in $u_2$ and $w_2$ so that
\beq\label{t650}
 u_2(x|\nu)=w_2(x|\nu)
 =\pi^2\Big(\frac{\nu_0^2}{\sin^2(\pi x)}
           +\frac{\nu_3^2}{\cos^2(\pi x)}
           +(\nu_1^2+\nu_2^2)\sin^2(\pi x)\Big)\,.
\eq
Finally, the term linear in the spins is
\beq\label{t651}
\begin{array}{c}
 \displaystyle{
 \frac12\,\tr\Big(\mA^{ii}\ti Y^{(0)}(q_i)\Big)
 =\pi^2\Big(-\frac{\nu_0s_0^i}{\sin^2(\pi q_i)}+\cos(\pi q_i)(\nu_1s_1^i+\nu_2s_2^i)
            +\frac{\nu_3s_3^i}{\cos^2(\pi q_i)}\Big)\,.
 }
\end{array}
\eq

\paragraph{Explicit expression for the Hamiltonian.}
Substituting (\ref{t645})--(\ref{t651}) into (\ref{t606}) and combining
the scalar potential with the self-interaction and boundary terms gives
\beq\label{t652}
\begin{array}{c}
 \displaystyle{
 H^{\hbox{\tiny{XXZ tops}}}
 =\frac12\sum_{i=1}^np_i^2
  -\frac{\pi^2}{2}\sum_{i=1}^n\Big(
      \frac{(\nu_0+s_0^i)^2}{\sin^2(\pi q_i)}
     +\frac{(\nu_3-s_3^i)^2}{\cos^2(\pi q_i)}
     +(\nu_1^2+\nu_2^2)\sin^2(\pi q_i)\Big)+
 }
 \\ \ \\
 \displaystyle{
 +\pi^2\sum_{i=1}^n\cos(\pi q_i)
                    (\nu_1s_1^i+\nu_2s_2^i)
  -\frac{\pi^2}{2}\sum_{i=1}^n(s_3^i)^2-
 }
 \\ \ \\
 \displaystyle{
 -2\pi^2\sum\limits_{i>j}^n
 \frac{s_0^is_0^j+s_3^is_3^j
       +\cos\big(\pi(q_i-q_j)\big)(s_1^is_1^j+s_2^is_2^j)}
      {\sin^2\big(\pi(q_i-q_j)\big)}-
 }
 \\ \ \\
 \displaystyle{
 -2\pi^2\sum\limits_{i>j}^n
 \frac{s_0^is_0^j+s_3^is_3^j
       +\cos\big(\pi(q_i+q_j)\big)(s_1^is_1^j+s_2^is_2^j)}
      {\sin^2\big(\pi(q_i+q_j)\big)}+C_{\rm XXZ}\,.
 }
\end{array}
\eq
With the normalization (\ref{t650}), the additive term is exactly
\beq\label{t653}
 C_{\rm XXZ}=-\frac{\pi^2}{6}\sum_{i=1}^n(s_0^i)^2\,.
\eq
It is constant on the chosen orbits; on the reduction level
$s_0^i=g/2$ it equals $-n\pi^2g^2/24$.
The Hamiltonian (\ref{t652}) describes axially symmetric tops with
position-dependent XXZ pair interactions and a boundary field.

\section{Quantum interacting tops and long-range spin chains associated
with the \texorpdfstring{${\rm BC}_n$}{BCn} root system}\label{sec7}
\setcounter{equation}{0}

In this section, we quantize the Hamiltonians obtained above and discuss
their relation to long-range spin chains. We begin with the standard
construction for the spin Calogero--Moser model and then apply an analogous
quantization to the Hamiltonians of interacting tops.

\subsection{The spin Calogero--Moser model as an isotropic model of interacting tops}
%
%
%
\paragraph{Classical spin model revisited.}
Consider the spin Calogero--Moser model (\ref{t306})--(\ref{t316}).
As noted above, for rank$(S)=1$ it reduces to the spinless model
(\ref{t318})--(\ref{t319}). More generally, the unreduced spin phase space
$\mO^{\hbox{\tiny{$n$}}}$ can be parameterized using $T^*\mC^{nN}$ as follows:
  \beq\label{t701}
  \begin{array}{c}
  \displaystyle{
S_{ij}=\sum\limits_{a=1}^N\xi^i_a\eta^j_a\,,
  }
  \\ \ \\
  \displaystyle{
 \{\xi^i_a,\eta^j_b\}=\delta_{ab}\delta^{ij}\,,\quad
 i,j=1...n\,,\quad a,b=1...N\,.
 }
 \end{array}
 \eq
Here we assume rank$(S)\leq N$. This construction reproduces the Poisson
structure (\ref{t306}). We now rewrite $S_{ij}S_{ji}$ as
  \beq\label{t702}
  \begin{array}{c}
  \displaystyle{
 S_{ij}S_{ji}=\sum\limits_{a,b=1}^N\xi^i_a\eta^j_a\xi^j_b\eta^i_b=
 \sum\limits_{a,b=1}^N\mS^{ii}_{ab}\mS^{jj}_{ba}=\tr(\mS^{ii}\mS^{jj})=\tr_{12}\Big(\mS^{ii}_1\mS^{jj}_2P_{12}\Big)
 }
 \end{array}
 \eq
where the $n$ rank-one matrices $\mS^{ii}\in\Mat$ have entries
  \beq\label{t703}
  \begin{array}{c}
  \displaystyle{
 \mS^{ii}_{ab}=\xi^i_a\eta^i_b\,,\quad
 i=1...n\,,\quad a,b=1...N\,.
 }
 \end{array}
 \eq
Thus, the interaction potential in (\ref{t316}) takes the
form\footnote{In (\ref{t704}), the function $\wp$ in (\ref{t316}) is replaced
by $E_2$, defined in (\ref{a08}). On the constraint surface (\ref{t317}),
this amounts to shifting the Hamiltonian by a Casimir function.}
  \beq\label{t704}
  \begin{array}{c}
  \displaystyle{
-S_{ij}S_{ji}E_2(q_i-q_j)=-\tr(\mS^{ii}\mS^{jj})E_2(q_i-q_j)
=-\tr_{12}\Big(\mS^{ii}_1\mS^{jj}_2P_{12}\Big)E_2(q_i-q_j)\,.
 }
 \end{array}
 \eq
This term can also be viewed as the interaction term
$\tr_{12}(F_{12}^0(q_i-q_j)\mS^{ii}_1\mS^{jj}_2)
=\tr_{12}(\p r_{12}(q_i-q_j)\mS^{ii}_1\mS^{jj}_2)$
in the Hamiltonian of interacting tops (\ref{t375}), with
$r_{12}(x)=P_{12}E_1(x)$. This interpretation requires three qualifications.

First, $r_{12}(x)=P_{12}E_1(x)$ is not a classical $r$-matrix.
The matrix $R^z_{12}(x)=P_{12}\phi(z,x)$ satisfies the QYBE (\ref{t108})
and the AYBE (\ref{t109}) because the permutation operator satisfies the
QYBE and the Kronecker function $\phi$ satisfies the addition formula
(\ref{a12}). However, its expansion (\ref{a10}) near $z=0$,
$R^z_{12}(x)=P_{12}\phi(z,x)=P_{12}/z+P_{12}E_1(x)+O(z)$,
is not a quasiclassical expansion of the form (\ref{t110}), since its
leading coefficient is $P_{12}$ rather than $1_N\otimes1_N$.

Second, the spin Calogero--Moser model has no terms of the form
$\tr(\mS^{ii}{\mathbb J}(\mS^{ii}))$. Thus, ${\mathbb J}=0$
in this case\footnote{Equivalently, one may take
${\mathbb J}={\rm Id}:\ {\mathbb J}(\mS^{ii})=\mS^{ii}$, since
$\tr(\mS^{ii}{\mathbb J}(\mS^{ii}))=\tr(\mS^{ii}\mS^{ii})$
is a Casimir function.}, and there is no internal dynamics of
$\mS^{ii}$. This reflects the full isotropy of the model.

Third, the rational limit gives $P_{12}E_1(x)\to P_{12}/x$, which is
a classical $r$-matrix. In the model of interacting tops, this case
arises from Yang's rational $R$-matrix
$R^z_{12}(x)=1_{N^2}/z+P_{12}/x$; see \cite{GSZ}.
The rational spin Calogero--Moser model therefore admits two descriptions:
the rational version of (\ref{t306})--(\ref{t316}), with an $n\times n$
Lax representation, and the isotropic model of interacting tops
(\ref{t375}), with the $nN\times nN$ Lax pair
(\ref{t364})--(\ref{t366}) based on Yang's rational $R$-matrix.

\paragraph{Spin exchange operator.}
The classical Hamiltonian (\ref{t316}) can be quantized as follows.
The many-body variables are quantized in the standard coordinate representation:
  \beq\label{t705}
  \begin{array}{c}
  \displaystyle{
p_i\to {\hat p}_i=\hbar\p_{q_i}\,,\quad q_i\to {\hat q}_i=q_i:\qquad [{\hat p}_i,{\hat q}_j]=\hbar \delta_{ij}
 }
 \end{array}
 \eq
For the spin variables, we write the Hamiltonian in terms of $\mS^{ii}$
and replace the original phase space $\mO^{\hbox{\tiny{$n$}}}//\mD$ by
  \beq\label{t706}
  \begin{array}{c}
  \displaystyle{
\underbrace{ \mO^{\hbox{\tiny{min,$N$}}}\times ...\times \mO^{\hbox{\tiny{min,$N$}}} }_{\hbox{\tiny{$n$ times}}}
 }
 \end{array}
 \eq
as in (\ref{t380}). The Poisson structure then becomes
  \beq\label{t707}
  \begin{array}{c}
  \displaystyle{
\{\mS^{ii}_1,\mS_2^{jj}\}=\delta^{ij}[P_{12},\mS_{1}^{ii}]\,.
 }
 \end{array}
 \eq
Since the phase space has changed, this model is formally distinct from
the original spin Calogero--Moser model.

In the basis $T_a$ defined in (\ref{t231}), the quantization of
(\ref{t707}) is given by the following
prescription\footnote{In the standard basis, this prescription reads
$\mS^{ii}_{ab}\to\stackrel{i}{e}_{ba}$.}:
for each $a\in\mZ_N\times\mZ_N$,
  \beq\label{t708}
  \begin{array}{c}
  \displaystyle{
\mS^{ii}_a\to{\hat \mS}^{ii}_a=\frac{1}{N}\,\stackrel{i}{T}_{-a}=\frac{1}{N}\,
\underbrace{ 1_N\otimes ...\otimes 1_N\otimes T_{-a}\otimes 1_N ...\otimes 1_N}_{\hbox{\tiny{$T_{-a}$ on $i$-th place}}}\,.
 }
 \end{array}
 \eq
Thus, $\mS^{ii}_a$ is mapped to $(1/N)T_{-a}$ acting on the $i$th
tensor factor of the Hilbert space $(\mC^N)^{\otimes n}$.
This realizes the fundamental representation of each of the $n$ copies
of the Lie algebra ${\rm gl}_N$.

Then for $i\neq j$
  \beq\label{t709}
  \begin{array}{c}
  \displaystyle{
S_{ij}S_{ji}=\tr_{12}\Big(\mS^{ii}_1\mS^{jj}_2P_{12}\Big)\to \tr_{12}(P_{12}{\hat\mS}^{ii}_1{\hat\mS}^{jj}_2)
\stackrel{(\ref{t238})}{=}\frac{1}{N}
\sum\limits_{a\in\mZ_N\times \mZ_N}\tr\Big(T_a{\hat\mS}^{ii}\Big)\tr\Big(T_{-a}{\hat\mS}^{jj}\Big)
\stackrel{(\ref{t236})}{=}
 }
 \\ \ \\
  \displaystyle{
=N\sum\limits_{a\in\mZ_N\times \mZ_N}{\hat\mS}^{ii}_{-a}{\hat\mS}^{jj}_{a}\stackrel{(\ref{t708})}{=}
\frac{1}{N}\sum\limits_{a\in\mZ_N\times \mZ_N} \stackrel{i}{T}_{a} \stackrel{j}{T}_{-a}=P_{ij}\,.
 }
 \end{array}
 \eq
This gives the permutation operator, also known as the spin exchange
operator. This quantization already appears in the original papers \cite{GH}.

Also,
  \beq\label{t710}
  \begin{array}{c}
  \displaystyle{
S_{ii}=\tr(\mS^{ii})=N\mS^{ii}_{(0,0)}\to T_{(0,0)}=1_N\,.
 }
 \end{array}
 \eq
Similarly to (\ref{t709}), we obtain
  \beq\label{t711}
  \begin{array}{c}
  \displaystyle{
S_{ii}S_{ii}=\tr_{12}\Big(\mS^{ii}_1\mS^{ii}_2P_{12}\Big)\to \tr_{12}(P_{12}{\hat\mS}^{ii}_1{\hat\mS}^{ii}_2)
\stackrel{(\ref{t238})}{=}\frac{1}{N}
\sum\limits_{a\in\mZ_N\times \mZ_N}\tr\Big(T_a{\hat\mS}^{ii}\Big)\tr\Big(T_{-a}{\hat\mS}^{ii}\Big)
\stackrel{(\ref{t236})}{=}
 }
 \\ \ \\
  \displaystyle{
=N\sum\limits_{a\in\mZ_N\times \mZ_N}{\hat\mS}^{ii}_{-a}{\hat\mS}^{ii}_{a}\stackrel{(\ref{t708})}{=}
\frac{1}{N}\sum\limits_{a\in\mZ_N\times \mZ_N} \stackrel{i}{T}_{a} \stackrel{i}{T}_{-a}=P_{ii}=N1_N\,.
 }
 \end{array}
 \eq
Thus, $S_{ii}S_{ii}$ is quantized as $P_{ii}$. Under this prescription,
$({\hat S}_{ii})^2$ differs from $\widehat{S_{ii}S_{ii}}$.

\subsection{Quantum Hamiltonians of interacting tops}

\paragraph{Quantum spin Calogero--Inozemtsev model: isotropic potential.}

We apply the fundamental-representation prescription for the local spins
(\ref{t701})--(\ref{t709}) to the one-spin Hamiltonian (\ref{t533}).
Substitution into (\ref{t533}), including its Casimir term (\ref{t534}),
yields
\beq\label{t7234}
\begin{array}{c}
\displaystyle{
 \hat H^{\hbox{\tiny{iso}}}
 =\frac12\sum_{i=1}^n\hat p_i^2\,\mathbf 1
 -\sum_{1\leq i<j\leq n}P_{ij}
 \Big(\wp(q_i-q_j)+\wp(q_i+q_j)\Big)-
 }\\ \ \\
\displaystyle{
 -\frac N2\sum_{i=1}^n\wp(2q_i)\,\mathbf 1
 -\frac12\sum_{i=1}^n\sum_{a=0}^3
 \la_a(\la_a+1)\wp(q_i+\Om_a)\,\mathbf 1
 +\hat C_{\rm iso}\,,
 }
\end{array}
\eq
where
\beq\label{t7235}
\begin{array}{c}
\displaystyle{
 \hat C_{\rm iso}
 =\frac{c_\tau}{2}\Big(
 2nN\mathbf 1+4\sum_{1\leq i<j\leq n}P_{ij}
 +n\sum_{a=0}^3\la_a\,\mathbf 1\Big)=
n c_\tau\Big(N+\frac12\sum_{a=0}^3\la_a\Big)\mathbf 1
 +2c_\tau\sum_{1\leq i<j\leq n}P_{ij}\,.
 }
\end{array}
\eq
and $c_\tau=m(0)=\vth'''(0)/(3\vth'(0))$.
The sum of all transpositions is central in the permutation algebra.
Hence $\hat C_{\rm iso}$ is independent of the coordinates and commutes
with $\hat H^{\hbox{\tiny{iso}}}$.

The duplication identity (\ref{a24}) combines the one-body terms:
\beq\label{t7236}
\begin{array}{c}
\displaystyle{
 \frac N2\wp(2q)+\frac12\sum_{a=0}^3
 \la_a(\la_a+1)\wp(q+\Om_a)
 =\frac12\sum_{a=0}^3
 \Big(\big(\la_a+1/2\big)^2+\frac{N-1}{4}\Big)\wp(q+\Om_a)\,.
 }
\end{array}
\eq
Thus the Hamiltonian is
\beq\label{t7237}
\begin{array}{c}
\displaystyle{
 \hat H^{\hbox{\tiny{iso}}}
 =\frac12\sum_{i=1}^n\hat p_i^2\,\mathbf 1
 -\sum_{1\leq i<j\leq n}P_{ij}
 \Big(\wp(q_i-q_j)+\wp(q_i+q_j)\Big)-
 }\\ \ \\
\displaystyle{
 -\frac12\sum_{i=1}^n\sum_{a=0}^3
 \Big(\big(\la_a+1/2\big)^2+\frac{N-1}{4}\Big)
 \wp(q_i+\Om_a)\,\mathbf 1+\hat C_{\rm iso}\,.
 }
\end{array}
\eq
Equivalently, using $E_2(x)=\wp(x)-c_\tau$ from (\ref{a08}), the
transposition term in (\ref{t7235}) cancels the constant part of the
pair potential. The exact result can therefore be written with a
scalar additive constant:
\beq\label{t7238}
\begin{array}{c}
\displaystyle{
 \hat H^{\hbox{\tiny{iso}}}
 =\frac12\sum_{i=1}^n\hat p_i^2\,\mathbf 1
 -\sum_{1\leq i<j\leq n}P_{ij}
 \Big(E_2(q_i-q_j)+E_2(q_i+q_j)\Big)-
 }\\ \ \\
\displaystyle{
 -\frac12\sum_{i=1}^n\sum_{a=0}^3
 \Big(\big(\la_a+1/2\big)^2+\frac{N-1}{4}\Big)\wp(q_i+\Om_a)\,\mathbf 1
 +n c_\tau\Big(N+\frac12\sum_{a=0}^3\la_a\Big)\mathbf 1\,.
 }
\end{array}
\eq
Hamiltonians of this type were obtained and studied in \cite{Finkel}
using Dunkl operators.

\paragraph{Quantum interacting Euler tops: anisotropic potential.}

Consider the elliptic Hamiltonian (\ref{t623}) and apply the fundamental
representation prescription directly to this expression. The spin space
is $\mathcal H_{\rm spin}=(\mC^2)^{\otimes n}$. Write
$\mathbf 1=1_{\mathcal H_{\rm spin}}$ and set
\beq\label{t724}
\begin{array}{c}
 \displaystyle{
 \sigma_a^{(i)}=1_2^{\otimes(i-1)}\otimes\sigma_a
                        \otimes1_2^{\otimes(n-i)}\,,\qquad
 s_a^i\ \longmapsto\ \hat s_a^i=\frac12\,\sigma_a^{(i)}\,,
 \quad a=0,1,2,3\,,
 }
 \\ \ \\
 \displaystyle{
 \sigma_0=1_2\,,\qquad \hat s_0^i=\frac12\,\mathbf 1\,,\qquad
 \hat p_i=\hbar\p_{q_i}\,.
 }
\end{array}
\eq
The spin generators have the unit normalization of (\ref{t708}), while
the momenta are quantized as in (\ref{t705}). In particular,
$[\hat s_a^i,\hat s_b^j]
=\imath\delta^{ij}\sum_{c=1}^3\epsilon_{abc}\hat s_c^i$ for
$a,b=1,2,3$. The operators at different sites commute, and the required
products are
\beq\label{t725}
 (\hat s_a^i)^2=\frac14\,\mathbf 1\,,\qquad
 \hat s_a^i\hat s_a^j
   =\frac14\,\sigma_a^{(i)}\sigma_a^{(j)}\quad(i\neq j)\,,
 \qquad a=0,1,2,3\,.
\eq
The pure Euler contribution in (\ref{t623}) vanishes since
$\sum_{a=1}^3\wp(\Om_a)=0$. The self-interaction also becomes scalar
in spin space and admits a useful simplification. The degeneration of
(\ref{a21}) gives
\beq\label{t726}
 \sum_{a=1}^3\vf_a(2q)=2E_1(q)-E_1(2q)\,,\qquad
 \sum_{a=1}^3\vf'_a(2q)
   =E_2(2q)-E_2(q)=\wp(2q)-\wp(q)\,.
\eq
Here the prime denotes differentiation with respect to the argument of
$\vf_a$, as in (\ref{t612}). Hence, the quadratic terms involving $2q_i$ give
\[
 \frac14\Big(\sum_{a=1}^3\vf'_a(2q_i)-\wp(2q_i)\Big)\mathbf 1
 =-\frac14\,\wp(q_i)\mathbf 1\,.
\]
Substitution of (\ref{t724})--(\ref{t726}) into (\ref{t623}) yields
\beq\label{t727}
\begin{array}{c}
 \displaystyle{
 \hat H^{\hbox{\tiny{XYZ tops}}}
 =\frac12\sum_{i=1}^n\hat p_i^2\,\mathbf 1
  -\frac12\sum_{i=1}^n\sum_{k=0}^3
                  \rho_k^2\wp(q_i+\Om_k)\,\mathbf 1  -\frac{2\rho_0+1}{4}\sum_{i=1}^n\wp(q_i)\,\mathbf 1-
 }
 \\ \ \\
 \displaystyle{
  -\frac12\sum\limits_{i>j}^n
       \Big(\wp(q_i-q_j)+\wp(q_i+q_j)\Big)\mathbf 1
 +\frac12\sum_{i=1}^n\sum_{a=1}^3
       \rho_a\vf'_a(q_i-\Om_a)\sigma_{4-a}^{(i)}+
 }
 \\ \ \\
 \displaystyle{
 +\frac12\sum\limits_{i>j}^n\sum_{a=1}^3
       \Big(\vf'_a(q_i-q_j)+\vf'_a(q_i+q_j)\Big)
       \sigma_{4-a}^{(i)}\sigma_{4-a}^{(j)}
       +C_\tau^{\rm fund}\mathbf 1\,,
 }
\end{array}
\eq
where the couplings $\rho_a$ are given in (\ref{t618}), and
\beq\label{t728}
 C_\tau^{\rm fund}
 =\frac{n(n+\rho_0)}{2}\,c_\tau\,,\qquad
 c_\tau=\frac{\vth'''(0)}{3\vth'(0)}\,.
\eq
The terms linear in the Pauli matrices describe the boundary field;
the bilinear terms describe the pair interactions with XYZ anisotropy.
The classical rank-one relation is used in deriving (\ref{t623}), but
is not imposed again on operator products: in the representation
(\ref{t724}), $\sum_{a=1}^3(\hat s_a^i)^2=(3/4)\mathbf 1$, whereas
$(\hat s_0^i)^2=\mathbf 1/4$. In particular, applying the same substitution
to the first expression in (\ref{t615}), before the classical rank-one
simplification, adds $n c_\tau\mathbf 1/4$ to (\ref{t727}). This changes
only the additive constant.

\paragraph{Quantum interacting tops: XXZ case.}
We use the fundamental representation (\ref{t724}) for the general
trigonometric Hamiltonian (\ref{t606}). The parameters $\nu_a$ retain
the normalization of (\ref{t641}). For comparison, the squared
boundary factors in the simplified expression (\ref{t652}) become
\beq\label{t729}
\begin{array}{c}
 \displaystyle{
 (\nu_0\mathbf 1+\hat s_0^i)^2
   =\Big(\nu_0+\frac12\Big)^2\mathbf 1\,,\qquad
 (\nu_3\mathbf 1-\hat s_3^i)^2
   =\Big(\nu_3^2+\frac14\Big)\mathbf 1
      -\nu_3\sigma_3^{(i)}\,.
 }
\end{array}
\eq
Moreover, $C_{\rm XXZ}$ in (\ref{t653}) becomes
$-n\pi^2/24$, while the term
$-\frac{\pi^2}{2}\sum_i(\hat s_3^i)^2$ contributes
$-n\pi^2\mathbf 1/8$. However, direct substitution in (\ref{t652})
uses a classical rank-one identity that does not hold for the
operator products. Quantizing before that simplification adds
\[
 \frac{\pi^2}{8}\sum_i
 \left(\frac1{\cos^2(\pi q_i)}-\frac1{\sin^2(\pi q_i)}\right)\mathbf1.
\]
Including this correction, in agreement with the general prescription
(\ref{t744}), the complete quantum Hamiltonian is
\beq\label{t730}
\begin{array}{c}
 \displaystyle{
 \hat H^{\hbox{\tiny{XXZ tops}}}
 =\frac12\sum_{i=1}^n\hat p_i^2\,\mathbf 1
 -\frac{\pi^2}{2}\sum_{i=1}^n\Big(
      \frac{\nu_0^2+\nu_0+1/2}{\sin^2(\pi q_i)}
     +\frac{\nu_3^2}{\cos^2(\pi q_i)}
     +(\nu_1^2+\nu_2^2)\sin^2(\pi q_i)\Big)\mathbf 1+
 }
 \\ \ \\
 \displaystyle{
 +\frac{\pi^2}{2}\sum_{i=1}^n\Big(
      \frac{\nu_3}{\cos^2(\pi q_i)}\sigma_3^{(i)}
     +\cos(\pi q_i)
          \big(\nu_1\sigma_1^{(i)}+\nu_2\sigma_2^{(i)}\big)\Big)-
 }
 \\ \ \\
 \displaystyle{
 -\frac{\pi^2}{2}\sum\limits_{i>j}^n
 \frac{\mathbf 1+\sigma_3^{(i)}\sigma_3^{(j)}
       +\cos\big(\pi(q_i- q_j)\big)
        \big(\sigma_1^{(i)}\sigma_1^{(j)}
             +\sigma_2^{(i)}\sigma_2^{(j)}\big)}
      {\sin^2\big(\pi(q_i- q_j)\big)}
 -
 }
  \\ \ \\
 \displaystyle{
 -\frac{\pi^2}{2}\sum\limits_{i>j}^n
 \frac{\mathbf 1+\sigma_3^{(i)}\sigma_3^{(j)}
       +\cos\big(\pi(q_i+ q_j)\big)
        \big(\sigma_1^{(i)}\sigma_1^{(j)}
             +\sigma_2^{(i)}\sigma_2^{(j)}\big)}
      {\sin^2\big(\pi(q_i+ q_j)\big)}
 -\frac{n\pi^2}{6}\,\mathbf 1\,.
 }
\end{array}
\eq

\subsection{Relation to quantum long-range spin chains}

\paragraph{Quantization for arbitrary \texorpdfstring{$N$}{N} in the fundamental representation.}
Consider (\ref{t605}) on $\mathcal H_{\rm spin}=(\mC^N)^{\otimes n}$,
with $\mathbf 1=1_N^{\otimes n}$. We use the quantization formulae
(\ref{t707})--(\ref{t711}). The commutation relations in the standard
basis are
\beq\label{t740}
\begin{array}{c}
 \displaystyle{
 [\hat\mA^{ii}_{ab},\hat\mA^{jj}_{cd}]
 =\delta^{ij}\Big(\delta_{ad}\hat\mA^{ii}_{cb}
                 -\delta_{bc}\hat\mA^{ii}_{ad}\Big)\,.
 }
\end{array}
\eq
Substitution of (\ref{t707})-(\ref{t711}) into (\ref{t605}) yields the quantum Hamiltonian
\beq\label{t744}
\begin{array}{c}
 \displaystyle{
 \hat H_N^{\hbox{\tiny{tops}}}
 =\frac12\sum_{i=1}^n\hat p_i^2\,\mathbf 1
  -\frac12\sum_{i=1}^n u_N(q_i|\nu)\,\mathbf 1+
 }
 \\ \ \\
 \displaystyle{
 +\sum\limits_{i>j}^n
       \Big(F^0_{ij}(q_i-q_j)+\ti F^0_{ij}(q_i+q_j)\Big)
 +\frac12\sum_{i=1}^n\Big(
       m_{ii}(0)+\ti F^0_{ii}(2q_i)
       +\ti Y_i^{(0)}(q_i|\nu)\Big)\,.
 }
\end{array}
\eq

\paragraph{Elliptic Hamiltonian.}
For the Baxter--Belavin matrix (\ref{t239}), equation (\ref{t2431}) gives
\beq\label{t745}
  m_{ii} (0)= \frac{c_\tau}{N}\, P_{ii}
     -\frac{1}{N^2}\, \sum_{\substack{\al\neq(0,0)\\ \al\in\,\mZ_N\times\mZ_N}} \wp(\om_\al)
      \stackrel{i}{T}_\al \stackrel{i}{T}_{-\al}
 =c_\tau 1_N\,,
\eq
Here we have used $T_\al T_{-\al}=1_N$,
$E_2(x)=\wp(x)-c_\tau$ and $\sum_{\al\neq0}\wp(\om_\al)=0$.
The last identity follows by comparing the constant terms at $x=0$
in the multiplication identity
$\sum_\al\wp(x+\om_\al)=N^2\wp(Nx)$.

To evaluate $\ti F^0_{ii}(2q_i)$, consider the matrix acting in a single spin space,
\begin{equation}\label{t746}
 \tilde R^z_{11}(x):=\frac1N
 \sum_{\al\in\mZ_N\times\mZ_N}
       \vf_\al\Big(x,\frac zN+\om_\al\Big)T_{2\al}
 \in\Mat\,,
\end{equation}
where $T_{2\al}$ arises from the identity $T_\al^2=T_{2\al}$.
The matrix (\ref{t746}) is related to the $K$-matrix by
\begin{equation}\label{t747}
 \ti R^z_{11}(2x)=\frac12\,\ti K^z(x|\nu^*)\,,\qquad
 \nu^*=(1,1,(-1)^N,1)\,.
\end{equation}
To verify this identity, group the terms in (\ref{t252}) according
to the four half-period shifts. For $\nu=\nu^*$ the coefficient of
the summand indexed by $\be\in\mZ_{2N}^2$ is proportional to
$(1+(-1)^{\be_1})(1+(-1)^{\be_2})$.
Only the terms with $\be=2\al$ survive, yielding (\ref{t746}).
Differentiating with respect to $x$ and
taking the constant term at $z=0$ gives
\beq\label{t748}
\ti F^0_{ii}(2q_i)=\p_x\ti R^{(0)}_{ii}(x)\Big|_{x=2q_i}
 =\frac14\,\ti Y^{(0)}_i(q_i|\nu^*)\,.
\eq
Consequently, the elliptic version of (\ref{t744}) takes the following
form for every $N$:
\beq\label{t749}
\begin{array}{c}
 \displaystyle{
 \hat H^{\hbox{\tiny{tops}}}
 =\frac12\sum_{i=1}^n\hat p_i^2\,\mathbf 1
  -\frac12\sum_{i=1}^n\sum_{a=0}^3
             \bar\nu_a^2\wp(q_i+\Om_a)\,\mathbf 1+\frac{n c_\tau}{2}\,\mathbf 1+
 }
 \\ \ \\
 \displaystyle{
 +\sum\limits_{i>j}^n
       \Big(F^0_{ij}(q_i-q_j)+\ti F^0_{ij}(q_i+q_j)\Big)
 +\frac12\sum_{i=1}^n\ti Y_i^{(0)}(q_i|\nu)
  +\frac18\sum_{i=1}^n\ti Y_i^{(0)}(q_i|\nu^*)
\,.
 }
\end{array}
\eq
The last two terms combine as follows:
\beq\label{t7491}
\begin{array}{c}
 \displaystyle{
 \ti Y^{(0)}(q|\nu)+\frac14\ti Y^{(0)}(q|\nu^*)
   =\ti Y^{(0)}(q|\nu^{\rm eff})\,.
 }
 \\ \ \\
 \displaystyle{
 \nu^{\rm eff}=\nu+\frac14\nu^*
 =\Big(\nu_0+\frac14,\,\nu_1+\frac14,\,
        \nu_2+\frac{(-1)^N}{4},\,\nu_3+\frac14\Big)\,,
 }
\end{array}
\eq

\paragraph{Long-range spin chains.}
The prescription introduced in \cite{Poly}, known as the Polychronakos
freezing trick, produces long-range spin chains from quantum spin
Calogero--Moser systems; see also \cite{MZ}. It consists of removing
the differential part of the quantum Hamiltonian, arising from
${\hat p}_i$, and fixing the particle coordinates at an equilibrium
configuration of the corresponding classical spinless Hamiltonian.

As an example, consider the isotropic elliptic spin Calogero--Moser system:
  \beq\label{t750}
  \begin{array}{c}
    \displaystyle{
 \hat{H}^{\hbox{\tiny{spin}}}=\frac{1}{2}\sum\limits_{i=1}^n{\hat
 p}_i^2-\sum\limits_{i>j}^n {g}({g}+\hbar\, P_{ij})\wp(q_i-q_j)\,.
 }
 \end{array}
 \eq
The corresponding classical spinless model (\ref{t102}) has an equilibrium
configuration consisting of equally spaced points on the real-period
circle $\mathbb R/\mathbb Z$,
  \beq\label{t751}
  \begin{array}{c}
    \displaystyle{
z_j=\frac{j}{n}\,,\quad j=1,...,n\,,\quad z_n\equiv z_0\pmod{1}\,,
 }
 \end{array}
 \eq
which satisfy
  \beq\label{t752}
  \begin{array}{c}
    \displaystyle{
\sum\limits_{k:k\neq i}^n\wp'(z_i-z_k)=0\ \forall i=1,...,n\,.
 }
 \end{array}
 \eq
Applying this prescription to (\ref{t750}) gives, up to an additive
scalar constant and the overall factor $-g\hbar$, the Inozemtsev spin
chain \cite{Inoz90}:
  \beq\label{t753}
  \begin{array}{c}
    \displaystyle{
 \hat{H}^{\hbox{\tiny{Inoz}}}=\sum\limits_{i>j}^n P_{ij}\wp(z_i-z_j)\,.
 }
 \end{array}
 \eq
More precisely,
  \beq\label{t7531}
  \begin{array}{c}
    \displaystyle{
 \left.\hat H^{\hbox{\tiny{spin}}}\right|_{\hat p\mapsto0,\ q\mapsto z}
 =-g^2\sum_{i>j}^n\wp(z_i-z_j)\,\mathbf1
   -g\hbar\,\hat H^{\hbox{\tiny{Inoz}}}.
 }
 \end{array}
 \eq
 %
The freezing prescription does not by itself establish the integrability
of the resulting model. This requires a separate proof; see, for example,
\cite{MZ,Ch3,Lam,Lam2}.
Applying the same procedure to the trigonometric spin Calogero--Moser
model yields the Haldane--Shastry chain \cite{HS1}:
  \beq\label{t754}
  \begin{array}{c}
    \displaystyle{
 \hat{H}^{\hbox{\tiny{HS}}}=\sum\limits_{i>j}^n \frac{P_{ij}}{\sin^2(\pi(z_i-z_j))}
 }
 \end{array}
 \eq
with the same equilibrium positions (\ref{t751}), since (\ref{t752})
remains valid in the trigonometric limit:
  \beq\label{t7541}
  \begin{array}{c}
    \displaystyle{
 \hat H^{\hbox{\tiny{Inoz}}}\ \longrightarrow\
 \pi^2\hat H^{\hbox{\tiny{HS}}}
 -\frac{\pi^2}{3}\sum_{i>j}^n P_{ij},
 \qquad
 \wp(x)\ \longrightarrow\
 \frac{\pi^2}{\sin^2(\pi x)}-\frac{\pi^2}{3}.
 }
 \end{array}
 \eq
The sum $\sum_{i>j}P_{ij}$ is central in the permutation algebra,
but is not generally scalar on the full spin space. It can be absorbed
into an additive scalar constant on a fixed irreducible permutation
sector; the overall factor $\pi^2$ must also be taken into account.

Generalizations of the Haldane--Shastry model with ${\rm BC}_n$-type
dependence on $z_i$ were introduced in \cite{BPS} and then extended in \cite{EFGR} for isotropic trigonometric
and rational models using the Dunkl operator approach.

Anisotropic elliptic chains associated with the $A$-type models
(\ref{t117}) were introduced in \cite{SeZ18}:
  \beq\label{t755}
  \begin{array}{c}
    \displaystyle{
 \hat{H}^{\hbox{\tiny{long-range}}}=\sum\limits_{i>j}^n \p r_{ij}(z_i-z_j)\,,\qquad \p r_{ij}(x)=\p_x r_{ij}(x)=F^0_{ij}(x)\,.
 }
 \end{array}
 \eq
The construction uses half-quantum $R$-matrix-valued Lax pairs and
provides an alternative interpretation of the freezing trick.

Similarly, a ${\rm BC}_n$ version of the Hamiltonian (\ref{t755})
was proposed in \cite{MMZ,MMZ2} using the ${\rm BC}_n$ AYBEs and
$R$-matrix-valued Lax pairs. To match the quantum Hamiltonian
(\ref{t749}), its boundary $K$-matrix is evaluated at the effective
parameters $\nu^{\rm eff}$ defined in (\ref{t7491}):
  \beq\label{t756}
  \begin{array}{c}
    \displaystyle{
 \hat\mH^{\hbox{\tiny{long-range}}}(z|\nu^{\rm eff})
 =\sum_{i>j}^n\Big(\p r_{ij}(z_i-z_j)
                       +\p\ti r_{ij}(z_i+z_j)\Big)
   +\frac12\sum_{i=1}^n\ti Y^{(0)}_i(z_i|\nu^{\rm eff})\,,
 }
 \end{array}
 \eq
Here $\nu^{\rm eff}=\nu+\nu^*/4$, as in (\ref{t7491}).
Applying (\ref{a48}) to $\nu^{\rm eff}$ also gives
  \beq\label{t7561}
  \begin{array}{c}
    \displaystyle{
 \overline{\nu^{\rm eff}}=\bar\nu+\frac18
 \left(\begin{array}{c}
 3+(-1)^N\\
 1-(-1)^N\\
 (-1)^N-1\\
 1-(-1)^N
 \end{array}\right).
 }
 \end{array}
 \eq
Thus the sites are nonsingular stationary points of (\ref{t118})
with $\bar\nu_a$ replaced by $(\overline{\nu^{\rm eff}})_a$:
\begin{equation}
 \label{t757}
   \displaystyle{
\sum_{k:k\neq i}^{n} \Big(\wp'({z}_i-{z}_k) + \wp'({z}_i+{z}_k)\Big) +
\frac{1}{2}\sum\limits_{a=0}^3
\big((\overline{\nu^{\rm eff}})_a\big)^2
\wp'({z}_i+\Om_a)=0\,,\quad i=1,...,n\,.
}
 \end{equation}
The existence of real solutions to these equations was proved in
\cite{Tak} under suitable conditions on the coupling constants.
Explicit examples of the Hamiltonian (\ref{t756}) are given in \cite{MMZ2}.

Replacing $z_i$ by $q_i$ in (\ref{t756}) defines
$\hat\mH^{\hbox{\tiny{long-range}}}(q|\nu^{\rm eff})$.
Equations (\ref{t749})--(\ref{t7491}) give the exact operator identity
\begin{equation}
 \label{t7571}
   \displaystyle{
 \hat H^{\hbox{\tiny{tops}}}
 =\frac12\sum_i\hat p_i^2\,\mathbf1
 -\frac12\sum_i u_N(q_i|\nu)\,\mathbf1
 +\frac{nc_\tau}{2}\,\mathbf1
 +\hat\mH^{\hbox{\tiny{long-range}}}(q|\nu^{\rm eff}),
 \qquad \nu^{\rm eff}=\nu+\frac14\nu^*.
}
 \end{equation}
The scalar term $u_N(q_i|\nu)$ in this identity retains the original
$K$-matrix parameters.

Consider the XXZ example. Here $\nu_a$ denote the parameters of the trigonometric
$K$-matrix (\ref{t641}), in the normalization used in (\ref{t730}).
The single-site contributions are scalar:
\begin{equation}
 \label{t7572}
   \displaystyle{
 \frac12m_{ii}(0)=-\frac{\pi^2}{6}\,\mathbf1,\qquad
 \frac12\ti F^0_{ii}(2q_i)
 =-\frac{\pi^2}{4\sin^2(\pi q_i)}\,\mathbf1.
}
 \end{equation}
They therefore enter the scalar part that is subtracted when
extracting the chain. The resulting Hamiltonian of the chain is
\beq\label{t758}
\begin{array}{c}
 \displaystyle{
 \hat H^{\hbox{\tiny{XXZ long-range}}}
 =
\frac{\pi^{2}}{2}\sum_{i=1}^n\Big(
     \cos(\pi z_i)
          \big(\nu_1\sigma_1^{(i)}+\nu_2\sigma_2^{(i)}\big)+ \frac{\nu_3}{\cos^2(\pi z_i)}\sigma_3^{(i)}
     \Big)-
 }
 \\ \ \\
 \displaystyle{
 -\frac{\pi^{2}}{2}\sum\limits_{i>j}^n
 \frac{
       \cos\big(\pi(z_i- z_j)\big)
        \big(\sigma_1^{(i)}\sigma_1^{(j)}
             +\sigma_2^{(i)}\sigma_2^{(j)}\big)+\sigma_3^{(i)}\sigma_3^{(j)}}
      {\sin^2\big(\pi(z_i- z_j)\big)}
 -
 }
  \\ \ \\
 \displaystyle{
 -\frac{\pi^{2}}{2}\sum\limits_{i>j}^n
 \frac{\cos\big(\pi(z_i+ z_j)\big)
        \big(\sigma_1^{(i)}\sigma_1^{(j)}
             +\sigma_2^{(i)}\sigma_2^{(j)}\big)+\sigma_3^{(i)}\sigma_3^{(j)}}
      {\sin^2\big(\pi(z_i+ z_j)\big)}\,,
 }
\end{array}
\eq
For the XXZ chain, we choose the classical spinless background to be
the scalar part of (\ref{t730}), replacing the differential kinetic
term by its classical counterpart. Up to the additive constant
$-n\pi^2/6$, this Hamiltonian is
\beq\label{t7582}
 \begin{split}
 &\frac12\sum_i p_i^2-\frac{\pi^2}{2}\sum_{i>j}^n
 \left(\frac1{\sin^2(\pi(q_i-q_j))}
       +\frac1{\sin^2(\pi(q_i+q_j))}\right)\\
 &\qquad-\frac12\sum_i u_2(q_i|\nu)
 -\frac{\pi^2}{2}\Big(\nu_0+\frac12\Big)
       \sum_i\frac1{\sin^2(\pi q_i)},
 \end{split}
\eq
where $u_2$ is given by the unitarity of the trigonometric $K$-matrix
in (\ref{t649})--(\ref{t650}), with its original parameters $\nu_a$.
The scalar pair coefficient is $-\pi^2/2$, corresponding to
$g^2=1/2$ in the normalization of (\ref{t118}).
The sites $z_i$ are chosen to be nonsingular stationary points of
this scalar potential. Their equilibrium equations are
\beq\label{t759}
\begin{array}{c}
\displaystyle{
\sum_{j\neq i}^{n} \Big(\frac{\cos(\pi z_{ij})}{\sin^3(\pi z_{ij})} + \frac{\cos(\pi z^+_{ij})}{\sin^3(\pi z^+_{ij})}\Big) +
\Big(\nu_0^2+\nu_0+\frac12\Big)
\frac{\cos(\pi z_i)}{\sin^3(\pi z_i)}-
}
\\ \ \\
\displaystyle{
-\nu_3^2\frac{\sin(\pi z_i)}{\cos^3(\pi z_i)}-
\frac{1}{2}(\nu^2_1+\nu^2_2)\sin{(2\pi z_i)}=0\,,\qquad i=1,\ldots,n\,.
}
\end{array}
\eq
%

\section{Conclusion}\label{sec8}
\setcounter{equation}{0}

We have constructed ${\rm BC}_n$ generalizations of the model of
interacting tops using parametric associative Yang--Baxter identities
involving $R$- and $K$-matrices. The generalized model has two spin-orbit
components. We have derived its Hamiltonian structure and Lax equation.
Its reductions include one-orbit models, $N=1$ spin
Calogero--Inozemtsev systems and the spinless Calogero--Inozemtsev model.
We have also obtained explicit XYZ and XXZ interactions, quantized the
Hamiltonians in the fundamental representation, and described their
relation to long-range spin chains.

The main classical special cases are summarized in
Figure~\ref{fig1}. The two operations $N=1$ and
$\mA=\mB$ commute, as shown by the two paths to the one-spin model.
The rank-one branches include the moment constraints
$\tr(\mA^{ii})=g$ and the Cartan quotient described in
Sections~\ref{sec5} and \ref{sec6}.

\begin{figure}[!htbp]
\centering
\setlength{\unitlength}{1mm}
\setlength{\fboxsep}{2mm}
\setlength{\fboxrule}{0.4pt}
\newcommand{\BCnCase}[1]{\fbox{\parbox[c][14mm][c]{64mm}{\centering\small #1}}}
\begin{picture}(154,141)(-1,1)
\put(76,129){\makebox(0,0){\BCnCase{
  \textbf{Generalized model}\\[1pt]
  $\mA,\mB$, $N,n$\\[1pt]
  $H^{\hbox{\tiny{Gtops}}}$: \eqref{tt422}}}}

\put(34,101){\makebox(0,0){\BCnCase{
  \textbf{Two-spin Calogero--Inozemtsev}\\[1pt]
  $A,B\in{\rm Mat}_n$, $N=1$\\[1pt]
  $H^{\hbox{\tiny{2-spin}}}$: \eqref{t511}}}}
\put(118,101){\makebox(0,0){\BCnCase{
  \textbf{One-orbit model}\\[1pt]
  $\mA=\mB\in{\rm Mat}_{nN}$\\[1pt]
  $H^{\hbox{\tiny{Gtops}}}$: \eqref{tt455}}}}

\put(34,73){\makebox(0,0){\BCnCase{
  \textbf{One-spin Calogero--Inozemtsev}\\[1pt]
  $A=B\in{\rm Mat}_n$, $N=1$\\[1pt]
  $H^{\hbox{\tiny{1-spin}}}$: \eqref{t533}}}}
\put(118,73){\makebox(0,0){\BCnCase{
  \textbf{$n$ interacting ${\rm gl}_N$ tops}\\[1pt]
  $\mA^{ii}\in{\rm Mat}_N$, $\rank\mA^{ii}=1$\\[1pt]
  $H^{\hbox{\tiny{tops}}}$: \eqref{t605}, \eqref{t606}}}}

\put(34,46){\makebox(0,0){\BCnCase{
  \textbf{Spinless Calogero--Inozemtsev}\\[1pt]
  $N=1$, $A_{ij}=g$ after reduction\\[1pt]
  $H^{\hbox{\tiny{CI}}}$: \eqref{t578}}}}
\put(118,46){\makebox(0,0){\BCnCase{
  \textbf{Interacting Euler tops: XYZ}\\[1pt]
  $N=2$, elliptic $R$- and $K$-matrices\\[1pt]
  $H^{\hbox{\tiny{Euler tops}}}$: \eqref{t623}}}}

\put(118,19){\makebox(0,0){\BCnCase{
  \textbf{Interacting Euler tops: XXZ}\\[1pt]
  $N=2$, trigonometric $R$- and $K$-matrices\\[1pt]
  $H^{\hbox{\tiny{XXZ tops}}}$: \eqref{t652}}}}

\linethickness{0.5pt}
\put(60,120){\vector(-2,-1){20}}
\put(44,117){\makebox(0,0)[r]{\small $N=1$}}
\put(92,120){\vector(2,-1){20}}
\put(108,117){\makebox(0,0)[l]{\small $\mA=\mB$}}

\put(34,92){\vector(0,-1){10}}
\put(31,87){\makebox(0,0)[r]{\small $A=B$}}
\put(84,92){\vector(-2,-1){20}}
\put(75,90){\makebox(0,0)[b]{\small $N=1$}}
\put(118,92){\vector(0,-1){10}}
\put(121,87){\makebox(0,0)[l]{\shortstack[l]{\small $\rank\mA=1$\\\small Cartan reduction}}}

\put(34,64){\vector(0,-1){9}}
\put(31,59.5){\makebox(0,0)[r]{\shortstack[r]{\small $\rank A=1$\\\small Cartan reduction}}}
\put(84,64){\vector(-2,-1){18}}
\put(75,62){\makebox(0,0)[b]{\small $N=1$}}
\put(118,64){\vector(0,-1){9}}
\put(121,59.5){\makebox(0,0)[l]{\shortstack[l]{\small $N=2$\\\small elliptic case}}}

\put(118,37){\vector(0,-1){9}}
\put(115,32.5){\makebox(0,0)[r]{\shortstack[r]{\small trigonometric specialization\\
    \small \eqref{t640}, \eqref{t641}}}}
\end{picture}
\caption{\it Classical special cases of the generalized ${\rm BC}_n$
model. Each box gives the corresponding Hamiltonian and its equation
reference. Arrows indicate restrictions, rank-one reductions or
specializations of the $R$- and $K$-matrices.}
\label{fig1}
\end{figure}
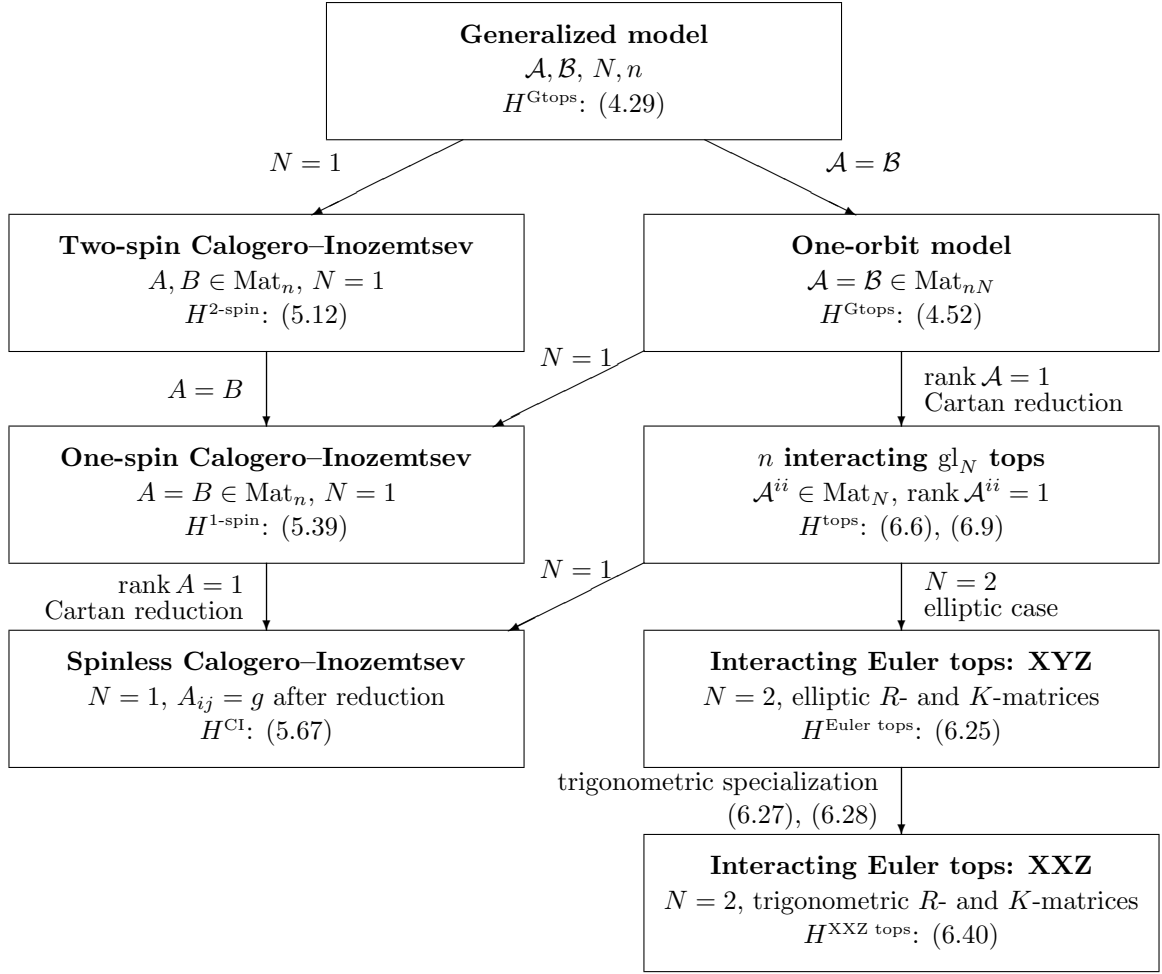

These results apply to the examples considered in \cite{MMZ2}, including
trigonometric and rational $R$-matrices satisfying the AYBE.
Such $R$-matrices were constructed in \cite{Chered} for ${\rm GL}_2$
and in \cite{Pol3,KrZ,AASZ} for ${\rm GL}_N$.


\section{Appendix A: elliptic functions}\label{secA}
\def\theequation{A.\arabic{equation}}
\setcounter{equation}{0}

\paragraph{The Kronecker function \texorpdfstring{$\phi$}{phi}.}
Most calculations use the elliptic Kronecker function \cite{Weil};
see also the appendices of \cite{MMZ,MMZ2,ZZ,TrZ}:
 \beq\label{a01}
\begin{array}{l}
  \displaystyle{
 \phi(z,u)=\frac{\vth'(0)\vth(z+u)}{\vth(z)\vth(u)}=\phi(u,z)\,,\quad\res\limits_{z=0}\phi(z,u)=1\,,
 \quad \phi(-z, -u) = -\phi(z, u)\,,
 }
 \end{array}
 \eq
where $\vth(z)$ is the first Jacobi theta function:
 \beq\label{a02}
 %
 \eq
Here $\tau$ is the modulus of the elliptic curve
$\mC/(\mZ+\tau\mZ)$, with ${\rm Im}(\tau)>0$.
In Jacobi's notation, the theta functions are
\beq\label{a04}
\begin{array}{c}
\displaystyle{
\theta_1(u|\tau )=\vth(u,\tau)=-i\sum_{k\in  \mZ}(-1)^k q^{(k+\frac{1}{2})^2}e^{\pi i (2k+1)u},
\qquad\theta_2(u|\tau )=\sum_{k\in  \mZ}q^{(k+\frac{1}{2})^2}e^{\pi i (2k+1)u}\,,}
\\ \ \\
\displaystyle{
\theta_3(u|\tau )=\sum_{k\in  \mZ}q^{k^2}e^{2\pi i ku},\qquad\theta_4(u|\tau )
=\sum_{k\in  \mZ}(-1)^kq^{k^2}e^{2\pi i ku}\,,}
\end{array}
\eq
where
$q=e^{\pi i \tau}$ and
\beq\label{a05}
\begin{array}{c}
\theta_2(u|\tau)=\theta_1(u+\frac{1}{2}|\tau), \! \quad \! \!\theta_3(u|\tau)
=q^{\frac{1}{4}}e^{\pi iu}\theta_1(u+\frac{\tau+1}{2}|\tau), \! \quad \!\!\theta_4(u|\tau)
=-iq^{\frac{1}{4}}e^{\pi iu}\theta_1(u+\frac{\tau}{2}|\tau)\,.
\end{array}
\eq
The partial derivative
\beq\label{a06}
\begin{array}{l}
\displaystyle{
f(z,u) = \partial_u \phi(z,u)
}
\end{array}
\eq
equals
\beq\label{a07}
\begin{array}{c} \displaystyle{
    f(z, u) = \phi(z, u)(E_1(z + u) - E_1(u)), \qquad f(-z, -u) = f(z, u)\,,
}\end{array}\eq
where $E_1(z)$ is the first Eisenstein function. The first and second
Eisenstein functions are defined by
\beq\label{a08}
\begin{array}{c} \displaystyle{
    E_1(z)=\frac{\vth'(z)}{\vth(z)}=\zeta(z)+\frac{z}{3}\,\frac{\vth'''(0)}{\vth'(0)}\,,
    \quad E_2(z) = - \partial_z E_1(z) = \wp(z) - \frac{\vartheta'''(0) }{3\vartheta'(0)}\,,
}\end{array}\eq
\beq\label{a09}
\begin{array}{c}
 \displaystyle{
    E_1(- z) = -E_1(z)\,, \quad E_2(-z) = E_2(z)\,,
}\end{array}
\eq
where $\wp(z)$ and $\zeta(z)$ are the Weierstrass functions.
The local expansion of $\phi(z,u)$ near $z=0$ is
  \beq\label{a10}
  \begin{array}{c}
  \displaystyle{
\phi(z,u)=\frac{1}{z}+E_1(u)+z\,\frac{E_1^2(u)-\wp(u)}2+O(z^2)\,,
 }
 \\[4mm]
  \displaystyle{
E_1(u)=\frac{1}{u}+\frac{u}{3}\,\frac{\vth'''(0)}{\vth'(0)}+O(u^3)\,,
 }
 \end{array}
 \eq
Therefore,
  \beq\label{a11}
  \begin{array}{c}
  \displaystyle{
 f(0,q)=-E_2(q)\,.
 }
 \end{array}
 \eq
The calculations rely on the addition formula
\beq\label{a12}
  \begin{array}{c}
  \displaystyle{
  \phi(z_1, u_1) \phi(z_2, u_2) = \phi(z_1, u_1 + u_2) \phi(z_2 - z_1, u_2)
  + \phi(z_2,u_1 + u_2) \phi(z_1 - z_2, u_1)\,.
 }
 \end{array}
 \eq
This is the scalar ($N=1$) version of the AYBE (\ref{t208}).
Applying $\p_{u_2}-\p_{u_1}$ to (\ref{a12}) gives
\beq\label{a13}
  \begin{array}{c}
  \displaystyle{
  \phi(z_1, u_1) f(z_2, u_2)-f(z_1,u_1)\phi(z_2,u_2) =
  }
  \\ \ \\
  \displaystyle{
  =\phi(z_1, u_1 + u_2) f(z_2 - z_1, u_2) - \phi(z_2, u_1 + u_2) f(z_1 - z_2, u_1)\,.
 }
 \end{array}
 \eq
We also use the following identities:
\beq\label{a14}
  \begin{array}{c}
  \displaystyle{
 \phi(z,u_1)\phi(z,u_2)=\phi(z,u_1+u_2)\Big(E_1(z)+E_1(u_1)+E_1(u_2)-E_1(z+u_1+u_2)\Big)\,,
 }
 \end{array}
 \eq
\beq\label{a15}
  \begin{array}{c}
  \displaystyle{
  \phi(z, u_1) f(z,u_2)-\phi(z, u_2) f(z,u_1)=\phi(z,u_1+u_2)\Big(\wp(u_1)-\wp(u_2)\Big)=
 }
 \\ \ \\
   \displaystyle{
 =\phi(z,u_1+u_2)\Big(E_2(u_1)-E_2(u_2)\Big)=\phi(z,u_1+u_2)\Big(f(0,u_2)-f(0,u_1)\Big)\,.
 }
 \end{array}
 \eq
\beq\label{a16}
  \begin{array}{c}
  \displaystyle{
  \phi(z, u) \phi(z, -u) = \wp(z)-\wp(u)=E_2(z)-E_2(u)\,,
 }
 \end{array}
 \eq
\beq\label{a17}
  \begin{array}{c}
  \displaystyle{
  \phi(z, u) f(z, -u)-\phi(z, -u) f(z, u)=\wp'(u)\,,
 }
 \end{array}
 \eq
and
\beq\label{a171}
\begin{array}{c}
\displaystyle{
    \Big(E_1(u+v) - E_1(u) - E_1(v)\Big)^2 = \wp(u+v) +  \wp(u) + \wp(v)\,.
}
\end{array}
\eq
Define the function
\beq\label{a172}
\begin{array}{c}
\displaystyle{
    m(x)=\frac{E_1^2(x)-\wp(x)}{2}\,.
}
\end{array}
\eq
This function appears in the expansion (\ref{a10}) and is the scalar
analogue of $m_{12}(z)$ in the quasiclassical expansion (\ref{t202});
see also (\ref{t243}). It satisfies
\beq\label{a173}
\begin{array}{c}
 \displaystyle{
    \phi (z_1, u) m(z_2) - E_1 (z_2) f(z_1, u) +
    \phi(z_2, u) f(z_{12}, u) - \phi(z_1, u) m(z_{21}) = -\frac{1}{2}\, \partial^2_u \phi(z_1, u)\,,
}
\end{array}
\eq
\beq\label{a174}
\begin{array}{c}
\displaystyle{
    \phi (z, u) m(z) - E_1 (z) f(z, u) -
     \phi(z, u) \wp(u) =- \frac{1}{2}\, \partial^2_u \phi(z, u)\,,
}\end{array}\eq

Define the half-periods
 \beq\label{a18}
 \begin{array}{c}
  \displaystyle{
 \Om_0=0\,,\quad\Om_1=\frac{1}{2}\,,\quad\Om_2=\frac{1+\tau}{2}\,,\quad\Om_3=\frac{\tau}{2}\,,
  }
 \end{array}
\eq
 and
 \beq\label{a19}
 \begin{array}{c}
  \displaystyle{
 \vf_{0}(z,\hbar)=\phi(z,\hbar)\,,\qquad\vf_{1}(z,\hbar+\Om_1)=\phi(z,\Om_1+\hbar)\,,
 }
 \\ \ \\
  \displaystyle{
 \vf_{2}(z,\hbar+\Om_2)=e^{\pi\imath z}\phi(z,\Om_2+\hbar)\,,\qquad\vf_{3}(z,\hbar+\Om_3)
 =e^{\pi\imath z}\phi(z,\Om_3+\hbar)\,.
 }
  \end{array}
 \eq
Then
 \beq\label{a20}
 \begin{array}{c}
  \displaystyle{
 \vf_{k}(z,\hbar+\Om_k)=\frac{\theta_1'(0)\theta_{k+1}(z+\hbar)}{\theta_{1}(z)
 \theta_{k+1}(\hbar)}\,,\quad k=0,1,2,3\,.
  }
 \end{array}
\eq
The functions introduced above are related by the linear transformation
  \beq\label{a21}
 %
 \eq
The matrix of this linear transformation can be written as
  \beq\label{a22}
  \begin{array}{c}
  \displaystyle{
I_{km}=\frac12\,\exp\Big(4\pi\imath\Big(\Om_{m-1}\p_\tau\Om_{k-1}-\Om_{k-1}\p_\tau
\Om_{m-1}\Big)\Big)\,,\quad k,m=1,...,4\,.
 }
 \end{array}
 \eq
It satisfies
  \beq\label{a23}
  \begin{array}{c}
  \displaystyle{
I^{-1}=I\,.
 }
 \end{array}
 \eq
We also use the standard identity
  \beq\label{a24}
 \begin{array}{c}
  \displaystyle{
\sum\limits_{a=0}^3\wp(z+\Om_a)=4\wp(2z)\,.
  }
 \end{array}
 \eq
The quasiperiodicity relations are
  \beq\label{a25}
  \begin{array}{l}
  \displaystyle{
 \phi(z+1,u)=\phi(z,u)\,,\qquad  \phi(z+\tau,u)=e^{-2\pi\imath u}\phi(z,u)\,,
 }
 \end{array}
 \eq
or
  \beq\label{a26}
  \begin{array}{l}
  \displaystyle{
 \phi(z\pm 2\Om_a,u)=e^{\mp4\pi\imath\p_\tau\Om_a u}\phi(z,u)\,.
 }
 \end{array}
 \eq
For the theta function (\ref{a02}) we have
  \beq\label{a27}
  \begin{array}{l}
  \displaystyle{
 \vth(z+2\Om_a)=-e^{-4\pi\imath(z+\Om_a)\p_\tau\Om_a}\vth(z)\,,\quad a\neq 0\,,
 }
 \end{array}
 \eq
or
  \beq\label{a28}
  \begin{array}{l}
  \displaystyle{
 \vth(z+\Om_a)=-e^{-4\pi\imath z\p_\tau\Om_a}\vth(z-\Om_a)\,,\quad a\neq 0\,.
 }
 \end{array}
 \eq
This gives
  \beq\label{a29}
  \begin{array}{l}
  \displaystyle{
 E_2(z+2\Om_a)=E_2(z)\,,\qquad E_1(z+2\Om_a)=E_1(z)-4\pi\imath\p_\tau\Om_a\,,
 }
 \end{array}
 \eq
and
  \beq\label{a30}
  \begin{array}{l}
  \displaystyle{
E_1(\Om_a)=-2\pi\imath\p_\tau\Om_a\,,\quad a\neq 0\,.
 }
 \end{array}
 \eq
Equations (\ref{a30}), (\ref{a07}) and (\ref{a12}) imply that, for $a=1,2,3$,
  \beq\label{a31}
  \begin{array}{l}
  \displaystyle{
\p_z\vf_a(z,\Om_a)=\vf_a(z,\Om_a)(E_1(z+\Om_a)-E_1(z)-E_1(\Om_a))=-\frac{\vf_1(z,\Om_1)
\vf_2(z,\Om_2)\vf_3(z,\Om_3)}{\vf_a(z,\Om_a)}\,.
 }
 \end{array}
 \eq

\paragraph{The functions \texorpdfstring{$\vf_\al$}{varphi(alpha)}.}
For $\al=(\al_1,\al_2)\in\mZ_N\times\mZ_N$, define the following
functions\footnote{The functions (\ref{a32}) form a basis of sections of
${\rm End}(V)$, where $V$ is a holomorphic vector bundle of degree $1$
over the elliptic curve.}:
\beq\begin{array}{c} \displaystyle{
\label{a32}
 \displaystyle{
    \vf_\al (z, \om_\al + u) = \exp (2 \pi i \frac{\al_2}{N}\, z) \phi (z, \om_\al + u),
    \quad \om_\al = \frac{\al_1 + \al_2 \tau}{N}\,,
    }
}\end{array}\eq
\beq\begin{array}{c} \displaystyle{
\label{a33}
 \displaystyle{
    f_\al (z, \om_\al + u) = \exp (2 \pi i \frac{\al_2}{N}\, z) f (z, \om_\al + u)\,,
    }
}\end{array}\eq
\beq\begin{array}{c} \displaystyle{
    \label{a34}
     \displaystyle{
    f_\al (z, \om_\al + u) = \partial_u \vf_\al (z, \om_\al + u)
    = \vf_\al (z, \om_\al +u) (E_1(z + \om_\al + u) - E_1(\om_\al + u))\,.}
}\end{array}\eq
The addition formulae take the form
\beq\begin{array}{c}
\label{a35}
 \displaystyle{
        \vf_\al (z_1, \om_\al + u_1) \vf_\be (z_2, \om_\be + u_2)
        = \vf_\al (z_1 - z_2,\om_\al + u_1) \vf_{\al + \be} (z_2, \om_{\al+\be} + u_1 + u_2) +}
         \\ \ \\
          \displaystyle{
        + \vf_\be (z_2 - z_1, \om_\be + u_2) \vf_{\al + \be} (z_1, \om_{\al + \be} + u_1 +u_2)\,.
        }
\end{array}\eq
In particular, for $\al,\be,\al+\be\neq0$ modulo $\mZ_N^2$,
\beq\begin{array}{c}
\label{a36}
 \displaystyle{
    \vf_\al(z-z_a, \om_\al) \vf_\be(z-z_b, \om_\be) =
  }
  \\ \ \\
   \displaystyle{
    =\vf_\al(z_{b}-z_a, \om_\al) \vf_{\al + \be}(z - z_b, \om_{\al + \be})
    +\vf_\be(z_{a}-z_b, \om_\be) \vf_{\al + \be}(z-z_a, \om_{\al + \be})
    }\,,
\end{array}\eq
and
\beq\begin{array}{c}
\label{a37}
\displaystyle{
        \vf_\al (z, \om_\al + u_1) \vf_\be (z, \om_\be + u_2)
        = \vf_{\al + \be}(z,\om_{\al + \be} + u_1 + u_2)  \times
         }
         \\ \ \\
         \displaystyle{
         \times\Big( E_1(z)+ E_1(\om_\al + u_1) + E_1(\om_\be + u_2)
         - E_1(z + \om_{\al +\be} + u_1 + u_2)\Big)\,,
         }
\end{array}\eq
and
\beq\begin{array}{c}
\label{a38}
\displaystyle{
        \vf_\al (z_1, \om_\al + u) \vf_\al (z_2, \om_\al + u) =
        }
        \\ \ \\
         \displaystyle{
         = \vf_{\al}(z_1 + z_2, \om_\al + u)  ( E_1(z_1) + E_1(z_2))
         - f_\al (z_1 + z_2,\om_\al + u)\,.
         }
\end{array}\eq
In these identities, we assume $\om_{\al+\be}=\om_\al+\om_\be$.

\paragraph{Discrete Fourier transform.}
Writing $\phi(z,u;\tau)$ for the Kronecker function (\ref{a01}), one has,
for every $m\in\mZ_N$,
\beq\label{a39}
\begin{array}{c}
\displaystyle{
\frac1N\,\sum_{k\in\mZ_N}\exp\!\Big(\frac{2\pi\imath\,mk}{N}\Big)
\phi\!\Big(z+\frac{k}{N},\,u;\,\tau\Big)=\exp\!\big(-2\pi\imath\,mz\big)\,
\phi(Nz,\,u-m\tau;\,N\tau)\,.
}
\end{array}
\eq
The companion sum along the $\tau$-direction is, for every $m\in\mZ_N$,
\beq\label{a40}
\begin{array}{c}
\displaystyle{
\sum_{k\in\mZ_N}\exp\!\Big(\frac{2\pi\imath\,k(u+m)}{N}\Big)
\phi\!\Big(z+\frac{k\tau}{N},\,u;\,\tau\Big)=\phi\!\Big(z,\,\frac{u+m}{N};\,\frac{\tau}{N}\Big)\,.
}
\end{array}
\eq
The Fourier symmetry (\ref{t221}) is based on the identities
\beq\label{a41}
 \begin{array}{c}
  \displaystyle{
\frac{1}{{ N}}\,\sum\limits_{\al\in\,\mZ_N\times\mZ_N} \kappa_{\al,\ga}^2\,
\vf_\al({N}\hbar,\om_\al+\frac{z}{{ N}})=\vf_\ga(z,\om_\ga+\hbar)\,,
\quad\forall\ga\in\mZ_N\times\mZ_N\,.
  }
 \end{array}
 \eq
Degenerations of these identities yield
\begin{equation}
\label{a42}
\begin{array}{c}
    \displaystyle{
    \frac{1}{N}\, E_1 (z) + \frac{1}{N}\, \sum_{\substack{\al\neq(0,0)\\\al\in\,\mZ_N\times\mZ_N}}
    \kappa^2_{\al, \be} \varphi_\al (z, \om_\al) = E_1 (\om_\be +\frac{z}{N})
    + 2 \pi i \partial_\tau \om_\be
    }\,.
\end{array}
\end{equation}
\begin{equation}\label{a43}
\begin{array}{c}
    \displaystyle{
    \frac{1}{N}\, \sum_{\substack{\al\neq(0,0)\\ \al\in\,\mZ_N\times\mZ_N}}
    \kappa^2_{\al,\be} \Big(E_1 (\om_\al) + 2 \pi i \partial_\tau \om_\al \Big) = E_1(\om_\be)
    + 2 \pi i\partial_\tau \om_\be, \quad \be \neq 0
    }\,.
\end{array}
\end{equation}
\begin{equation}
\begin{array}{c} \label{a44}
    \displaystyle{
    \frac{1}{N}\, \sum_{\substack{\al\neq(0,0)\\ \al\in\,\mZ_N\times\mZ_N}} \Big(E_1(\om_\al)
    + 2 \pi i \partial_\tau \om_\al \Big)
    = 0
    }\,.
\end{array}
\end{equation}
See \cite{Z2} for details.

\paragraph{The function \texorpdfstring{$v(z,u)$}{v(z,u)}.}
Using the functions (\ref{a18})--(\ref{a19}), define
  \beq\label{a45}
  \begin{array}{c}
  \displaystyle{
v(z,u)\equiv v(z,u|\nu)=\sum\limits_{a=0}^3\nu_a\vf_a(2z,u+\Om_a)=\sum\limits_{a=0}^3
\nu_a \exp(4\pi \imath z\p_\tau\Om_a)\phi(2z,u+\Om_a)\,,
 }
 \end{array}
 \eq
and its partial derivative
  \beq\label{a46}
  \begin{array}{c}
  \displaystyle{
v'(z,u)=\p_uv(z,u)=\sum\limits_{a=0}^3\nu_a \exp(4\pi \imath z\p_\tau\Om_a)f(2z,u+\Om_a)\,.
 }
 \end{array}
 \eq
Identities for $v(z,u)$ and $v'(z,u)$ are given in \cite{KH2,Z04}.
In particular,
\beq\label{a47}
v(z,u)v(-z,u)=\sum\limits_{a=0}^3\Big(\nu_a^2\wp(u+\Om_a)-\bar\nu_a^2\wp(z+\Om_a)\Big)
=-v(z,u)v(z,-u)\,.
\eq
where, as in (\ref{a21}),
  \beq\label{a48}
 %
 \eq
We also have
\beq\label{a50}
v(-z,-u)=-v(z,u)\,.
\eq
Differentiating (\ref{a47}) with respect to $u$ gives
  \beq\label{a51}
  \begin{array}{c}
  \displaystyle{
v(z,u)v'(z,-u)-v(z,-u)v'(z,u)=\sum\limits_{a=0}^3\nu_a^2\wp'(u+\Om_a)\,.
 }
 \end{array}
 \eq
In \cite{MMZ}, the following identities were also used:
\beq\label{a52}
  \begin{array}{c}
  \displaystyle{
v(x,u)\phi(x+y,w-u)+v(x,w)\phi(x-y,u-w)+v(y,-u)\phi(x+y,u+w)=
 }
 \\ \ \\
   \displaystyle{
=v(y,w)\phi(x-y,u+w)\,,
}
\end{array}
\eq
and
\beq\label{a53}
\begin{array}{c}
\phi(z,u-w)(v'(0,w)-v'(0,u))=
\\ \ \\
= 2v(-z,w)f(z,u+w)+2v(z,u)f(-z,u+w)+
\\ \ \\
+v'(-z,w)\phi(z,u+w)+v'(z,u)\phi(-z,u+w)\,,
\end{array}
\eq
where
  \beq\label{a54}
  \begin{array}{c}
  \displaystyle{
v'(0,u)=\sum\limits_{a=0}^3\nu_a f(0,u+\Om_a)=-\sum\limits_{a=0}^3\nu_a E_2(u+\Om_a)=
-\sum\limits_{a=0}^3\nu_a\wp(u+\Om_a)+\frac{\vth'''(0)}{3\vth'(0)}\,\sum\limits_{a=0}^3\nu_a\,.
 }
 \end{array}
 \eq
%

\section{Appendix B: proof of Theorem \ref{th1}}\label{secB}
\def\theequation{B.\arabic{equation}}
\setcounter{equation}{0}

We first write the Lax pair (\ref{tt402})--(\ref{tt410}) and the
equations of motion (\ref{tt417})--(\ref{tt421}) before simplifying
them using the constraints (\ref{t416}), the variables
${\mathcal A}$ and ${\mathcal B}$, and the symmetry properties
(\ref{t222}).

\underline{Lax pair.}
\beq\label{t402}
\begin{array}{c}
  \displaystyle{
    \mL(z) =\sum^{n}_{i,j=1}E_{ij}\otimes \mL^{ij}(z)+\sum^{n}_{i,j=1}E_{\bar{i}j}\otimes
    \mL^{\bar{i}j}(z)+\sum^{n}_{i,j=1}E_{i\bar{j}}\otimes\mL^{i\bar{j}}(z)
    +\sum^{n}_{i,j=1}E_{\bar{i}\bar{j}}\otimes \mL^{\bar{i}\bar{j}}(z)\,,
  }
\end{array}
\eq
\beq\label{t403}
\begin{array}{c}
  \displaystyle{
    \bar{i}=i+n,~ \qquad\bar{\bar{i}}=i, \qquad \qquad \mL^{ij}(z)\in \text{Mat}_N
    \qquad\mL(z) \in \text{Mat}_{2Nn}\,,
  }
\end{array}
\eq
\beq\label{t404}
\begin{array}{c}
  \displaystyle{
    \mL^{ij}(z)=\delta_{ij}\left(p_{i}1_N+\tr_{2}\Big(\mS^{ii}_{2}R^{z,(0)}_{12}P_{12}\Big)\right)
    +(1-\de_{ij})\tr_{2}\Big(\mS^{ij}_{2}R^{z}_{12}(q_{ij})P_{12}\Big)\,,
  }
\end{array}
\eq
\beq\label{t405}
\begin{array}{c}
  \displaystyle{
    \mL^{i\bar{j}}(z)=\de_{ij}\left( \ti{K}^{z}_{1}(q_{i})+\tr_{2}\Big(\mS^{i\bar{i}}_{2}
    \ti{R}^{z}_{12}(2q_{i})P_{12}\Big)\right) +(1-\de_{ij})\tr_{2}\Big(\mS^{i\bar{j}}_2
    \ti{R}^{z}_{12}(q_{ij}^{+})P_{12}\Big)\,,
  }
\end{array}
\eq
\beq\label{t406}
\begin{array}{c}
  \displaystyle{
    \mL^{\bar{i}\bar{j}}(z)=-\de_{ij}\left(p_{i}1_{N}
    +\tr_{2}\Big(\mS^{\bar{i}\bar{i}}_{2}R^{-z,(0)}_{12}P_{12}\Big)\right)-(1-\de_{ij})
    \tr_{2}\Big(\mS^{\bar{i}\bar{j}}_{2}R^{-z}_{12}(q_{ij})P_{12}\Big)\,,
  }
\end{array}
\eq
\beq\label{t407}
\begin{array}{c}
  \displaystyle{
    \mL^{\bar{i}j}(z)=-\de_{ij}\left( \ti{K}^{-z}_{1}(q_{i})+\tr_{2}\Big(\mS^{\bar{i}i}_{2}
    \ti{R}^{-z}_{12}(2q_{i})P_{12}\Big)\right)-(1-\de_{ij})\tr_{2}\Big(\mS^{\bar{i}j}_{2}
    \ti{R}^{-z}_{12}(q_{ij}^{+})P_{12}\Big)\,.
  }
\end{array}
\eq
Similarly, for $\mM(z)\in\text{Mat}_{2Nn}$, we have
\beq\label{t408}
\begin{array}{c}
  \displaystyle{
    \mM^{ij}(z)=\de_{ij}\tr_{2}\Big(\mS^{ii}_{2}R^{z,(1)}_{12}P_{12}\Big)+(1-\de_{ij})
    \tr_{2}\Big(\mS^{ij}_{2}F^{z}_{12}(q_{ij})P_{12}\Big)\,,
  }
\end{array}
\eq
\beq\label{t409}
\begin{array}{c}
  \displaystyle{
    \mM^{i\bar{j}}(z)=\de_{ij}\left(\frac{1}{2}\ti{Y}^{z}_{1}(q_{i})+\tr_{2}\Big(\mS^{i\bar{i}}_2
    \ti{F}^{z}_{12}(2q_{i})P_{12}\Big)\right) +(1-\de_{ij})\tr_{2}\Big(\mS^{i\bar{j}}_{2}
    \ti{F}^{z}_{12}(q_{ij}^{+})P_{12}\Big)\,,
  }
\end{array}
\eq
\beq\label{t410}
\begin{array}{c}
  \displaystyle{
    \mM^{\bar{i}\bar{j}}(z)=\de_{ij}\tr_{2}\Big(\mS^{\bar{i}\bar{i}}_2R_{12}^{-z,(1)}P_{12}\Big)
    +(1-\de_{ij})\tr_{2}\Big(\mS^{\bar{i}\bar{j}}_{2}F^{-z}_{12}(q_{ij})P_{12}\Big)\,,
  }
\end{array}
\eq
\beq\label{t411}
\begin{array}{c}
  \displaystyle{
    \mM^{\bar{i}j}(z)=\de_{ij}\left(\frac{1}{2}\ti{Y}^{-z}_{1}(q_{i})+\tr_{2}\Big(\mS^{\bar{i}i}_{2}
    \ti{F}^{-z}_{12}(2q_{i})P_{12}\Big)\right)+(1-\de_{ij})\tr_{2}\Big(\mS^{\bar{i}j}_2
    \ti{F}^{-z}_{12}(q_{ij}^{+})P_{12}\Big)\,.
  }
\end{array}
\eq

\underline{Equations of motion.}
For the off-diagonal ($i\neq j$) $N\times N$ blocks of $\mS$,
\beq\label{t417}
\begin{array}{c}
  \displaystyle{
    \dot{\mS}^{ij} = {} \sum_{k:k\neq i,j}^{n}\tr_2\Big(P_{12}\mS_{2}^{ik}
    \mS_1^{kj}F_{21}^{0}(q_{kj})\Big) - \sum_{k:k\neq i,j}^{n}
    \tr_2\Big(F_{12}^{0}(q_{ik})P_{12}\mS_2^{ik}\mS_1^{kj}\Big) +
  }
  \\ \ \\
  \displaystyle{
    + \sum_{k:k\neq i,j}^{n}\tr_{2}\Big(P_{12}\mS_{2}^{i\bar{k}}\mS_{1}^{\bar{k}j}
    \ti{F}_{21}^{0}(q_{kj}^{+})\Big) - \sum_{k:k\neq i,j}^{n}
    \tr_2\Big(\ti{F}_{12}^{0}(q_{ik}^{+})P_{12}\mS_{2}^{i\bar{k}}\mS_{1}^{\bar{k}j}\Big) +
  }
\end{array}
\eq
$$
\begin{array}{c}
  \displaystyle{
    + \tr_2\Big(P_{12}\mS_{2}^{ii}\mS_{1}^{ij}F_{21}^{0}(q_{ij})\Big)
    -\tr_2\Big(F_{12}^{0}(q_{ij})P_{12}\mS_{2}^{ij}\mS_{1}^{jj}\Big)+\frac{1}{2}\,
    \mS^{i\bar{j}}_{1}\ti{Y}^{0}_{1}(q_{j}) +
  }
  \\ \ \\
  \displaystyle{
    + \tr_2\Big(P_{12}\mS_2^{ij}\mS_1^{jj}r_{21}^{(1)}\Big)
    -\tr_2\Big(r_{12}^{(1)}P_{12}\mS_2^{ii}\mS_1^{ij}\Big)-\frac{1}{2}\,
    \ti{Y}^{0}_{1}(q_{i})\mS^{i\bar{j}}_{1} +
  }
  \\ \ \\
  \displaystyle{
    + \tr_{2}\Big(P_{21}\mS_{2}^{i\bar{j}}\mS_{1}^{\bar{j}j}\ti{F}^{0}_{21}(2q_{j})\Big)
    -\tr_{2}\Big(\ti{F}^{0}_{12}(q_{ij}^{+})P_{12}\mS_{2}^{i\bar{j}}\mS_{1}^{\bar{j}j}\Big)+
  }
  \\ \ \\
  \displaystyle{
    +\tr_{2}\Big(P_{12}\mS_{2}^{i\bar{i}}\mS_{1}^{\bar{i}j}\ti{F}^{0}_{21}(q_{ij}^{+})\Big)
    -\tr_{2}\Big(\ti{F}^{0}_{12}(2q_{i})P_{12}\mS_{2}^{i\bar{i}}\mS_{1}^{\bar{i}j}\Big)\,,
  }
\end{array}
$$

\beq\label{t418}
\begin{array}{c}
  \displaystyle{
    \dot{\mS}^{i\bar{j}} = \sum_{k:k\neq i,j}\tr_2\Big(P_{12}\mS_{2}^{ik}\mS_1^{k\bar{j}}
    \ti{F}_{21}^{0}(q_{kj}^{+})\Big) -\sum_{k:k\neq i,j}\tr_2\Big(F_{12}^{0}(q_{ik})P_{12}
    \mS_2^{ik}\mS_1^{k\bar{j}}\Big) +
  }
  \\ \ \\
  \displaystyle{
    + \sum_{k:k\neq i,j}\tr_{2}\Big(P_{12}\mS_{2}^{i\bar{k}}\mS_{1}^{kj}F_{21}^{0}(q_{kj})\Big)
    -\sum_{k:k\neq i,j}\tr_2\Big(\ti{F}_{12}^{0}(q_{ik}^{+})P_{12}\mS_2^{i\bar{k}}\mS_1^{kj}\Big) +
  }
  \\ \ \\
  \displaystyle{
    + \tr_2\Big(P_{12}\mS_2^{ii}\mS_1^{i\bar{j}}\ti{F}_{21}^{0}(q_{ij}^{+})\Big)
    -\tr_2\Big(\ti{F}_{12}^{0}(q_{ij}^{+})P_{12}\mS_2^{i\bar{j}}\mS_1^{\bar{j}\bar{j}}\Big)
    +\frac{1}{2}\,\mS^{ij}_{1}\ti{Y}^{0}_{1}(q_{j}) +
  }
  \\ \ \\
  \displaystyle{
    + \tr_2\Big(P_{12}\mS_2^{i\bar{j}}\mS_1^{\bar{j}\bar{j}}r_{21}^{(1)}\Big)
    -\tr_2\Big(r_{12}^{(1)}P_{12}\mS_2^{ii}\mS_1^{i\bar{j}}\Big)-\frac{1}{2}\,
    \ti{Y}^{0}_{1}(q_{i})\mS^{ij}_{1} +
  }
  \\ \ \\
  \displaystyle{
    + \tr_2\Big(P_{12}\mS_{2}^{ij}\mS_1^{j\bar{j}}\ti{F}_{21}^{0}(2q_{j})\Big)
    -\tr_2\Big(F_{12}^{0}(q_{ij})P_{12}\mS_2^{ij}\mS_1^{j\bar{j}}\Big) +
  }
  \\ \ \\
  \displaystyle{
    + \tr_{2}\Big(P_{12}\mS_{2}^{i\bar{i}}\mS_{1}^{\bar{i}\bar{j}}F_{21}^{0}(q_{ij})\Big)
    -\tr_2\Big(\ti{F}_{12}^{0}(2q_{i})P_{12}\mS_2^{i\bar{i}}\mS_1^{\bar{i}\bar{j}}\Big)\,.
  }
\end{array}
\eq
for the diagonal $N\times N$ blocks of $\mS$,
\beq\label{t419}
\begin{array}{c}
  \displaystyle{
    \dot{\mS}^{ii} = [\mS^{ii},\tr_2\Big(m_{12}(0)\mS_2^{ii}\Big)] +\frac{1}{2}\,
    \mS^{i\bar{i}}_{1}\ti{Y}^{0}_{1}(q_{i})-\frac{1}{2}\,\ti{Y}^{0}_{1}(q_{i})\mS^{i\bar{i}}_{1} +
  }
  \\ \ \\
  \displaystyle{
    + \sum_{k:k\neq i}^n\left(\tr_2\Big(P_{12}\mS_2^{ik}\mS_1^{ki}F_{12}^{0}(q_{ik})\Big)
    -\tr_2\Big(F_{12}^{0}(q_{ik})P_{12}\mS_2^{ik}\mS_1^{ki}\Big)\right) +
  }
  \\ \ \\
  \displaystyle{
    + \sum_{k:k\neq i}^n\left(\tr_2\Big(P_{12}\mS_2^{i\bar{k}}\mS_1^{\bar{k}i}
    \ti{F}_{12}^{0}(q_{ik}^{+})\Big)-\tr_2\Big(\ti{F}_{12}^{0}(q_{ik}^{+})P_{12}
    \mS_2^{i\bar{k}}\mS_1^{\bar{k}i}\Big)\right) +
  }
  \\ \ \\
  \displaystyle{
    + \tr_{2}\Big(P_{12}\mS_{2}^{i\bar{i}}\mS_{1}^{\bar{i}i}\ti{F}^{0}_{21}(2q_{i})\Big)
    -\tr_{2}\Big(\ti{F}^{0}_{12}(2q_{i})P_{12}\mS_{2}^{i\bar{i}}\mS_{1}^{\bar{i}i}\Big)\,,
  }
\end{array}
\eq
\beq\label{t420}
\begin{array}{c}
  \displaystyle{
    \dot{\mS}^{i\bar{i}} = \sum_{k:k\neq i}^n\tr_2\Big(P_{21}\mS_2^{ik}\mS_1^{k\bar{i}}
    \ti{F}_{21}^{0}(q_{ik}^{+})\Big)-\sum_{k:k\neq i}^n \tr_2\Big(F_{12}^{0}(q_{ik})P_{12}
    \mS_2^{ik}\mS_1^{k\bar{i}}\Big) +
  }
  \\ \ \\
  \displaystyle{
    + \sum_{k:k\neq i}^n\tr_2\Big(P_{12}\mS_2^{i\bar{k}}
    \mS_1^{\bar{k}\bar{i}}F_{21}^{0}(q_{ki})\Big) -\sum_{k:k\neq i}^n
    \tr_2\Big(\ti{F}_{12}^{0}(q_{ik}^{+})P_{12}\mS_2^{i\bar{k}}\mS_1^{\bar{k}\bar{i}}\Big)+
  }
  \\ \ \\
  \displaystyle{
    + \tr_2\Big(P_{12}\mS_2^{ii}\mS_1^{i\bar{i}}\ti{F}_{21}^{0}(2q_i)\Big)
    -\tr_2\Big(\ti{F}_{12}^{0}(2q_i)P_{12}\mS_2^{i\bar{i}}\mS_1^{\bar{i}\bar{i}}\Big)
    -\frac{1}{2}\,\ti{Y}^{0}_{1}(q_{i})\mS^{ii}_{1} +
  }
  \\ \ \\
  \displaystyle{
    + \tr_2\Big(P_{12}\mS_2^{i\bar{i}}\mS_1^{\bar{i}\bar{i}}r_{21}^{(1)}\Big)
    -\tr_2\Big(r_{12}^{(1)}P_{12}\mS_{2}^{ii}\mS_{1}^{i\bar{i}}\Big) +\frac{1}{2}\,
    \mS^{ii}_{1}\ti{Y}^{0}_{1}(q_{i})\,,
  }
\end{array}
\eq
and for the momenta,
\beq\label{t421}
\begin{array}{c}
  \displaystyle{
    \dot{p}_{i} = - \sum_{k:k \neq i}^{n} \tr_{23} \Big(\partial_{q_{ik}}F_{32}^{0}(q_{ik}) P_{23}
    \mS_2^{ik}\mS_3^{ki}\Big)-\tr_{23}\Big(\partial_{2q_{i}}\ti{F}^{0}_{32}(2q_{i})P_{23}
    \mS_{2}^{i\bar{i}}\mS_{3}^{\bar{i}i}\Big)-
  }
  \\ \ \\
  \displaystyle{
    - \sum_{k:k \neq i} \tr_{23} \Big(\partial_{q_{ik}^{+}} \ti{F}_{32}^{0} (q_{ik}^{+})P_{23}
    \mS_2^{i\bar{k}} \mS_3^{\bar{k}i}\Big)+ \frac{1}{2}\, \p_{q_i}u_N(q_i|\nu)
    -\frac{1}{2}\,\tr_{2}\Big(\partial_{q_{i}}\ti{Y}^{0}_{2}(q_{i})\mS^{i\bar{i}}_{2}\Big)\,.
  }
\end{array}
\eq
\begin{proof}
Throughout the calculation, we use $p_i=\dot q_i$, as follows from
the Hamiltonian equations.

\textbf{1.} We begin with the off-diagonal block $i\bar j$,
where $i\neq j$. The left-hand side of the Lax equation is
\beq\label{b01}
\begin{array}{c}
  \displaystyle{
    \text{l.h.s.} =\tr_{2}\Big(\dot{\mS}^{i\bar{j}}_{2}\ti{R}_{12}^{z}(q_{ij}^{+})P_{12}\Big)
    +(\dot{q}_{i}+\dot{q}_{j})\tr_{2}\Big(\mS^{i\bar{j}}_{2}
    \ti{F}^{z}_{12}(q_{ij}^{+})P_{12}\Big) \,.
  }
\end{array}
\eq
The right-hand side is
\beq\label{b02}
\begin{array}{c}
  \displaystyle{
    \text{r.h.s.}=\mL^{i\bar{j}}(z)\mM^{\bar{j}\bar{j}}(z)-\mM^{ii}(z)\mL^{i\bar{j}}(z)
    +\mL^{ii}(z)\mM^{i\bar{j}}(z)-\mM^{i\bar{j}}(z)\mL^{\bar{j}\bar{j}}(z) +
  }
  \\ \ \\
  \displaystyle{
    +\mL^{ij}(z)\mM^{j\bar{j}}(z)-\mM^{i\bar{i}}(z)\mL^{\bar{i}\bar{j}}(z)+\mL^{i\bar{i}}(z)
    \mM^{\bar{i}\bar{j}}(z)-\mM^{ij}(z)\mL^{j\bar{j}}(z) +
  }
  \\ [3mm]
  \displaystyle{
    + \sum_{k:k\neq i,j}^{n}\left(\mL^{ik}(z)\mM^{k\bar{j}}(z)-\mM^{ik}(z)\mL^{k\bar{j}}(z)\right)
    + \sum_{k:k\neq i,j}^{n}\left(\mL^{i\bar{k}}(z)\mM^{\bar{k}\bar{j}}(z)-\mM^{i\bar{k}}(z)
    \mL^{\bar{k}\bar{j}}(z)\right)\,.
  }
\end{array}
\eq
The last two sums are evaluated using (\ref{b03})--(\ref{b04}).

For $k\neq i,j$ we have
\beq\label{b05}
\begin{array}{c}
  \displaystyle{
    \mL^{ik}(z)\mM^{k\bar{j}}(z)-\mM^{ik}(z)\mL^{k\bar{j}}(z) =
  }
  \\ \ \\
  \displaystyle{
    =\tr_{23}\Big(\mS_{2}^{ik}R_{12}^{z}(q_{ik})P_{12}\mS_{3}^{k\bar{j}}
    \ti{F}_{13}^{z}(q_{kj}^{+})P_{13}\Big) -\tr_{23}\Big(\mS_{2}^{ik}F_{12}^{z}(q_{ik})P_{12}
    \mS_{3}^{k\bar{j}}\ti{R}_{13}^{z}(q_{kj}^{+})P_{13}\Big) =
  }
  \\ \ \\
  \displaystyle{
    =\tr_{23}\Big(\left(R_{12}^{z}(q_{ik})\ti{F}_{23}^{z}(q_{kj}^{+})-F_{12}^{z}(q_{ik})
    \ti{R}_{23}^{z}(q_{kj}^{+})\right)P_{12}P_{13}\mS_{2}^{ik}\mS_{3}^{k\bar{j}}\Big)
    \stackrel{(\ref{b03})}{=}
  }
  \\ \ \\
  \displaystyle{
    =\tr_{23}\Big(\left(\ti{F}^{0}_{23}(q_{kj}^{+})\ti{R}_{13}^{z}(q_{ij}^{+})
    -\ti{R}_{13}^{z}(q_{ij}^{+})F^{0}_{12}(q_{ik})\right)P_{12}P_{13}\mS_{2}^{ik}
    \mS_{3}^{k\bar{j}}\Big) =
  }
  \\ \ \\
  \displaystyle{
    =\tr_{23}\Big(\ti{R}_{13}^{z}(q_{ij}^{+})P_{13}\left(P_{23}\mS_{2}^{ik}\mS_{3}^{k\bar{j}}
    \ti{F}^{0}_{23}(q_{kj}^{+})-F^{0}_{32}(q_{ik})P_{23}\mS^{ik}_{2}
    \mS^{k\bar{j}}_{3}\right)\Big)\,,
  }
\end{array}
\eq
and
\beq\label{b06}
\begin{array}{c}
  \displaystyle{
    \mL^{i\bar{k}}(z)\mM^{\bar{k}\bar{j}}(z)-\mM^{i\bar{k}}(z)\mL^{\bar{k}\bar{j}}(z) =
  }
  \\ \ \\
  \displaystyle{
    =\tr_{23}\Big(\mS_{2}^{i\bar{k}}\ti{R}_{12}^{z}(q_{ik}^{+})P_{12}
    \mS_{3}^{\bar{k}\bar{j}}F_{13}^{-z}(q_{kj})P_{13}\Big) +\tr_{23}\Big(\mS_{2}^{i\bar{k}}
    \ti{F}_{12}^{z}(q_{ik}^{+})P_{12}\mS_{3}^{\bar{k}\bar{j}}R_{13}^{-z}(q_{kj})P_{13}\Big) =
  }
  \\ \ \\
  \displaystyle{
    =\tr_{23}\Big(\left(\ti{R}_{12}^{z}(q_{ik}^{+})F_{23}^{-z}(q_{kj})
    +\ti{F}_{12}^{z}(q_{ik}^{+})R_{23}^{-z}(q_{kj})\right)P_{12}P_{13}\mS_{2}^{i\bar{k}}
    \mS_{3}^{\bar{k}\bar{j}}\Big) \stackrel{(\ref{b04})}{=}
  }
  \\ \ \\
  \displaystyle{
    =\tr_{23}\Big(\left(F^{0}_{23}(q_{kj})\ti{R}_{13}^{z}(q_{ij}^{+})-\ti{R}_{13}^{z}(q_{ij}^{+})
    \ti{F}^{0}_{12}(q_{ik}^{+})\right)P_{12}P_{13}\mS_{2}^{i\bar{k}}\mS_{3}^{\bar{k}\bar{j}}\Big) =
  }
  \\ \ \\
  \displaystyle{
    =\tr_{23}\Big(\ti{R}_{13}^{z}(q_{ij}^{+})P_{13}\left(P_{23}\mS_{2}^{i\bar{k}}
    \mS_{3}^{\bar{k}\bar{j}}F^{0}_{23}(q_{kj})-\ti{F}^{0}_{32}(q_{ik}^{+})P_{23}
    \mS^{i\bar{k}}_{2}\mS^{\bar{k}\bar{j}}_{3}\right)\Big)\,.
  }
\end{array}
\eq
These expressions give the first two lines of (\ref{t418}).
The momentum-dependent term $(p_i+p_j)\mM^{i\bar j}(z)$ on the
right-hand side of (\ref{b02}) is canceled by the last term on the
left-hand side of (\ref{b01}). We next evaluate the parts of the first
and fourth terms in (\ref{b02}) that are independent of the momenta:
\beq\label{b09}
\begin{array}{c}
  \displaystyle{
    \mL^{i\bar{j}}(z)\mM^{\bar{j}\bar{j}}(z)-\mM^{i\bar{j}}(z)(\mL^{\bar{j}\bar{j}}(z)+p_{j}1_{N}) =
  }
  \\ \ \\
  \displaystyle{
    =\tr_{23}\Big(\left(\ti{R}^{z}_{12}(q_{ij}^{+})R^{-z,(1)}_{23}
    +\ti{F}^{z}_{12}(q_{ij}^{+})R^{-z,(0)}_{23}\right)P_{12}P_{13}\mS^{i\bar{j}}_{2}
    \mS^{\bar{j}\bar{j}}_{3}\Big)\stackrel{(\ref{b07})}{=}
  }
  \\ \ \\
  \displaystyle{
    =\tr_{23}\Big(\left(r^{(1)}_{23}\ti{R}^{z}_{13}(q_{ij}^{+})-\ti{R}^{z}_{13}(q_{ij}^{+})
    \ti{F}^{0}_{12}(q_{ij}^{+})-\frac{1}{2}\,\partial^{2}_{q_{ij}^{+}}
    \ti{R}^{z}_{12}(q_{ij}^{+})P_{23}\right)P_{12}P_{13}\mS^{i\bar{j}}_{2}
    \mS^{\bar{j}\bar{j}}_{3}\Big)\,,
  }
\end{array}
\eq
The first two terms in this expression give two contributions to the
third and fourth lines of (\ref{t418}). The last term in (\ref{b09})
is an additional term that must cancel.

Similarly, (\ref{b08}) gives
\beq\label{b10}
\begin{array}{c}
  \displaystyle{
    (\mL^{ii}(z)-p_{i}1_{N})\mM^{i\bar{j}}(z)-\mM^{ii}(z)\mL^{i\bar{j}}(z) =
  }
  \\ \ \\
  \displaystyle{
    = \tr_{23}\Big(\left(R_{12}^{(0),z}\ti{F}_{23}^{z}(q_{ij}^{+})- R_{12}^{z,(1)}
    \ti{R}_{23}^{z} (q_{ij}^{+})\right)P_{12}P_{13}\mS^{ii}_{2}\mS^{i\bar{j}}_{3}\Big)
    \stackrel{(\ref{b08})}{=}
  }
  \\ \ \\
  \displaystyle{
    =\tr_{23}\Big(\left(\ti{F}^{0}_{23}(q_{ij}^{+})\ti{R}^{z}_{13}(q_{ij}^{+})
    -\ti{R}^{z}_{13}(q_{ij}^{+})r^{(1)}_{12}+\frac{1}{2}\,\partial_{q_{ij}^{+}}^{2}
    \ti{R}^{z}_{13}(q_{ij}^{+})P_{12}\right)P_{12}P_{13}\mS^{ii}_{2}\mS^{i\bar{j}}_{3}\Big)\,.
  }
\end{array}
\eq
The first two terms give the other two contributions to the third and
fourth lines of (\ref{t418}). The last term cancels the last term in
(\ref{b09}) after taking the partial traces and imposing
$\tr\Big(\mS^{ii}\Big)=\tr\Big(\mS^{\bar{j}\bar{j}}\Big)$,
as required by (\ref{t416}).

The terms involving $\ti K$-matrices are evaluated using (\ref{t229}):
\beq\label{b12}
\begin{array}{c}
  \displaystyle{
    \frac{1}{2}\,\mL^{ij}(z)\ti{Y}^{z}_{1}(q_{j})- \mM^{ij}(z) \ti{K}^{z}_{1}(q_{j})
    +\ti{K}^{z}_{1}(q_{i})\mM^{\bar{i}\bar{j}}(z)-\frac{1}{2}\,\ti{Y}^{z}_{1}(q_{i})
    \mL^{\bar{i}\bar{j}}(z) \stackrel{(\ref{t4012})}{=}
  }
  \\ \ \\
  \displaystyle{
    =\tr_{2}\Big(\Bigl(\frac{1}{2}\,R_{12}^{z}(q_{ij})\ti{Y}^{z}_{2}(q_{j})
    -F_{12}^{z}(q_{ij})\ti{K}^{z}_{2}(q_{j}) +
  }
  \\ \ \\
  \displaystyle{
    +\ti{K}^{z}_{1}(q_{i})F_{12}^{-z}(q_{ij})+\frac{1}{2}\,
    \ti{Y}^{z}_{1}(q_{i})R_{12}^{-z}(q_{ij})\Bigr)P_{12}\mS_{2}^{ij}\Big) =
  }
  \\ \ \\
  \displaystyle{
    \stackrel{(\ref{t229})}{=}\frac{1}{2}\,\tr_{2}\Big(\left(\ti{Y}^{(0)}_{2}(q_{j})
    \ti{R}_{12}^{z}(q_{ij}^{+})-\ti{R}^{z}_{12}(q_{ij}^{+})\ti{Y}_{1}^{(0)}(q_{i})\right)P_{12}
    \mS_{2}^{ij}\Big)\,.
  }
\end{array}
\eq
The remaining terms are
\beq\label{b13}
\begin{array}{c}
  \displaystyle{
    \mL^{ij}(z)(\mM^{j\bar{j}}(z)-\frac{1}{2}\,\ti{Y}^{z}_{1}(q_{j}))
    -\mM^{ij}(z)(\mL^{j\bar{j}}(z)- \ti{K}^{z}_{1}(q_{j})) =
  }
  \\ \ \\
  \displaystyle{
    =\tr_{23}\Big(\left(R_{12}^{z}(q_{ij})\ti{F}_{23}^{z}(2q_{j})-F_{12}^{z}(q_{ij})
    \ti{R}_{23}^{z}(2q_{j})\right)P_{12}P_{13}\mS_{2}^{ij}\mS_{3}^{j\bar{j}}\Big)
    \stackrel{(\ref{b03})}{=}
  }
  \\ \ \\
  \displaystyle{
    =\tr_{23}\Big(\left(\ti{F}^{0}_{23}(2q_{j})\ti{R}_{13}^{z}(q_{ij}^{+})
    -\ti{R}_{13}^{z}(q_{ij}^{+})F^{0}_{12}(q_{ij})\right)P_{12}P_{13}\mS_{2}^{ij}
    \mS_{3}^{j\bar{j}}\Big)\,,
  }
\end{array}
\eq
and
\beq\label{b14}
\begin{array}{c}
  \displaystyle{
    (\mL^{i\bar{i}}(z)-\ti{K}^{z}_{1}(q_{i}))\mM^{\bar{i}\bar{j}}(z)-(\mM^{i\bar{i}}(z)
    -\frac{1}{2}\,\ti{Y}^{z}_{1}(q_{i}))\mL^{\bar{i}\bar{j}}(z) =
  }
  \\ \ \\
  \displaystyle{
    =\tr_{23}\Big(\left(\ti{R}_{12}^{z}(2q_{i})F_{23}^{-z}(q_{ij})
    +\ti{F}_{12}^{z}(2q_{i})R_{23}^{-z}(q_{ij})\right)P_{12}P_{13}\mS_{2}^{i\bar{i}}
    \mS_{3}^{\bar{i}\bar{j}}\Big)\stackrel{(\ref{b04})}{=}
  }
  \\ \ \\
  \displaystyle{
    =\tr_{23}\Big(\left(F^{0}_{23}(q_{ij})\ti{R}_{13}^{z}(q_{ij}^{+})-\ti{R}_{13}^{z}(q_{ij}^{+})
    \ti{F}^{0}_{12}(2q_{i})\right)P_{12}P_{13}\mS_{2}^{i\bar{i}}\mS_{3}^{\bar{i}\bar{j}}\Big)\,.
  }
\end{array}
\eq
\textbf{2.} We next consider the $i\bar i$ block of the Lax equation.
Its left-hand side is
\beq\label{b15}
\begin{array}{c}
  \displaystyle{
    \text{l.h.s.} = \dot{\mL}^{i\bar{i}}(z)=\tr_{2}\Big(\dot{\mS}^{i\bar{i}}_{2}
    \ti{R}^{z}_{12}(2q_{i})P_{12}\Big) + 2\dot{q}_{i}\tr_{2}\left(\mS^{i\bar{i}}_{2}
    \ti{F}^{z}_{12}(2q_{i})P_{12}\right)+\dot{q}_{i}\ti{Y}_{1}^{z}(q_{i})\,.
  }
\end{array}
\eq
The right-hand side is
\beq\label{b16}
\begin{array}{c}
  \displaystyle{
    \text{r.h.s.}=\mL^{i\bar{i}}(z)\mM^{\bar{i}\bar{i}}(z)-\mM^{ii}(z)\mL^{i\bar{i}}(z)
    +\mL^{ii}(z)\mM^{i\bar{i}}(z)-\mM^{i\bar{i}}(z)\mL^{\bar{i}\bar{i}}(z) +
  }
  \\ \ \\
  \displaystyle{
    + \sum_{k:k\neq i}^{n}\left(\mL^{ik}(z)\mM^{k\bar{i}}(z)-\mM^{ik}(z)\mL^{k\bar{i}}(z)\right)
    + \sum_{k:k\neq i}^{n}\left(\mL^{i\bar{k}}(z)\mM^{\bar{k}\bar{i}}(z)-\mM^{i\bar{k}}(z)
    \mL^{\bar{k}\bar{i}}(z)\right)\,.
  }
\end{array}
\eq
Again, the last two sums are evaluated using (\ref{b03}) and
(\ref{b04}). For $k\neq i$, we obtain
\beq\label{b17}
\begin{array}{c}
  \displaystyle{
    \mL^{ik}(z)\mM^{k\bar{i}}(z)-\mM^{ik}(z)\mL^{k\bar{i}}(z)=
  }
  \\ \ \\
  \displaystyle{
    \stackrel{(\ref{b03})}{=}\tr_{23}\Big(\ti{R}_{13}^{z}(2q_{i})P_{13}\left(P_{23}
    \mS_{2}^{ik}\mS_{3}^{k\bar{i}}\ti{F}^{0}_{23}(q_{ki}^{+})-F^{0}_{32}(q_{ik})P_{23}
    \mS^{ik}_{2}\mS^{k\bar{i}}_{3}\right)\Big)\,,
  }
\end{array}
\eq
$$
\begin{array}{c}
  \displaystyle{
    \mL^{i\bar{k}}(z)\mM^{\bar{k}\bar{i}}(z)-\mM^{i\bar{k}}(z)\mL^{\bar{k}\bar{i}}(z) =
  }
  \\ \ \\
  \displaystyle{
    \stackrel{(\ref{b04})}{=}\tr_{23}\Big(\ti{R}_{13}^{z}(2q_{i})P_{13}\left(P_{23}
    \mS_{2}^{i\bar{k}}\mS_{3}^{\bar{k}\bar{i}}F^{0}_{23}(q_{ki})-\ti{F}^{0}_{32}(q_{ik}^{+})P_{23}
    \mS^{i\bar{k}}_{2}\mS^{\bar{k}\bar{i}}_{3}\right)\Big)\,.
  }
\end{array}
$$
The full momentum-dependent contribution on the right-hand side of
(\ref{b16}) is $2p_i\mM^{i\bar i}(z)$. Its part involving $\ti Y$ is
$p_i\ti Y^z_1(q_i)$. Since $p_i=\dot q_i$, this full contribution
equals the sum of the second and third terms on the left-hand side of
(\ref{b15}). We now consider the parts of the third and fourth terms in
(\ref{b16}) that contain neither momenta nor $\ti K$- or $\ti Y$-matrices:
\beq\label{b18}
\begin{array}{c}
  \displaystyle{
    (\mL^{ii}(z)-p_{i}1_{N})(\mM^{i\bar{i}}(z)-\frac{1}{2}\,\ti{Y}^{z}_{1}(q_{i}))
    -(\mM^{i\bar{i}}(z)-\frac{1}{2}\,\ti{Y}^{z}_{1}(q_{i}))(\mL^{\bar{i}\bar{i}}(z)+p_{i}1_{N}) +
  }
  \\ \ \\
  \displaystyle{
    + (\mL^{i\bar{i}}(z)-\ti{K}^{z}_{1}(q_{i}))\mM^{\bar{i}\bar{i}}(z)
    -\mM^{ii}(z)(\mL^{i\bar{i}}(z)-\ti{K}^{z}_{1}(q_{i})) =
  }
  \end{array}
\eq
$$
\begin{array}{c}
  \displaystyle{
    =\tr_{23}\Big(R_{12}^{(0),z}\ti{F}_{23}^{z}(2q_{i})P_{12}P_{13}\mS^{ii}_{2}
    \mS^{i\bar{i}}_{3}\Big) +\tr_{23}\Big(\ti{F}^{z}_{12}(2q_{i})R_{23}^{(0),-z}P_{12}P_{13}
    \mS^{i\bar{i}}_{2}\mS^{\bar{i}\bar{i}}_{3}\Big) +
  }
  \\ \ \\
  \displaystyle{
    + \tr_{23}\Big( \ti{R}_{12}^{z} (2q_{i}) R_{23}^{-z,(1)}P_{12}P_{13}\mS_{2}^{i\bar{i}}
    \mS_{3}^{\bar{i}\bar{i}}\Big) - \tr_{23}\Big(R_{12}^{z,(1)}\ti{R}_{23}^{z} (2q_{i})P_{12} P_{13}
    \mS_{2}^{ii}\mS_{3}^{i\bar{i}}\Big)\stackrel{(\ref{b07}),(\ref{b08})}{=}
  }
  \\ \ \\
  \displaystyle{
    =\tr_{23}\Big(\ti{F}^{0}_{23}(2q_{i})\ti{R}_{13}^{z}(2q_{i})P_{12}P_{13}\mS_{2}^{ii}
    \mS_{3}^{i\bar{i}}\Big)-\tr_{23}\Big(\ti{R}_{13}^{z}(2q_{i})\ti{F}^{0}_{12}(2q_{i})P_{12}P_{13}
    \mS_{2}^{i\bar{i}}\mS_{3}^{\bar{i}\bar{i}}\Big) -
  }
\end{array}
$$
$$
\begin{array}{c}
  \displaystyle{
    -\tr_{23}\Big(\ti{R}_{13}^{z}(2q_{i})r^{(1)}_{12}P_{12}P_{13}\mS_{2}^{ii}
    \mS_{3}^{i\bar{i}}\Big) +\frac{1}{2}\,\tr_{23}\Big(\partial_{2q_{i}}^{2}
    \ti{R}_{13}^{z}(2q_{i})P_{12}P_{12}P_{13}\mS_{2}^{ii}\mS_{3}^{i\bar{i}}\Big) +
  }
  \\ \ \\
  \displaystyle{
    +\tr_{23}\Big(r^{(1)}_{23}\ti{R}_{13}^{z}(2q_{i})P_{12}P_{13}\mS_{2}^{i\bar{i}}
    \mS_{3}^{\bar{i}\bar{i}}\Big) -\frac{1}{2}\,\tr_{23}\Big(P_{23}\partial_{2q_{i}}^{2}
    \ti{R}_{13}^{z}(2q_{i})P_{12}P_{13}\mS_{2}^{i\bar{i}}\mS_{3}^{\bar{i}\bar{i}}\Big)\,.
  }
\end{array}
$$
The terms involving $\ti K$- and $\ti Y$-matrices but no momenta are
\beq\label{b19}
\begin{array}{c}
  \displaystyle{
    \frac{1}{2}\,(\mL^{ii}(z)-p_{i}1_{N})\ti{Y}^{z}_{1}(q_{i})-\frac{1}{2}\,
    \ti{Y}^{z}_{1}(q_{i})(\mL^{\bar{i}\bar{i}}(z)+p_{i}1_{N}) +\ti{K}^{z}_{1}(q_{i})
    \mM^{\bar{i}\bar{i}}(z)- \mM^{ii}(z) \ti{K}^{z}_{1}(q_{i}) =
  }
  \\ \ \\
  \displaystyle{
    = \tr_{2}\Big(\left(\frac{1}{2}\,R^{(0),z}_{12}\ti{Y}^{z}_{2}(q_{i})+\frac{1}{2}\,
    \ti{Y}^{z}_{1}(q_{i})R^{(0),-z}_{12}+\ti{K}^{z}_{1}(q_{i})R^{-z,(1)}_{12}-R^{z,(1)}_{12}
    \ti{K}^{z}_{2}(q_{i})\right)P_{12}\mS^{ii}_{2}\Big) =
  }
  \\ \ \\
  \displaystyle{
    = \frac{1}{2}\,\tr_{2}\Big(\left(\ti{Y}^{0}_{2}(q_{i})\ti{R}^{z}_{12}(2q_{i})
    -\ti{R}^{z}_{12}(2q_{i})\ti{Y}^{0}_{1}(q_{i})\right)P_{12}\mS^{ii}_{2}\Big)\,.
  }
\end{array}
\eq
Here we have used (\ref{b20}). The additional terms in (\ref{b18})
cancel, and (\ref{b17})--(\ref{b19}) yield the equations of motion
(\ref{t420}).

\textbf{3.} Now consider the off-diagonal block $ij$, where $i\neq j$.
The left-hand side of the Lax equation is
\beq\label{b21}
\begin{array}{c}
  \displaystyle{
    \text{l.h.s.} =\tr_{2}\Big(\dot{\mS}^{ij}_{2}R_{12}^{z}(q_{ij})P_{12}\Big)
    +(\dot{q}_{i}-\dot{q}_{j})\tr_{2}\Big(\mS^{ij}_{2}F^{z}_{12}(q_{ij})P_{12}\Big)\,.
  }
\end{array}
\eq
The right-hand side is
\beq\label{b22}
\begin{array}{c}
  \displaystyle{
    \text{r.h.s.}=\mL^{i\bar{j}}(z)\mM^{\bar{j}j}(z)-\mM^{ii}(z)\mL^{ij}(z)+\mL^{ii}(z)
    \mM^{ij}(z)-\mM^{i\bar{j}}(z)\mL^{\bar{j}j}(z) +
  }
  \\ \ \\
  \displaystyle{
    +\mL^{ij}(z)\mM^{jj}(z)-\mM^{i\bar{i}}(z)\mL^{\bar{i}j}(z)+\mL^{i\bar{i}}(z)\mM^{\bar{i}j}(z)
    -\mM^{ij}(z)\mL^{jj}(z) +
  }
  \\ \ \\
  \displaystyle{
    + \sum_{k:k\neq i,j}^{n} \left(\mL^{ik}(z)\mM^{kj}(z)-\mM^{ik}(z)\mL^{kj}(z)\right)
    +\sum_{k:k\neq i,j}^{n}\left(\mL^{i\bar{k}}(z)\mM^{\bar{k}j}(z)-\mM^{i\bar{k}}(z)
    \mL^{\bar{k}j}(z)\right)\,.
  }
\end{array}
\eq
The last two sums are evaluated using (\ref{b23}) and (\ref{b24}).
 For $k\neq i,j$ we have
\beq\label{b25}
\begin{array}{c}
  \displaystyle{
    \mL^{ik}(z)\mM^{kj}(z)-\mM^{ik}(z)\mL^{kj}(z) =
  }
  \\ \ \\
  \displaystyle{
    = \tr_{23}\Big(\left(R_{12}^{z}(q_{ik})F_{23}^{z}(q_{kj})
    -F_{12}^{z}(q_{ik})R_{23}^{z}(q_{kj})\right)P_{12}P_{13}\mS_{2}^{ik}\mS_{3}^{kj}\Big)
    \stackrel{(\ref{b23})}{=}
  }
  \\ \ \\
  \displaystyle{
    = \tr_{23}\Big(\left(F^{0}_{23}(q_{kj})R_{13}^{z}(q_{ij})
    -R_{13}^{z}(q_{ij})F^{0}_{12}(q_{ik})\right)P_{12}P_{13}\mS_{2}^{ik}\mS_{3}^{kj}\Big) =
  }
  \\ \ \\
  \displaystyle{
    =\tr_{23}\Big(R_{13}^{z}(q_{ij})P_{13}\left(P_{23}\mS_{2}^{ik}\mS_{3}^{kj}F^{0}_{23}(q_{kj})
    -F^{0}_{32}(q_{ik})P_{23}\mS^{ik}_{2}\mS^{kj}_{3}\right)\Big)\,,
  }
\end{array}
\eq
and
\beq\label{b26}
\begin{array}{c}
  \displaystyle{
    \mL^{i\bar{k}}(z)\mM^{\bar{k}j}(z)-\mM^{i\bar{k}}(z)\mL^{\bar{k}j}(z) =
  }
  \\ \ \\
  \displaystyle{
    = \tr_{23}\Big(\left(\ti{R}_{12}^{z}(q_{ik}^{+})\ti{F}_{23}^{-z}(q_{kj}^{+})
    +\ti{F}_{12}^{z}(q_{ik}^{+})\ti{R}_{23}^{-z}(q_{kj}^{+})\right)P_{12}P_{13}
    \mS_{2}^{i\bar{k}}\mS_{3}^{\bar{k}j}\Big) \stackrel{(\ref{b24})}{=}
  }
  \\ \ \\
  \displaystyle{
    = \tr_{23}\Big(\left(\ti{F}^{0}_{23}(q_{kj}^{+})R_{13}^{z}(q_{ij})-R_{13}^{z}(q_{ij})
    \ti{F}^{0}_{12}(q_{ik}^{+})\right)P_{12}P_{13}\mS_{2}^{i\bar{k}}\mS_{3}^{\bar{k}j}\Big) =
  }
  \\ \ \\
  \displaystyle{
    =\tr_{23}\Big(R_{13}^{z}(q_{ij})P_{13}\left(P_{23}\mS_{2}^{i\bar{k}}\mS_{3}^{\bar{k}j}
    \ti{F}^{0}_{23}(q_{kj}^{+})-\ti{F}^{0}_{32}(q_{ik}^{+})P_{23}\mS^{i\bar{k}}_{2}
    \mS^{\bar{k}j}_{3}\right)\Big)\,.
  }
\end{array}
\eq
These expressions give the first two lines of (\ref{t417}).
The momentum-dependent term $(p_i-p_j)\mM^{ij}(z)$ on the right-hand
side of (\ref{b22}) is canceled by the last term on the left-hand side
of (\ref{b21}). We next evaluate the momentum-independent parts of the
first and fourth terms in the second line of (\ref{b22}):
\beq\label{b29}
\begin{array}{c}
  \displaystyle{
    \mL^{ij}(z)\mM^{jj}(z)-\mM^{ij}(z)(\mL^{jj}(z)-p_{j}1_{N}) =
  }
  \\ \ \\
  \displaystyle{
    =\tr_{23}\Big(\left(R^{z}_{12}(q_{ij})R^{z,(1)}_{23}
    -F^{z}_{12}(q_{ij})R^{z,(0)}_{23}\right)P_{12}P_{13}\mS^{ij}_{2}\mS^{jj}_{3}\Big)
    \stackrel{(\ref{b28})}{=}
  }
  \\ \ \\
  \displaystyle{
    = \tr_{23}\Big(\left(r^{(1)}_{23}R^{z}_{13}(q_{ij})-R^{z}_{13}(q_{ij})F^{0}_{12}(q_{ij})
    -\frac{1}{2}\,\partial^{2}_{q_{ij}}R^{z}_{12}(q_{ij})P_{23}\right)P_{12}P_{13}
    \mS^{ij}_{2}\mS^{jj}_{3}\Big)\,,
  }
\end{array}
\eq
The first two terms give two contributions to the third and fourth
lines of (\ref{t417}). The last term in (\ref{b29}) is an additional
term that must cancel.

Similarly, (\ref{b27}) gives
\beq\label{b30}
\begin{array}{c}
  \displaystyle{
    (\mL^{ii}(z)-p_{i}1_{N})\mM^{ij}(z)-\mM^{ii}(z)\mL^{ij}(z) =
  }
  \\ \ \\
  \displaystyle{
    = \tr_{23}\Big(\left(R_{12}^{(0),z}F_{23}^{z}(q_{ij})
    -R_{12}^{z,(1)} R_{23}^{z} (q_{ij})\right)P_{12}P_{13}\mS^{ii}_{2}\mS^{ij}_{3}\Big)
    \stackrel{(\ref{b27})}{=}
  }
  \\ \ \\
  \displaystyle{
    =\tr_{23}\Big(\left(F^{0}_{23}(q_{ij})R^{z}_{13}(q_{ij})-R^{z}_{13}(q_{ij})r^{(1)}_{12}
    +\frac{1}{2}\,\partial_{q_{ij}}^{2}R^{z}_{13}(q_{ij})P_{12}\right)P_{12}P_{13}
    \mS^{ii}_{2}\mS^{ij}_{3}\Big)\,.
  }
\end{array}
\eq
The first two terms give the other two contributions to the third and
fourth lines of (\ref{t417}). The last term cancels the last term in
(\ref{b29}) after taking the partial traces and imposing
$\tr\Big(\mS^{ii}\Big)=\tr\Big(\mS^{jj}\Big)$,
as required by (\ref{t416}).

The terms involving $\ti K$-matrices are evaluated using (\ref{t230}):
\beq\label{b32}
\begin{array}{c}
  \displaystyle{
    \frac{1}{2}\,\mL^{i\bar{j}}(z)\ti{Y}^{-z}_{1}(q_{j})+ \mM^{i\bar{j}}(z)\ti{K}^{-z}_{1}(q_{j})
    +\ti{K}^{z}_{1}(q_{i})\mM^{\bar{i}j}(z)-\frac{1}{2}\,\ti{Y}^{z}_{1}(q_{i})\mL^{\bar{i}j}(z)
    \stackrel{(\ref{t4012})}{=}
  }
  \\ \ \\
  \displaystyle{
    =\tr_{2}\Big(\Bigl(\frac{1}{2}\,\ti{R}_{12}^{z}(q_{ij}^{+})\ti{Y}^{-z}_{2}(q_{j})
    +\ti{F}_{12}^{z}(q_{ij}^{+})\ti{K}^{-z}_{2}(q_{j}) +
  }
  \\ \ \\
  \displaystyle{
    +\ti{K}^{z}_{1}(q_{i})\ti{F}_{12}^{-z}(q_{ij}^{+})+\frac{1}{2}\,\ti{Y}^{z}_{1}(q_{i})
    \ti{R}_{12}^{-z}(q_{ij}^{+})\Bigr)P_{12}\mS_{2}^{i\bar{j}}\Big) =
  }
  \\ \ \\
  \displaystyle{
    \stackrel{(\ref{t230})}{=}\frac{1}{2}\,\tr_{2}\Big(\left(
    \ti{Y}^{(0)}_{2}(q_{j})R_{12}^{z}(q_{ij})-R^{z}_{12}(q_{ij})\ti{Y}_{1}^{(0)}(q_{i})\right)P_{12}
    \mS_{2}^{i\bar{j}}\Big) \,.
  }
\end{array}
\eq
The remaining terms are
\beq\label{b33}
\begin{array}{c}
  \displaystyle{
    \mL^{i\bar{j}}(z)(\mM^{\bar{j}j}(z)-\frac{1}{2}\,\ti{Y}^{-z}_{1}(q_{j}))
    -\mM^{i\bar{j}}(z)(\mL^{\bar{j}j}(z)+ \ti{K}^{-z}_{1}(q_{j})) =
  }
  \\ \ \\
  \displaystyle{
    = \tr_{23}\Big(\left(\ti{R}_{12}^{z}(q_{ij}^{+})\ti{F}_{23}^{-z}(2q_{j})
    +\ti{F}_{12}^{z}(q_{ij}^{+})\ti{R}_{23}^{-z}(2q_{j})\right)P_{12}P_{13}\mS_{2}^{i\bar{j}}
    \mS_{3}^{\bar{j}j}\Big) \stackrel{(\ref{b24})}{=}
  }
  \\ \ \\
  \displaystyle{
    = \tr_{23}\Big(\left(\ti{F}^{0}_{23}(2q_{j})R_{13}^{z}(q_{ij})-R_{13}^{z}(q_{ij})
    \ti{F}^{0}_{12}(q_{ij}^{+})\right)P_{12}P_{13}\mS_{2}^{i\bar{j}}\mS_{3}^{\bar{j}j}\Big) =
  }
  \\ \ \\
  \displaystyle{
    = \tr_{23}\Big(R_{13}^{z}(q_{ij})P_{13}\left(P_{23}\mS_{2}^{i\bar{j}}\mS_{3}^{\bar{j}j}
    \ti{F}^{0}_{23}(2q_{j})-\ti{F}^{0}_{32}(q_{ij}^{+})P_{23}\mS^{i\bar{j}}_{2}
    \mS^{\bar{j}j}_{3}\right)\Big)\,,
  }
\end{array}%
\eq
\textbf{4.} Finally, consider the diagonal block $ii$ of the Lax
equation. Its left-hand side is
\beq\label{b35}
\begin{array}{c}
  \displaystyle{
    \text{l.h.s.} = \dot{\mL}^{ii}(z)=\dot{p}_{i}1_{N}
    +\tr_{2}\left(\dot{\mS}^{ii}_{2}R^{(0),z}_{12}P_{12}\right)\,.
  }
\end{array}
\eq
The right-hand side is
\beq\label{b36}
\begin{array}{c}
  \displaystyle{
    \text{r.h.s.}= \mL^{i\bar{i}}(z)\mM^{\bar{i}i}(z)-\mM^{ii}(z)\mL^{ii}(z)+\mL^{ii}(z)
    \mM^{ii}(z)-\mM^{i\bar{i}}(z)\mL^{\bar{i}i}(z) +
  }
  \\ \ \\
  \displaystyle{
    + \sum_{k:k\neq i}^{n} \left(\mL^{ik}(z)\mM^{ki}(z)-\mM^{ik}(z)\mL^{ki}(z)\right)
    + \sum_{k:k\neq i}^{n}\left(\mL^{i\bar{k}}(z)\mM^{\bar{k}i}(z)-\mM^{i\bar{k}}(z)
    \mL^{\bar{k}i}(z)\right)\,.
  }
\end{array}
\eq
The commutator $[\mL^{ii}(z),\mM^{ii}(z)]$ in (\ref{b36}) gives the
commutator term in (\ref{t419}). This is the contribution from the
internal dynamics of the $ii$-th top, derived in \cite{LOZ161};
see also \cite{GSZ}.

We evaluate
$\mL^{i\bar{i}}(z)\mM^{\bar{i}i}(z)-\mM^{i\bar{i}}(z)\mL^{\bar{i}i}(z)$
in three parts. First, the terms involving only $\ti K$- and
$\ti Y$-matrices give
\beq\label{b37}
%
\eq
where we have used a degeneration of (\ref{t220}), obtained by
differentiating that identity with respect to $q$.

The terms involving neither $\ti K$ nor $\ti Y$ give
\beq\label{b41}
%
\eq
The remaining terms sum to
\beq\label{b42}
%
\eq
Finally, we evaluate the two sums in (\ref{b36}):
\beq\label{b43}
%
\eq
Terms of the form appearing at the end of (\ref{b43}) and (\ref{b44})
contribute to the momentum equations (\ref{t421}). Indeed,
\beq\label{b45}
%
\eq
These terms, together with those already independent of the spectral
parameter $z$, yield (\ref{t421}). The terms containing $R^{z,(0)}$
yield the equations of motion (\ref{t419}).

The equations for the remaining blocks,
$\dot{\mL}^{\bar{i}\bar{j}}(z)$ and $\dot{\mL}^{\bar{i}j}(z)$,
are derived in the same way as (\ref{t417})--(\ref{t421}), using the
same identities for $\ti R$, $R$ and $\ti K$ with appropriate changes
of signs or arguments.

Without the constraints
$\mS^{\bar{i}\bar{j}}=\mS^{ij},~\mS^{i\bar{j}}=\mS^{\bar{i}j}$,
the equations for the blocks $ii$ and ${\bar i}{\bar i}$ are incompatible.
These constraints also ensure that
$\dot{\mS}^{ij}=\dot{\mS}^{\bar{i}\bar{j}},~\dot{\mS}^{i\bar{j}}=\dot{\mS}^{\bar{i}j}$.
\end{proof}

\section{Appendix C: proof of Theorem \ref{th2}}\label{secC}
\def\theequation{C.\arabic{equation}}
\setcounter{equation}{0}

We prove the identities
\beq\label{t422}
\begin{array}{c}
  \displaystyle{
    H = \frac{1}{2}\,\sum_{i=1}^{n}p^{2}_{i} +\frac{1}{4}\,\sum_{i=1}^{2n}
    \tr_{23}\Big(m_{23}(0)\mS^{ii}_{2}\mS^{ii}_{3}\Big) + \frac{1}{4}\, \sum_{i\neq j}^{n}\tr_{23}\Big(F^{0}_{32}(q_{ij})P_{23}
    \mS^{ij}_{2}\mS^{ji}_{3}\Big) +
  }
  \\ \ \\
  \displaystyle{
    + \frac{1}{4}\, \sum_{i\neq j}^{n}\tr_{23}\Big(F^{0}_{32}(q_{ij})P_{23}\mS^{\bar{i}\bar{j}}_{2}
    \mS^{\bar{j}\bar{i}}_{3}\Big) +\frac{1}{2}\,\sum_{i\neq j}^{n}
    \tr_{23}\Big(\ti{F}^{0}_{32}(q_{ij}^{+})P_{23}\mS^{i\bar{j}}_{2}\mS^{\bar{j}i}_{3}\Big) +
  }
  \\ \ \\
  \displaystyle{
    +\frac{1}{2}\,\sum_{i=1}^{n}\tr_{23}\Big(\ti{F}^{0}_{32}(2q_{i})P_{23}\mS^{i\bar{i}}_{2}
    \mS^{\bar{i}i}_{3}\Big)-\frac{1}{2}\,\sum_{i=1}^{n}u_N(q_i|\nu) +
 \frac{1}{2}\,\sum_{i=1}^{n}\tr_{2}\Big(\mS^{i\bar{i}}_{2}\partial_{q_{i}}
    \ti{K}^{(0)}_{2}(q_{i}) \Big)
  }
\end{array}
\eq
and
\beq\label{t423}
\begin{array}{c}
  \displaystyle{
    \frac{1}{4N}\,\tr\Big(\mL^{2}(z)\Big)=
      }
  \\ \ \\
   \displaystyle{
   =H-\frac{\partial_{z}\widetilde{E}_{1}(z)}{2N}\,\Big[2\varsigma_N\sum_{i=1}^{n}
   \tr\Big(\mS^{i\bar{i}} \Big)+\tr\Big(\sum\limits_{i,j=1}^n \mS^{ij}\mS^{ji}\Big)
   +\tr\Big(\sum\limits_{i,{j}=1}^n \mS^{i{\bar j}}\mS^{{\bar j}i}\Big)\Big]+
  }
  \\ \ \\
   \displaystyle{
   +\frac{{\widetilde E}_1^2(z)-\widetilde\wp(z)}{2N^2}\,\sum
   \limits_{i=1}^n\Big(\tr\Big(\mS^{ii}\Big)\Big)^2+\frac{n}{2}\,w_N(z|\nu)\,.
   }
\end{array}
\eq
\begin{proof}
We evaluate the expression
\beq\label{c01}
\begin{array}{c}
  \displaystyle{
    \frac{1}{4N}\,\tr \Big(\mL^2(z)\Big)= \frac{1}{4N}\,\sum_{i=1}^{n}\left(\tr\Big((\mL^{ii}(z))^2\Big)
    +\tr\Big((\mL^{\bar i\bar i}(z))^2\Big)\right)+\frac{1}{2N}\,\sum_{i=1}^{n}
    \tr\Big(\mL^{i\bar i}(z)\mL^{\bar i i}(z)\Big) +
  }
  \\ \ \\
  \displaystyle{
    +\frac{1}{4N}\,\sum_{\substack{i\neq j}}^{n} \left(\tr\Big(\mL^{ij}(z)\mL^{ji}(z)\Big)
    +\tr\Big(\mL^{\bar i\bar j}(z)\mL^{\bar j\bar i}(z)\Big)\right)+\frac{1}{2N}\,
    \sum_{\substack{i\neq j}}^{n} \tr\Big(\mL^{i\bar j}(z)\mL^{\bar ji}(z)\Big)\,.
  }
\end{array}
\eq
The normalization differs from the factor $1/(2N)$ in the
$A_{n-1}$ model. The present Lax matrix contains the two momentum
blocks $p_i1_N$ and $-p_i1_N$. Their squares give $2Np_i^2$ in
$\tr \mL^2(z)$, and the factor $1/(4N)$ therefore produces the canonical
kinetic term $p_i^2/2$ and the equation $\dot q_i=p_i$.

We begin with the diagonal blocks. As in the $A_{n-1}$ case,
we use (\ref{c02}) and its limiting form (\ref{c03}).
We use the scalar partial-trace conditions
(\ref{t424})--(\ref{t4241}) assumed in Theorem \ref{th2}.
For the elliptic matrix (\ref{t239}), these are contained in the
identities below; see (\ref{t242})--(\ref{t243}) and (\ref{a172}):
\beq\label{c04}
\begin{array}{c}
  \displaystyle{
    \tr_{1}R_{12}^{q}(z)=\tr_{2}R_{12}^{q}(z) =\widetilde\phi(z,\frac{q}{N})1_N,
    \qquad\tr_{1}r_{12}(z)=\widetilde E_{1}(z)1_N\,,
    }
    \\ \ \\
  \displaystyle{
     \tr_{1}m_{12}(z)=\frac{{\widetilde E}^2_{1}(z)-{\widetilde \wp}(z)}{2N}1_N\,,
  }
\end{array}
\eq
Only the identities for $r_{12}(z)$ and $m_{12}(z)$ and their
derivatives are needed below. We set
$\widetilde E_{2}(z)1_N=-\tr_{3}F_{13}^{0}(z)
=-\partial_z\widetilde E_{1}(z)1_N$.
Substituting (\ref{c03}), we find that the two
commutator-type terms cancel after taking the trace over the first
tensor component. By (\ref{c04}), the terms containing $m_{12}(z)$
and $m_{13}(z)$ reduce to spectral functions multiplied by squares
of the fixed diagonal traces. The terms proportional to
$\widetilde E_2(z)\tr((\mA^{ii})^2)$ must be retained: individual
diagonal-block squares are not Casimirs of the full spin algebra.
They combine with the off-diagonal contributions into a global
quadratic Casimir only after all blocks have been summed. We obtain
\beq\label{c05}
\begin{array}{c}
  \displaystyle{
    \left(\tr\Big((\mL^{ii}(z))^2\Big) +\tr\Big((\mL^{\bar i\bar i}(z))^2\Big)\right) =
  }
  \\ \ \\
  \displaystyle{
    = 2Np^{2}_{i} + 2p_{i}~\tr_{12}\Big(r_{12}(z)\mS_{2}^{ii}\Big)+2p_{i}~
    \tr_{12}\Big(r_{12}(-z)\mS_{2}^{\bar{i}\bar{i}}\Big) +
  }
\end{array}
\eq
$$
\begin{array}{c}
  \displaystyle{
    + \tr_{123}\Big(r_{12}(z)r_{13}(z)\mS^{ii}_{3}\mS_{2}^{ii}\Big)
    +\tr_{123}\Big(r_{12}(-z)r_{13}(-z)\mS_{2}^{\bar{i}\bar{i}}\mS^{\bar{i}\bar{i}}_{3}\Big)
    \stackrel{(\ref{c03}),(\ref{c04})}{=}
  }
  \\ \ \\
  \displaystyle{
    =2Np_{i}^{2}+N\tr_{12}\Big(m_{12}(0)\mS^{ii}_{1}\mS^{ii}_{2}\Big)+N\tr_{12}\Big(m_{12}(0)
    \mS^{\bar{i}\bar{i}}_{1}\mS^{\bar{i}\bar{i}}_{2}\Big) +
  }
  \\ \ \\
  \displaystyle{
    +\widetilde{E}_{2}(z)\tr\Big(\mS^{ii}\mS^{ii}\Big)+\widetilde{E}_{2}(z)
    \tr\Big(\mS^{\bar{i}\bar{i}}\mS^{\bar{i}\bar{i}}\Big)
    +\frac1N\Big({\widetilde E}^2_{1}(z)-{\widetilde \wp}(z)\Big)\Big[
    \Big(\tr\Big(\mS^{ii}\Big)\Big)^2
    +\Big(\tr\Big(\mS^{{\bar i}{\bar i}}\Big)\Big)^2 \Big]\,.
  }
\end{array}
$$
The terms linear in $p_i$ cancel between the $ii$ and
$\bar i\bar i$ blocks because of the constraints
(\ref{t4012}) and the oddness of $\widetilde E_1(z)$.

Next, consider
\beq\label{c06}
\begin{array}{c}
  \displaystyle{
    \tr\Big(\mL^{ij}(z)\mL^{ji}(z)\Big)+\tr\Big(\mL^{\bar{i}\bar{j}}(z)
    \mL^{\bar{j}\bar{i}}(z)\Big) =
  }
  \\ \ \\
  \displaystyle{
    =\tr_{123}\Big(\Big(R^{z}_{12}(q_{ij})R^{z}_{23}(q_{ji})
    +R^{-z}_{12}(q_{ij})R^{-z}_{23}(q_{ji})\Big)P_{12}P_{13}\mS^{ij}_{2}\mS^{ji}_{3}\Big) =
  }
\end{array}
\eq
$$
\begin{array}{c}
  \displaystyle{
    = \tr_{123}\Big(\Big([r_{13}(z),r_{32}(q_{ij})]-[r_{31}(z),r_{32}(q_{ij})] -
F^{0}_{13}(z)-F^{0}_{31}(z)+2F^{0}_{32}(q_{ij})\Big)P_{23}\mS^{ij}_{2}\mS^{ji}_{3}\Big)
    \stackrel{(\ref{c04})}{=}
  }
  \\ \ \\
  \displaystyle{
    =2N\tr_{12}\Big(F^{0}_{21}(q_{ij})P_{12}\mS_{1}^{ij}\mS_{2}^{ji}\Big)+2\widetilde{E}_{2}(z)
    \tr\Big(\mS^{ij}\mS^{ji}\Big)\,.
  }
\end{array}
$$
The last term in (\ref{c01}) takes the form
\beq\label{c09}
\begin{array}{c}
  \displaystyle{
    \tr\Big(\mL^{i\bar j}(z)\mL^{\bar j i}(z)\Big) =-\tr_{123}\Big(\ti{R}^{z}_{12}(q_{ij}^{+})
    \ti{R}^{-z}_{23}(q_{ij}^{+})P_{12}P_{13}\mS^{i\bar{j}}_{2}\mS^{\bar{j}i}_{3}\Big)
    \stackrel{(\ref{c08})}{=}
  }
  \\ \ \\
  \displaystyle{
    =-\tr_{123}\Big(\ti{r}_{23}(q_{ij}^{+})r_{13}(z)P_{23}\mS^{i\bar{j}}_{2}\mS^{\bar{j}i}_{3}\Big)
    +\tr_{123}\Big(r_{13}(z)\ti{r}_{32}(q_{ij}^{+})P_{23}\mS^{i\bar{j}}_{2}\mS^{\bar{j}i}_{3}\Big) +
  }
  \\ \ \\
  \displaystyle{
    +\tr_{123}\Big(\ti{F}^{0}_{32}(q_{ij}^{+})P_{23}\mS^{i\bar{j}}_{2}\mS^{\bar{j}i}_{3}\Big)
    + \widetilde{E}_{2}(z)\tr\Big(\mS^{i\bar{j}}\mS^{\bar{j}i}\Big)\,.
  }
\end{array}
\eq
After taking the trace over the first component, the first two terms
cancel.

It remains to evaluate
$\tr\Big(\mL^{i\bar{i}}(z)\mL^{\bar{i}i}(z)\Big)$.
In addition to the unitarity conditions (\ref{t219}) and (\ref{t220}),
we use the degeneration (\ref{c10}) of the four-term relation.

We divide the calculation into three parts. First, consider the
contribution to $\tr\Big(\mL^{i\bar{i}}(z)\mL^{\bar{i}i}(z)\Big)$
that does not involve $\ti K$:
\beq\label{c11}
\begin{array}{c}
  \displaystyle{
    \tr\Big((\mL^{i\bar{i}}(z)- \ti{K}^{z}_{1}(q_{i}))(\mL^{\bar i i}(z)+\ti{K}^{-z}_{1}(q_{i}))\Big)
    =-\tr_{123}\Big(\ti{R}^{z}_{12}(2q_{i})\ti{R}^{-z}_{23}(2q_{i})P_{12}P_{13}
    \mS^{i\bar{i}}_{2}\mS^{\bar{i}i}_{3}\Big) \stackrel{(\ref{c08})}{=}
  }
  \\ \ \\
  \displaystyle{
    =-\tr_{123}\Big(\ti{r}_{23}(2q_{i})r_{13}(z)P_{23}\mS^{i\bar{i}}_{2}\mS^{\bar{i}i}_{3}\Big)
    +\tr_{123}\Big(r_{13}(z)\ti{r}_{32}(2q_{i})P_{23}\mS^{i\bar{i}}_{2}\mS^{\bar{i}i}_{3}\Big)+
  }
  \\ \ \\
  \displaystyle{
    +\tr_{123}\Big(\ti{F}^{0}_{32}(2q_{i})P_{23}\mS^{i\bar{i}}_{2}\mS^{\bar{i}i}_{3}\Big)
    +\tr\Big(\mS^{i\bar{i}}\mS^{\bar{i}i}\Big){\widetilde E}_2(z)\,.
  }
\end{array}
\eq
The terms in the first line cancel. The terms quadratic in
$\ti K$-matrices give
\beq\label{c12}
\begin{array}{c}
  \displaystyle{
    - \tr\Big( \tilde{K}_{1}^{z}(q_{i}) \tilde{K}_{1}^{-z}(q_{i})\Big)\stackrel{(\ref{t220})}{=}
    -N(u_N(q_i|\nu)-w_N(z|\nu))\,.
  }
\end{array}
\eq
The last term depends only on the spectral parameter and the
fixed coupling constants.

Finally, the remaining terms in
$\tr\Big(\mL^{i\bar{i}}(z)\mL^{\bar{i}i}(z)\Big)$ give
\beq\label{c13}
\begin{array}{c}
  \displaystyle{
    \tr\Big(\ti{K}^{z}_{1}(q_{i})(\mL^{\bar i i}(z)+\ti{K}^{-z}_{1}(q_{i}))\Big)
    -\tr\Big((\mL^{i\bar{i}}(z)-\ti{K}^{z}_{1}(q_{i}))\ti{K}^{-z}_{1}(q_{i})\Big) =
  }
  \\ \ \\
  \displaystyle{
    = -\tr_{12}\Big(\mS^{i\bar{i}}_{2}\ti{R}^{z}_{12}(2q_{i})\ti{K}^{-z}_{2}(q_{i})P_{12}\Big)
    -\tr_{12}\Big(\mS^{\bar{i}i}_{2}\ti{K}^{z}_{1}(q_{i})\ti{R}^{-z}_{12}(2q_{i})P_{12}\Big)
    \stackrel{(\ref{t4012}),(\ref{c10})}{=}
  }
\end{array}
\eq
$$
\begin{array}{c}
  \displaystyle{
    = -\tr_{12}\Big(\mS^{i\bar{i}}_{2}\left(r_{21}(-z)\ti{K}^{(0)}_{2}(q_{i})
    -\ti{K}^{(0)}_{2}(q_{i})r_{21}(-z)\right)\Big) -
  }
  \\ \ \\
  \displaystyle{
    -\tr_{12}\Big(\mS^{i\bar{i}}_{2}\left(\partial_{-z}r_{21}(-z)\ti{K}^{(-1)}_{2}
    +\ti{K}^{(-1)}_{2}\partial_{-z}r_{21}(-z)-\partial_{q_{i}}
    \ti{K}^{(0)}_{2}(q_{i})\right)\Big) \stackrel{(\ref{c04})}{=}
  }
  \\ \ \\
  \displaystyle{
    = - 2\partial_{z}\widetilde{E}_{1}(z)\tr_{2}\Big(\mS^{i\bar{i}}_{2}
    \tilde{K}^{(-1)}_{2}\Big)+ N\tr_{2}\Big(\mS^{i\bar{i}}_{2}\partial_{q_{i}}
    \tilde{K}^{(0)}_{2}(q_{i})\Big)\,.
  }
\end{array}
$$
After summing the blocks with the weights in (\ref{c01}), the
quadratic terms proportional to $\widetilde E_2(z)$ in
(\ref{c05}), (\ref{c06}), (\ref{c09}) and (\ref{c11}) combine into
$\tr(\mA^2+\mB^2)$, the Casimir in (\ref{t429}). Together with
(\ref{c13}), their contribution is
\[
 \frac{\widetilde E_2(z)}{2N}
 \left(\tr(\mA^2+\mB^2)+2\varsigma_N\tr\mB\right),
\]
where $\tr\mB$ is a Casimir by (\ref{t438}). The remaining diagonal
trace squares are constant by the moment constraints (\ref{t428});
these constraints are not themselves Casimirs of the full spin
algebra. Combining (\ref{c05}), (\ref{c06}), (\ref{c09}),
(\ref{c11}), (\ref{c12}) and (\ref{c13}) yields the Hamiltonian
(\ref{t422}).
\end{proof}

\subsection*{Acknowledgments}
%
%
The work of A. Zotov was performed at the Steklov International Mathematical
Center and supported by the Ministry of Science and Higher Education
of the Russian Federation (agreement no. 075-15-2025-303).

\noindent The work of D. Domanevsky was supported by the Foundation for the Advancement of Theoretical Physics and Mathematics ``BASIS".


\begin{footnotesize}

\end{footnotesize}

\end{document}